\documentclass[a4paper,12pt,twoside]{article}

\usepackage{amsfonts}
\usepackage{amsmath}
\usepackage{amsopn}
\usepackage{amsthm}
\usepackage{amssymb}
\usepackage{graphicx}
\usepackage{rotating}
\usepackage[]{natbib}
\usepackage[english]{babel}
\usepackage[utf8]{inputenc}
\usepackage[colorlinks=true,citecolor=blue,linkcolor=blue]{hyperref}
\usepackage{xspace}
\usepackage{float}
\usepackage{setspace}
\usepackage{dsfont}
\usepackage{booktabs}
\usepackage[T1]{fontenc}
\usepackage{anysize}
\usepackage{wrapfig}
\usepackage[labelfont=bf,labelsep=period]{caption}
\floatstyle{plaintop}
\restylefloat{table}
\marginsize{1.87cm}{1.87cm}{1.87cm}{1.87cm} % left, right, top , bottom
\usepackage{fancyhdr}
\usepackage{algorithm}% http://ctan.org/pkg/algorithms
\usepackage[noend]{algpseudocode}% http://ctan.org/pkg/algorithmicx

\RequirePackage{txfonts}
\usepackage{times}

\newtheorem{Proposition}{Proposition}

\newtheorem{Corollary}{Corollary}
\newtheorem{Theorem}{Theorem}
\newtheorem{Lemma}{Lemma}
\newtheorem{Remark}{Remark}
\newtheorem{Assumption}{Assumption}

\DeclareMathOperator{\re}{\mathbb{R}}
\DeclareMathOperator{\na}{\mathbb{N}}
\newcommand{\abs}[1]{\left|#1\right|}
\newcommand{\E}{\mathbb{E}}
\newcommand{\ind}{\mathds{1}}

\newcommand{\Prob}{\mathbb{P}}

\allowdisplaybreaks

\begin{document}

\def\figureautorefname{Figure}
\def\algorithmautorefname{Algorithm}
\def\sectionautorefname{Section}
\def\subsectionautorefname{Section}
\def\subsubsectionautorefname{Section}
\def\Propositionautorefname{Proposition}
\def\Theoremautorefname{Theorem}
\def\Corollaryautorefname{Corollary}
\def\Lemmaautorefname{Lemma}
\def\Assumptionautorefname{Assumption}
\renewcommand*\footnoterule{}

\title{Non-reversible jump algorithms for Bayesian nested model selection}

\author{Philippe Gagnon $^{1}$, Arnaud Doucet $^{1}$}

\maketitle

\thispagestyle{empty}

\noindent $^{1}$Department of Statistics, University of Oxford, United Kingdom.

\begin{abstract}
Non-reversible Markov chain Monte Carlo methods often outperform their reversible counterparts in terms of asymptotic variance of ergodic averages and mixing properties. Lifting the state-space \citep{chen1999lifting,diaconis2000analysis} is a generic technique for constructing such samplers. The idea is to think of the random variables we want to generate as position variables and to associate to them direction variables so as to design Markov chains which do not have the diffusive behaviour often exhibited by reversible schemes. In this paper, we explore the benefits of using such ideas in the context of Bayesian model choice for nested models, a class of models for which the model indicator variable is an ordinal random variable. By lifting this model indicator variable, we obtain \textit{non-reversible jump algorithms}, a non-reversible version of the popular reversible jump algorithms introduced by \cite{green1995reversible}. This simple algorithmic modification provides samplers which can empirically outperform their reversible counterparts at no extra computational cost. The code to reproduce all experiments is available online.\footnote{See ancillary files on \href{https://arxiv.org/abs/1911.01340}{arXiv:1911.01340}.}
\end{abstract}


\section{Introduction}\label{sec_intro}

%\subsection{Reversible jump algorithms}\label{sec_RJ}

Reversible jump (RJ) algorithms are a popular class of Markov chain Monte Carlo (MCMC) methods introduced by \cite{green1995reversible,green2003trans}. They are used to sample from a target distribution $\pi$ defined on $\bigcup_{j\in \mathcal{K}} \{j\}\times  \re^{d_j}$, $\mathcal{K}$ being a countable set. In the statistics applications discussed in this paper, this distribution corresponds to the joint posterior distribution of a model indicator $k\in \mathcal{K}$ and its corresponding parameters $\mathbf{x}_k\in\re^{d_k}$. These samplers thus allow us to perform simultaneously model selection and parameter estimation. In the following, we assume for simplicity that the parameters of all models are continuous random variables and abuse notation by also using $\pi$ to denote the target density.

Given the current state $(k, \mathbf{x}_k)$, a RJ algorithm generates the next state by proposing a model candidate $k'$ from some probability mass function (PMF) $g(k,\cdot \,)$ then a proposal for its corresponding parameter values. This last step is usually achieved through two sub-steps:
\begin{enumerate}
\itemsep 0mm
 \item generate $\mathbf{u}_{k\mapsto k'}\sim q_{k\mapsto k'}$ (this vector corresponds to auxiliary variables used, for instance, to propose values for additional parameters when $d_{k'}>d_k$), where $q_{k\mapsto k'}$ is a probability density function (PDF),
 \item apply the function $T_{k\mapsto k'}$ to $(\mathbf{x}_k,\mathbf{u}_{k\mapsto k'})$, $T_{k\mapsto k'}(\mathbf{x}_k,\mathbf{u}_{k\mapsto k'})=:(\mathbf{y}_{k'},\mathbf{u}_{k'\mapsto k})$, where the vector $\mathbf{y}_{k'}$ represents the proposal for the parameters of model $k'$ and $T_{k\mapsto k'}$ is a diffeomorphism (i.e.\ a differentiable map having a differentiable inverse).
\end{enumerate}
The notation $k\mapsto k'$ in subscript is used to highlight a dependance on both the current and proposed models. When $k' = k$, we say that a \textit{parameter update} is proposed, whereas we say that a \textit{model switch} is proposed when $k' \neq k$. The proposal $(k',\mathbf{y}_{k'})$ is accepted with probability:
\begin{align}\label{eqn_acc_prob_RJ}
    \alpha_{\text{RJ}}((k,\mathbf{x}_{k}),(k',\mathbf{y}_{k'})):=1\wedge \frac{g(k',k) \, \pi(k',\mathbf{y}_{k'}) \, q_{k'\mapsto k}(\mathbf{u}_{k'\mapsto k})}{g(k,k') \, \pi(k,\mathbf{x}_{k}) \, q_{k\mapsto k'}(\mathbf{u}_{k\mapsto k'}) \, |J_{T_{k\mapsto k'}}(\mathbf{x}_{k}, \mathbf{u}_{k\mapsto k'})|^{-1}},
\end{align}
    where $x \wedge y := \min(x, y)$ and $|J_{T_{k\mapsto k'}}(\mathbf{x}_{k}, \mathbf{u}_{k\mapsto k'})|$ is the absolute value of the determinant of the Jacobian matrix of the function $T_{k\mapsto k'}$. If the proposal is rejected, the chain remains at the same state $(k,\mathbf{x}_{k})$.

%Looping over the steps described above gives rise to Markov chains that are reversible with respect to the target distribution. If in addition the chains are irreducible and aperiodic, they are therefore ergodic.

% It is often the case that either $\mathbf{u}_{k\mapsto k'}$ or $\mathbf{u}_{k'\mapsto k}$ does not actually exist, which can be viewed as setting the former or the latter to $\mathbf{0}$. Indeed, consider the case $d_{k'}>d_k$, meaning that model $k'$ has more parameters than model $k$ and $d^{k'} - d^k$ is the number of additional parameters in model $k'$. We often generate $d_{k'} - d_k$ random variables in Step 1 above and set $T_{k\mapsto k'}(\mathbf{x}_k,\mathbf{u}_{k\mapsto k'})=:\mathbf{y}_{k'}$; in this case, $\mathbf{u}_{k'\mapsto k}$ does not exist. Note that this scheme also defines the reverse move $ T_{k\mapsto k'}^{-1}(\mathbf{y}_{k'})=(\mathbf{x}_k,\mathbf{u}_{k\mapsto k'})=:T_{k'\mapsto k}(\mathbf{y}_{k'})$, which is deterministic in this case.

In this paper, we consider the special case of nested models; i.e.\ $K$ is an ordinal discrete random variable that reflects the complexity of the models. For instance, it represents the number of change-points in multiple change-point problems (Section 4, \cite{green1995reversible}), the number of components in mixture modelling \citep{richardson1997bayesian}, the order of an autoregressive process \citep{vermaak2004reversible}, the number of clusters in dependence structures for multivariate extremes \citep{vettori2019bayesian} or the number of principal components included in robust principal component regression \citep{gagnon2020automatic}. We restrict our attention to samplers that switch models by taking steps of $\pm 1$, i.e.\ $k' \in \{k - 1, k + 1\}$ when a model switch is proposed. This is a common choice which implies that the model space $\mathcal{K}$ is explored through a random walk, a process that often backtracks and thus exhibits a diffusive behaviour. This choice of neighbourhood $k' \in \{k - 1, k + 1\}$ in RJ makes the process reversible with respect to $\pi$ and thus ensures that $\pi$ is an invariant distribution.

The objective of this paper is to propose sampling schemes which do not suffer from such a diffusive behaviour by exploiting the lifting idea introduced by \cite{chen1999lifting} and \cite{diaconis2000analysis} to induce persistent movement in the model indicator. In the somewhat related contexts of simulated tempering \citep{sakai2016irreversible} and parallel tempering \citep{syed2019non}, lifting the temperature variable provides non-reversible samplers which perform substantially better than their reversible counterparts.

The changes that we make to the RJ sampling framework described above to apply the lifting idea are remarkably simple and require no additional computational effort. First, we extend the state-space by adding a direction variable $\nu \in \{-1, 1\}$ and assign it a uniform distribution $\mathcal{U}\{-1, 1\}$. Second, when a model switch is proposed and the current state is $(k, \mathbf{x}_k, \nu)$, the model to explore is selected deterministically instead of randomly by setting: $k' := k+\nu$. If the proposal for the model to explore next (model $k'$) along with its parameter values $\mathbf{y}_{k'}$ is accepted, the next state of the chain is $(k',\mathbf{y}_{k'}, \nu)$. The direction for the model indicator remains the same; this is what induces persistent movement. If the proposal is rejected, the next state of the chain is $(k,\mathbf{x}_{k}, -\nu)$, so the direction is reversed for $K$. A proposal may be rejected because there is negligible mass beyond $k$ in the direction followed; a change in direction may thus imply a return towards the high probability area. % A general non-reversible jump (NRJ) algorithm is described in detail in \autoref{algo_NRJ} in Section \ref{sec_NRJ}.
Such simple modifications lead to a non-reversible scheme and can be very efficient as illustrated in \autoref{fig_traces} (in which ESS stands for effective sample size). ESS per iteration is defined as the inverse of the integrated autocorrelation time. In this paper, it is used to evaluate algorithms regarding how efficient they are at sampling $K$. In particular, it measures their capacity of making the stochastic process $\{K(m): m \in \na\}$  traverse the model space, which is what we want to highlight. % For these purposes, only iterations in which model switches are proposed are retained for its computation and ESS is always reported per iteration for fairness. This is valid when the probability of proposing model switches do not depend on the current state and in fact makes our efficiency measure independent of this probability.

 \begin{figure}[ht]
  \centering
  $\begin{array}{cc}
  \vspace{-1mm}\hspace{4mm}\textbf{Random walk behaviour} & \hspace{5mm}\textbf{Persistent movement} \cr
   \vspace{-1mm}\hspace{4mm}\textbf{ESS = 0.09 per it.} & \hspace{5mm}\textbf{ESS = 0.35 per it.} \cr
 \includegraphics[width=0.40\textwidth]{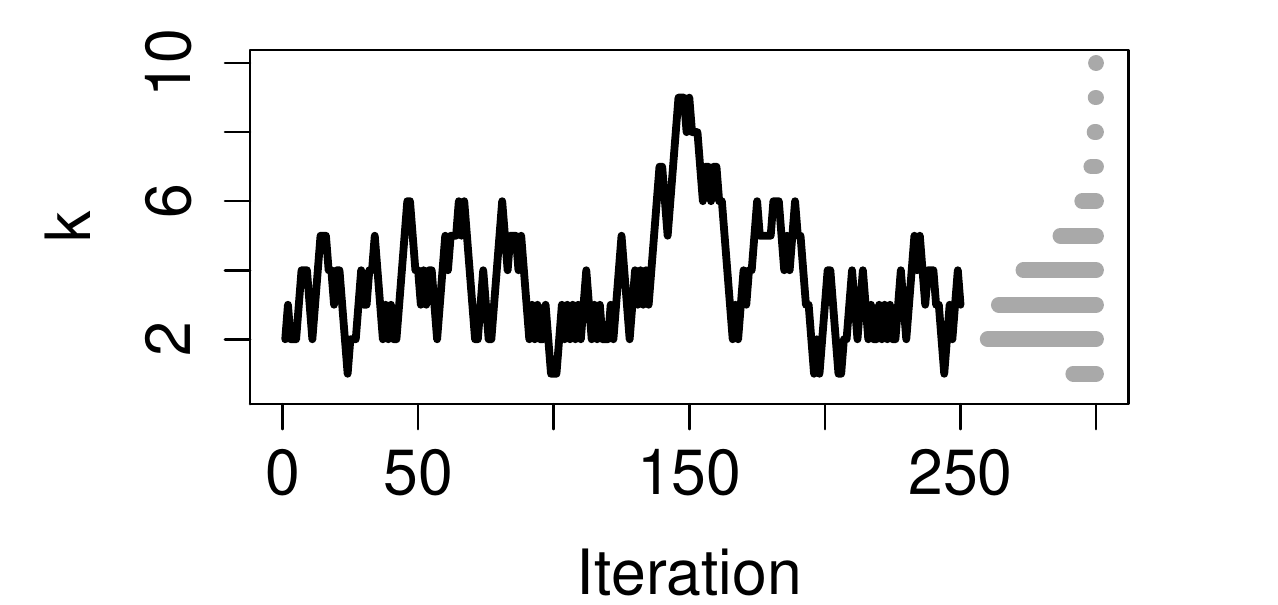} &  \includegraphics[width=0.40\textwidth]{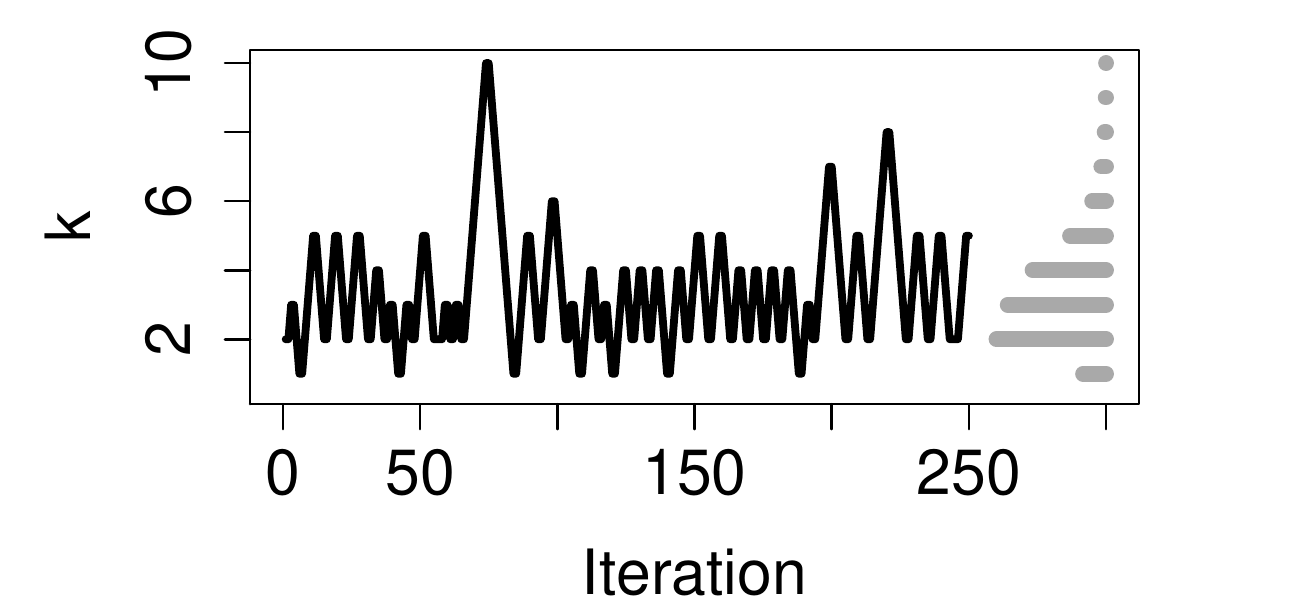}
  \end{array}$\vspace{-4mm}
  \caption{\small Trace plots for \textit{ideal} RJ and NRJ (we define what we mean by ideal in Section~\ref{sec_ideal}), and showing only the iterations in which model switches are proposed; the horizontal lines represent the marginal targeted PMF $\pi(k)$ in a real multiple change-point problem presented in \autoref{sec_changepoint}}\label{fig_traces}
 \end{figure}
\normalsize

% \subsection{Organisation of the paper}\label{sec_organisation}

The rest of this paper is organised as follows. We first introduce in \autoref{sec_NRJ_and_ideal} a general non-reversible jump (NRJ) algorithm
% We first introduce in \autoref{sec_NRJ_and_ideal} a general NRJ scheme
and establish its validity. We also present its \textit{ideal} version that is able to propose model parameter values from the conditional distributions $\pi(\, \cdot \mid k)$. This ideal algorithm is simple and allows us to exploit existing theoretical results to establish in \autoref{sec_optimality} that NRJ can outperform the corresponding ideal RJ under some assumptions on the marginal PMF $\pi(k)$. Although such an ideal sampler cannot be implemented in practice for complex models, we show in \autoref{sec_towards_ideal} how we can leverage methods that have been previously developed in the RJ literature to approximate this ideal NRJ sampler. The weak convergence of the resulting sampler towards the ideal NRJ sampler is established as a precision parameter increases without bounds. % ***We next provide theoretical evidences in \autoref{sec_optimality} that ideal NRJ outperform the optimal ideal RJ (constructed from a specific skewed proposal distribution $g$) when $\pi(k)$ is unimodal and not too concentrated. This makes NRJ optimal for this family of distributions when samplers are restricted to model switching proposals of the form $k\mapsto k'\in\{k-1, k+1\}$.****
In \autoref{sec_optimality}, we also prove that any NRJ (ideal or non-ideal) performs at least as good as its reversible counterpart. We present in \autoref{sec_numerical} numerical experiments to illustrate the performance of NRJ samplers on a toy example % for which the ideal sampler can be implemented
and a real multiple change-point problem. We provide a  discussion of implementation aspects and possible extensions in \autoref{sec_discussion}. All proofs of theoretical results are provided in \autoref{sec_proofs}.

\section{Non-reversible jump algorithms and ideal samplers}\label{sec_NRJ_and_ideal}

\subsection{Non-reversible jump schemes}\label{sec_NRJ}

\autoref{algo_NRJ} presents the general NRJ which takes as inputs an initial state $(k,\mathbf{x}_k,\nu)$, a total number of iterations, the functions $q_{k\mapsto k'}$ and $T_{k\mapsto k'}$, and $0\leq \tau \leq 1$ which represents the probability of proposing a parameter update at any given iteration. In trans-dimensional samplers, the probability of proposing a parameter update is typically allowed to depend on the current state. For ease of presentation, it is considered constant here.

\begin{algorithm}[ht]
\caption{NRJ} \label{algo_NRJ}
 \begin{enumerate}
 \itemsep 0mm

  \item Sample  $u_c \sim \mathcal{U}(0, 1)$.

  \item[2.(a)\hspace{-4.5mm}] \hspace{4mm} If $u_c \leq \tau$, attempt a parameter update using a MCMC kernel of invariant distribution $\pi(\, \cdot \mid k)$ while keeping the values of the model indicator $k$ and direction $\nu$ fixed.

  \item[2.(b)\hspace{-4.5mm}] \hspace{4mm} If $u_c > \tau$, attempt a model switch from model $k$ to model $k' = k + \nu$. Sample $\mathbf{u}_{k \mapsto k'} \sim q_{k \mapsto k'}$ and $u_a\sim\mathcal{U}(0, 1)$, and compute $(\mathbf{y}_{k'},\mathbf{u}_{k'\mapsto k}) = T_{k\mapsto k'}(\mathbf{x}_k,\mathbf{u}_{k\mapsto k'})$. If
  \begin{align}\label{eqn_acc_prob_NRJ}
  	\hspace{-3mm} u_a \leq \alpha_{\text{NRJ}}((k,\mathbf{x}_{k}),(k',\mathbf{y}_{k'})):=1\wedge \frac{\pi(k',\mathbf{y}_{k'}) \, q_{k'\mapsto k}(\mathbf{u}_{k'\mapsto k})}{\pi(k,\mathbf{x}_{k}) \, q_{k\mapsto k'}(\mathbf{u}_{k\mapsto k'}) \, |J_{T_{k\mapsto k'}}(\mathbf{x}_{k}, \mathbf{u}_{k\mapsto k'})|^{-1}},
  \end{align}
  set the next state of the chain to $(k',\mathbf{y}_{k'}, \nu)$. Otherwise, set it to $ (k,\mathbf{x}_{k}, -\nu)$.

  \item[3.] Go to Step 1.
 \end{enumerate}
\end{algorithm}

\autoref{prop_invariance} below ensures that \autoref{algo_NRJ} targets the correct distribution. Note that $\mathbf{x}_k$ can be vectors containing both position and velocity variables, which allows using Hamiltonian Monte Carlo (HMC, see, e.g., \cite{neal2011mcmc}) and more generally discrete-time piecewise-deterministic MCMC schemes \citep{vanetti2017piecewise} for updating the parameters within \autoref{algo_NRJ}. The only prerequisite is that the method leaves the conditional distributions $\pi(\,\cdot\mid k)$ invariant. The proof of \autoref{prop_invariance} establishes that any valid scheme used for parameter proposals during model switches in RJ framework, such as those of  \cite{karagiannis2013annealed} and \cite{andrieu2018utility} presented in \autoref{sec_towards_ideal}, are also valid in the non-reversible framework.

\begin{Proposition}[Invariance] \label{prop_invariance}
 The transition kernel of the Markov chain $\{(K, \mathbf{X}_K, \nu)(m): m\in\na\}$ simulated by \autoref{algo_NRJ} admits $\pi\otimes \mathcal{U}\{-1, 1\}$ as invariant distribution.
\end{Proposition}

Nothing prevents \autoref{algo_NRJ} from switching to models at a distance of more than 1, i.e.\ with $|k' - k| > 1$. In Step 2.(b), an additional random variable $\omega \in \{0, 1, \ldots\}$ can be independently generated from, for instance, a Poisson distribution with a given mean parameter. In this case, we attempt to make a transition to model $k' = k + \omega \nu$,  but nothing else changes and the algorithm is still valid. In practice, however, $|d_{k'}-d_k|$ typically increases with $|k'-k|$, requiring to design proposal distributions of high dimensions, which is often very difficult and motivates using jumps to models no further than $k\pm 1$ as in \cite{green1995reversible}, \cite{richardson1997bayesian}, \cite{vermaak2004reversible}, \cite{vettori2019bayesian} and \cite{gagnon2020automatic}.

\subsection{Ideal samplers and their advantages}\label{sec_ideal}

When switching models, one would ideally be able to sample from the correct conditional distributions $\pi(\,\cdot\mid k')$ to propose parameter values $\mathbf{y}_{k'}$. In this ideal situation, one can set $q_{k\mapsto k'}:= \pi(\,\cdot\mid k')$, $q_{k'\mapsto k}:= \pi(\,\cdot\mid k)$, and $T_{k\mapsto k'}$ such that $\mathbf{y}_{k'}:=\mathbf{u}_{k\mapsto k'}$ (which implies that $\mathbf{u}_{k'\mapsto k}:=\mathbf{x}_{k}$), and observe that the acceptance probabilities reduce to
 \begin{align}\label{eqn_acc_ideal}
    % \alpha_{\text{RJ}}((k,\mathbf{x}_{k}),(k',\mathbf{y}_{k'}))=1\wedge \frac{g(k',k) \, \pi(k')}{g(k,k') \, \pi(k)} \quad \text{and} \quad
    \alpha_{\text{NRJ}}((k,\mathbf{x}_{k}),(k',\mathbf{y}_{k'}))=1\wedge \frac{\pi(k')}{\pi(k)}.
  \end{align}
 These probabilities are independent of the current and proposed parameters values: a model proposal $k'$ is accepted solely on the basis of the ratio of marginal posterior probabilities.

   In general, the acceptance probabilities are as above whenever
\begin{align}\label{eqn_ass_toy}
 \frac{\pi(\mathbf{y}_{k'}\mid k') \, q_{k'\mapsto k}(\mathbf{u}_{k'\mapsto k})}{\pi(\mathbf{x}_{k}\mid k) \, q_{k\mapsto k'}(\mathbf{u}_{k\mapsto k'}) \, |J_{T_{k\mapsto k'}}(\mathbf{x}_{k}, \mathbf{u}_{k\mapsto k'})|^{-1}} = 1,
\end{align}
for any switch from model $k$ with parameter values $\mathbf{x}_{k}$ to model $k'\neq k$ with parameter values $\mathbf{y}_{k'}$, using the auxiliary variables $\mathbf{u}_{k\mapsto k'}$ and $\mathbf{u}_{k'\mapsto k}$. In this more general setting, we observe that \eqref{eqn_ass_toy} is verified if, starting with random variables distributed as $\pi(\, \cdot \mid k)\otimes q_{k\mapsto k'}$ and applying the function $T_{k\mapsto k'}$, we obtain random variables distributed as $\pi(\,\cdot \mid k') \otimes q_{k'\mapsto k}$.

For realistic scenarios where the proposals $\mathbf{y}_{k'}$ are not generated from the conditional distributions $\pi(\,\cdot\mid k')$, the acceptance probability of model switches for NRJ \eqref{eqn_acc_prob_NRJ} can be expressed as a ``noisy'' version of that in \eqref{eqn_acc_ideal}:
\begin{align}\label{eqn_acc_noisy}
    \alpha_{\text{NRJ}}((k,\mathbf{x}_{k}),(k',\mathbf{y}_{k'}))=1\wedge \frac{\pi(k')}{\pi(k)} \, \varepsilon(k, k', \mathbf{x}_k, q_{k\mapsto k'}, q_{k'\mapsto k}, T_{k\mapsto k'}),
\end{align}
where $\varepsilon$ represents multiplicative noise given by the left-hand side (LHS) in \eqref{eqn_ass_toy}.

For NRJ to be beneficial, it will be useful to have a low variance noise. Imagine that the targeted marginal PMF is that on the right of \autoref{fig_noisy_NRJ} and that the samplers are initialised at $(K, \nu)(0):=(10, -1)$ (as in \autoref{fig_noisy_NRJ}), the advantage of the ideal NRJ is that it continues following the direction $-1$ for several iterations as the ratios $\pi(k - 1)/\pi(k)$ are greater than 1 (which implies that the proposals are accepted). If the noise fluctuations are significant, such moves might be rejected.
% These interruptions are due to rejections caused by generating variables $(\mathbf{y}_{k'}, \mathbf{u}_{k'\mapsto k})$ that are not distributed as $\pi(\,\cdot\,\mid k')\otimes q_{k'\mapsto k}$ and have small densities under this distribution (and small $\varepsilon$). Interruptions cause changes in direction (and diffusive behaviour). Rejections and diffusive behaviour both slow down the state-space exploration and decrease the accuracy of the approximations to the posterior probabilities \citep{peskun1973optimum, tierney1998note, neal2004improving}.

 \begin{figure}[ht]
  \centering
  \includegraphics[width=0.5\textwidth]{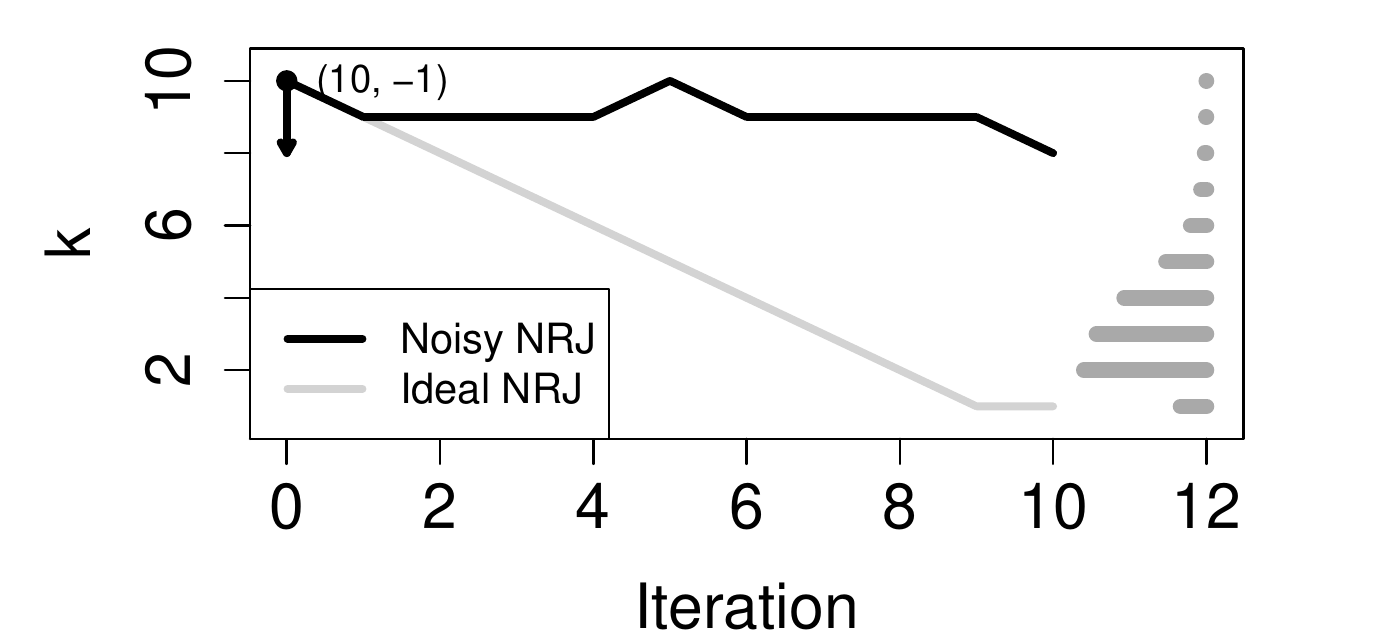}
  \vspace{-2mm}
  \caption{\small Trace plots for ``noisy'' and ideal NRJ, and showing only the iterations in which model switches are proposed; the horizontal lines represent the marginal targeted PMF $\pi(k)$ which is that in a real multiple change-point problem presented in \autoref{sec_changepoint}}\label{fig_noisy_NRJ}
 \end{figure}
\normalsize

\section{Towards ideal NRJ}\label{sec_towards_ideal}

We explained in the last section why it may be important to implement NRJ samplers that are close to their ideal counterparts, with low variance noise $\varepsilon$. We present in this section methods to achieve this by adapting some developed within the RJ framework. In \autoref{sec_andrieu_2013}, we present and adapt for NRJ the method of \cite{karagiannis2013annealed}. We proceed similarly in \autoref{sec_andrieu_2018} with the approach of \cite{andrieu2018utility}. In \autoref{sec_conv_algo2}, we prove that, as $\varepsilon \longrightarrow 1$ in distribution, the Markov chains produced by NRJ incorporating these approaches converge weakly to the chains produced by ideal NRJ.

\subsection{NRJ with the method of \cite{karagiannis2013annealed}}\label{sec_andrieu_2013}

Model $k$ and model $k'$ may be quite different. Jumping (``in one step'') from the former with parameters $\mathbf{x}_k$ to the latter with parameters $\mathbf{y}_{k'}$ may thus be difficult; i.e. starting with $(\mathbf{x}_k, \mathbf{u}_{k \mapsto k'}) \sim \pi(\,\cdot \mid k)\otimes q_{k\mapsto k'}$, it may be difficult to design a function $T_{k\mapsto k'}$ such that $(\mathbf{y}_{k'}, \mathbf{u}_{k' \mapsto k}) \sim \pi(\,\cdot \mid k')\otimes q_{k'\mapsto k}$ approximately, for any reasonable choice of $q_{k\mapsto k'}$ and $q_{k'\mapsto k}$.

\cite{karagiannis2013annealed} introduce a sequence of auxiliary distributions playing the role of a specific class of imaginary models to ease transitions between model $k$ and model $k'$. A proposal distribution is build by sampling an inhomogeneous Markov chain which targets at each step one of these auxiliary distributions in the spirit of annealed importance sampling \citep{neal2001annealed}. These auxiliary distributions take the form
\begin{align}\label{eqn_def_rho}
 \rho_{k\mapsto k'}^{(t)}(\mathbf{x}_k^{(t)},\mathbf{u}_{k\mapsto k'}^{(t)})&\propto \left[\pi(k,\mathbf{x}_k^{(t)})\, q_{k\mapsto k'}(\mathbf{u}_{k\mapsto k'}^{(t)}) \, |J_{T_{k\mapsto k'}}(\mathbf{x}_{k}^{(t)}, \mathbf{u}_{k\mapsto k'}^{(t)})|^{-1}\right]^{1-\gamma_t} \left[\pi(k',\mathbf{y}_{k'}^{(t)}) \, q_{k'\mapsto k}(\mathbf{u}_{k'\mapsto k}^{(t)})\right]^{\gamma_t}, \cr
 \rho_{k'\mapsto k}^{(t)}(\mathbf{y}_{k'}^{(t)},\mathbf{u}_{k'\mapsto k}^{(t)})&\propto \left[\pi(k,\mathbf{x}_k^{(t)})\, q_{k\mapsto k'}(\mathbf{u}_{k\mapsto k'}^{(t)}) \, |J_{T_{k\mapsto k'}}(\mathbf{x}_{k}^{(t)}, \mathbf{u}_{k\mapsto k'}^{(t)})|^{-1}\right]^{1-\gamma_{T-t}} \left[\pi(k',\mathbf{y}_{k'}^{(t)}) \, q_{k'\mapsto k}(\mathbf{u}_{k'\mapsto k}^{(t)})\right]^{\gamma_{T-t}},
\end{align}
for $t=0,\ldots,T$ where $T$ is a positive integer, $\gamma_0:=0, \gamma_T:=1$ and $\gamma_t\in[0,1]$ for $t\in\{1,\ldots,T-1\}$. We set $\gamma_t:=t/T$ in our numerical experiments as in \cite{karagiannis2013annealed}. When switching from model $k$ to model $k'$, we thus use at time $t$ a transition kernel $K_{k \mapsto k'}^{(t)}$ to target the distribution $\rho_{k\mapsto k'}^{(t)}$, which is at the beginning close to $(\pi(k,\cdot \,)\otimes q_{k\mapsto k'}) \, |J_{T_{k\mapsto k'}}|^{-1}$, and the end close to $\pi(k', \cdot \,)\otimes q_{k'\mapsto k}$. We wrote $\rho_{k\mapsto k'}^{(t)}$ as a function of $(\mathbf{x}_k^{(t)},\mathbf{u}_{k\mapsto k'}^{(t)})$ to emphasise that the starting point is $(\mathbf{x}_k^{(0)},\mathbf{u}_{k\mapsto k'}^{(0)})$. It is in fact also a function of $(\mathbf{y}_{k'}^{(t)}, \mathbf{u}_{k'\mapsto k}^{(t)})$ that can be found using $(\mathbf{y}_{k'}^{(t)},\mathbf{u}_{k'\mapsto k}^{(t)}) = T_{k\mapsto k'}(\mathbf{x}_k^{(t)},\mathbf{u}_{k\mapsto k'}^{(t)})$.

%The annealing distributions above are called \textit{geometric annealing distributions} in \cite{karagiannis2013annealed}. % Another choice of distributions is presented in that paper. We decided to present only geometric annealing distributions here because they seem to be the most practical.
% We build our proposal to switch from model $k$ to model $k'$ by sampling an inhomogeneous Markov chain of transition kernel $K_{k \mapsto k'}^{(t)}$ at time $t$ which is reversible with respect to $\rho_{k\mapsto k'}^{(t)}$.
The NRJ procedure incorporating such proposals is described in \autoref{algo_NRJ_andrieu_2013}. In Step 2.(b), the path can be generated through $(\mathbf{y}_{k'}^{(t)},\mathbf{u}_{k'\mapsto k}^{(t)})$ instead. % It is simply a question of which choice is the most practical.

\begin{algorithm}[ht]
\caption{NRJ incorporating the method of \cite{karagiannis2013annealed}} \label{algo_NRJ_andrieu_2013}
 \begin{enumerate}
 \itemsep 0mm

  \item Sample $u_c \sim \mathcal{U}(0, 1)$.

  \item[2.(a)\hspace{-4.5mm}] \hspace{4mm} If $u_c \leq \tau$, attempt a parameter update using a MCMC kernel of invariant distribution $\pi(\, \cdot \mid k)$ while keeping the values of the model indicator $k$ and direction $\nu$ fixed.

  \item[2.(b)\hspace{-4.5mm}] \hspace{4mm} If $u_c > \tau$, attempt a model switch from model $k$ to model $k' := k + \nu$. Sample $\mathbf{u}_{k \mapsto k'}^{(0)} \sim q_{k \mapsto k'}$ and $u_a\sim\mathcal{U}(0, 1)$, and set $\mathbf{x}_{k}^{(0)}:=\mathbf{x}_{k}$. Sample a path $(\mathbf{x}_{k}^{(1)}, \mathbf{u}_{k \mapsto k'}^{(1)}),\ldots,(\mathbf{x}_{k}^{(T-1)}, \mathbf{u}_{k \mapsto k'}^{(T-1)})$, where $(\mathbf{x}_{k}^{(t)}, \mathbf{u}_{k \mapsto k'}^{(t)})\sim K_{k \mapsto k'}^{(t)}((\mathbf{x}_{k}^{(t-1)}, \mathbf{u}_{k \mapsto k'}^{(t-1)}), \cdot \,)$. Compute $(\mathbf{y}_{k'}^{(t)},\mathbf{u}_{k'\mapsto k}^{(t)}) := T_{k\mapsto k'}(\mathbf{x}_{k}^{(t)},\mathbf{u}_{k \mapsto k'}^{(t)})$ for $t=0,\ldots,T-1$. If
      \small
  \begin{align*}
   \hspace{-7mm} u_a  \leq \alpha_{\text{NRJ2}}((k,\mathbf{x}_{k}^{(0)}),(k',\mathbf{y}_{k'}^{(T-1)}))&:=1 \wedge r_{\text{NRJ2}}((k,\mathbf{x}_{k}^{(0)}),(k',\mathbf{y}_{k'}^{(T-1)})) \quad \text{($r_{\text{NRJ2}}$ is defined in \eqref{eqn_ratio_A2013})},
  \end{align*}
  \normalsize
  set the next state of the chain to $(k',\mathbf{y}_{k'}^{(T-1)}, \nu)$. Otherwise, set it to $(k,\mathbf{x}_{k}, -\nu)$.

  \item[3.] Go to Step 1.
 \end{enumerate}
\end{algorithm}

\cite{karagiannis2013annealed} explain that the MH correction term in \autoref{algo_NRJ_andrieu_2013}, that we denote by
 \begin{align}\label{eqn_ratio_A2013}
r_{\text{NRJ2}}((k,\mathbf{x}_{k}^{(0)}),(k',\mathbf{y}_{k'}^{(T-1)})) := \prod_{t=0}^{T-1}\frac{\rho_{k\mapsto k'}^{(t+1)}(\mathbf{x}_k^{(t)},\mathbf{u}_{k\mapsto k'}^{(t)})}{\rho_{k\mapsto k'}^{(t)}(\mathbf{x}_k^{(t)},\mathbf{u}_{k\mapsto k'}^{(t)})},
\end{align}
 represents a consistent estimator of $\pi(k')/\pi(k)$ as $T\longrightarrow\infty$.
 %The authors in fact mention that it is the case for the product in the acceptance probabilities in RJ, but it is exactly the same product, which is multiplied by $g(k', k)/g(k,k')$ in their framework to make up their acceptance probabilities (recall the difference between acceptance probabilities in RJ and NRJ, see \eqref{eqn_acc_prob_RJ} and \eqref{eqn_acc_prob_NRJ} for instance).

 Under the following two conditions, the RJ corresponding to \autoref{algo_NRJ_andrieu_2013} and \autoref{algo_NRJ_andrieu_2013} itself are valid, in the sense that the target distribution is an invariant distribution. As mentioned in \cite{karagiannis2013annealed}, \eqref{eqn_symmetry} below is verified if for all $t$, $K_{k \mapsto k'}^{(t)}(\,\cdot \,, \cdot \,)$ and $K_{k' \mapsto k}^{(T-t)}(\,\cdot \,, \cdot \,)$ are Metropolis--Hastings (MH) kernels sharing the same proposal distributions.

\begin{description}
 \itemsep 0mm

 \item[Symmetry condition:] For $t=1,\ldots,T-1$ the pairs of transition kernels $K_{k \mapsto k'}^{(t)}(\,\cdot \,, \cdot \,)$ and $K_{k' \mapsto k}^{(T-t)}(\,\cdot \,, \cdot \,)$ satisfy
 \begin{align}\label{eqn_symmetry}
  K_{k \mapsto k'}^{(t)}((\mathbf{x}_{k}, \mathbf{u}_{k \mapsto k'}), \cdot \,) = K_{k' \mapsto k}^{(T-t)}((\mathbf{x}_{k}, \mathbf{u}_{k \mapsto k'}), \cdot \,) \quad \text{for any } (\mathbf{x}_{k}, \mathbf{u}_{k \mapsto k'}).
 \end{align}

 \item[Reversibility condition:] For $t=1,\ldots,T-1$, and for any $(\mathbf{x}_k,\mathbf{u}_{k\mapsto k'})$ and $(\mathbf{x}_{k}', \mathbf{u}_{k \mapsto k'}')$,
     \small
 \begin{align}\label{eqn_reversibility}
 \hspace{-10mm} \rho_{k\mapsto k'}^{(t)}(\mathbf{x}_k,\mathbf{u}_{k\mapsto k'})K_{k \mapsto k'}^{(t)}((\mathbf{x}_{k}, \mathbf{u}_{k \mapsto k'}), (\mathbf{x}_{k}', \mathbf{u}_{k \mapsto k'}'))=\rho_{k\mapsto k'}^{(t)}(\mathbf{x}_k',\mathbf{u}_{k\mapsto k'}')K_{k \mapsto k'}^{(t)}((\mathbf{x}_{k}', \mathbf{u}_{k \mapsto k'}'), (\mathbf{x}_{k}, \mathbf{u}_{k \mapsto k'})).
 \end{align}
 \normalsize
\end{description}

%\autoref{prop_invariance_andrieu_2013} below indicates that under the same conditions, \autoref{algo_NRJ_andrieu_2013} is valid as well.
%
%\begin{Proposition}[Invariance 2] \label{prop_invariance_andrieu_2013}
% Assume \eqref{eqn_symmetry} and \eqref{eqn_reversibility} are verified. The transition kernel of the Markov chain $\{(K, \mathbf{X}_K, \nu)(m): m\in\na\}$ simulated by \autoref{algo_NRJ_andrieu_2013} admits $\pi \otimes \mathcal{U}\{-1, 1\}$ as invariant distribution.
%\end{Proposition}
%
%\begin{proof} Analogous to that of \autoref{prop_invariance}. \end{proof}

The use of such sophisticated proposal schemes comes at a computational cost. As explained in \cite{karagiannis2013annealed}, the cost of using their approach is $\mathcal{O}(I \times T)$, $I$ denoting the number of iterations. Indeed, typically in Step 2.(b), $T - 1$ MH steps similar to those used to update the parameters (Step 2.(a)) are applied. Fortunately, the improvement as a function of $T$ for a fixed value of $I$ may be very marked for $T \leq T_0$, leading to better results that one would obtain by instead setting $T = 1$ (corresponding to \autoref{algo_NRJ} or vanilla RJ) and increasing $I$ to attain the same computational budget. This is what is observed for the multiple change-point problem presented in \autoref{sec_changepoint}. The cost is indeed offset by a large enough improvement in terms of total variation between the empirical and true marginal model distributions for $T$ in an interval including the value 100, which is the value used. That being said, the computational burden can be mitigated by designing better proposal functions $q_{k\mapsto k'}$ and $T_{k\mapsto k'}$ (when this is feasible), as shown in \cite{gagnon2019RJ}. In \autoref{sec_changepoint}, we only show the results of the sampler combining the approach of \cite{karagiannis2013annealed} with that presented in the next section for brevity.

\subsection{NRJ additionally with the method of \cite{andrieu2018utility}}\label{sec_andrieu_2018}

As mentioned, $r_{\text{NRJ2}}$ in \eqref{eqn_ratio_A2013} can be interpreted as an estimator of $\pi(k')/\pi(k)$. To further reduce the variance of this estimator, one could produce in parallel $N$ inhomogeneous Markov chains ending with $N$ proposals, that we denote by $\mathbf{y}_{k'}^{(T-1, 1)},\ldots,\mathbf{y}_{k'}^{(T-1, N)}$, and use instead the average of the $N$ estimates  $r_{\text{NRJ2}}((k,\mathbf{x}_{k}^{(0)}),(k', $ $\mathbf{y}_{k'}^{(T-1,1)})),\ldots, r_{\text{NRJ2}}((k,\mathbf{x}_{k}^{(0)}),(k',\mathbf{y}_{k'}^{(T-1,N)}))$. Simplifying notation, an estimate of $\pi(k')/\pi(k)$ is thus given by
\[
 \bar{r}(k, k'):=\frac{1}{N} \sum_{j=1}^N r_{\text{NRJ2}}((k,\mathbf{x}_{k}^{(0)}),(k',\mathbf{y}_{k'}^{(T-1,j)})).
\]
However, applying this method naively does not lead to valid algorithms. The approach of \cite{andrieu2018utility} exploits this averaging idea while leading to valid schemes. We present in \autoref{algo_NRJ_andrieu_2018} the NRJ version of this algorithm.

\begin{algorithm}[ht]
\caption{NRJ additionally incorporating the method of \cite{andrieu2018utility}} \label{algo_NRJ_andrieu_2018}
 \begin{enumerate}
 \itemsep 0mm

  \item Sample $u_{c, 1} \sim \mathcal{U}(0, 1)$.

  \item[2.(a)\hspace{-4.5mm}] \hspace{4mm} If $u_{c, 1} \leq \tau$, attempt a parameter update using a MCMC kernel of invariant distribution $\pi(\, \cdot \mid k)$ while keeping the values of the model indicator $k$ and direction $\nu$ fixed.

  \item[2.(b)\hspace{-4.5mm}] \hspace{4mm} If  $u_{c, 1} > \tau$, attempt a model switch from model $k$ to model $k' := k + \nu$. Sample $u_a, u_{c, 2}\sim \mathcal{U}(0, 1)$. If  $u_{c, 2}\leq 1/2$ go to Step 2.(b-i), otherwise go to Step 2.(b-ii).

  \item[2.(b-i)\hspace{-6.8mm}] \hspace{7mm} Sample $N$ proposals $\mathbf{y}_{k'}^{(T-1, 1)},\ldots,\mathbf{y}_{k'}^{(T-1, N)}$ as in Step 2.(b) of \autoref{algo_NRJ_andrieu_2013}. Sample $j^*$ from a PMF such that $\Prob(J^*=j)\propto r_{\text{NRJ2}}((k,\mathbf{x}_{k}),(k',\mathbf{y}_{k'}^{(T-1,j)}))$. If $u_a\leq \bar{r}(k, k')$, set the next state of the chain to $(k',\mathbf{y}_{k'}^{(T-1,j^*)}, \nu)$. Otherwise, set it to $(k,\mathbf{x}_{k}, -\nu)$.

  \item[2.(b-ii)\hspace{-8.5mm}] \hspace{4mm} \quad Sample one forward path as in Step 2.(b) of \autoref{algo_NRJ_andrieu_2013}. Denote by $\mathbf{y}_{k'}^{(T-1,1)}$ the endpoint. From $\mathbf{y}_{k'}^{(T-1,1)}$, generate $N-1$ reverse paths again as in Step 2.(b) of \autoref{algo_NRJ_andrieu_2013}, yielding $N-1$ proposals for the parameters of model $k$. If $u_a \leq \bar{r}(k', k)^{-1}$, set the next state of the chain to $(k',\mathbf{y}_{k'}^{(T-1,1)}, \nu)$. Otherwise, set it to $(k,\mathbf{x}_{k}, -\nu)$.

  \item[3.] Go to Step 1.
 \end{enumerate}
\end{algorithm}

%\autoref{prop_invariance_andrieu_2018} indicates that \autoref{algo_NRJ_andrieu_2018} is valid under the same conditions as \autoref{prop_invariance_andrieu_2013}.
%
%\begin{Proposition}[Invariance 3] \label{prop_invariance_andrieu_2018}
% Assume \eqref{eqn_symmetry} and \eqref{eqn_reversibility} are verified. The transition kernel of the Markov chain $\{(K, \mathbf{X}_K, \nu)(m): m\in\na\}$ simulated by \autoref{algo_NRJ_andrieu_2018} admits $\pi \otimes \mathcal{U}\{-1, 1\}$ as invariant distribution.
%\end{Proposition}
%
%\begin{proof} Analogous to that of \autoref{prop_invariance}. \end{proof}

\cite{andrieu2018utility} prove that increasing $N$ decreases the asymptotic variance of the Monte Carlo estimates produced by RJ incorporating their approach. Their proof cannot be easily extended to NRJ. However, we have observed empirically that increasing $N$ (as increasing $T$ in \autoref{algo_NRJ_andrieu_2013}) leads to a steady increase in the ESS until the samplers are close enough to be ideal. % So again, the strategy is to find the approximate location of the threshold and to select a suitable smaller value for $N$.

An advantage of the approach presented here over that presented in the previous section is that several computations can be executed in parallel. The $N$ proposals in Step 2.(b-i) are indeed generated from the same starting point, implying that this part and the computations of the ratios $r_{\text{NRJ2}}$ can be executed in parallel. Also, in Step 2.(b-ii), once $\mathbf{y}_{k'}^{(T-1,1)}$ has been generated, the $N - 1$ reverse paths can be generated in parallel. The computational cost associated to \autoref{algo_NRJ_andrieu_2018} is thus that of running \autoref{algo_NRJ_andrieu_2013} with the same value for $T$ but with an additional 50\% of model switches (because 50\% of model switches use Step 2.(b-ii)), and to this we need to add the cost of computational overhead which depends on $N$. If \autoref{algo_NRJ_andrieu_2018} is run for $I$ iterations, then its cost is upper bounded by that of running \autoref{algo_NRJ_andrieu_2013} for $1.5I$ iterations plus the cost of computational overhead.

We for instance try implementing \autoref{algo_NRJ_andrieu_2018} using the R package \texttt{parallel} for the multiple change-point problem presented in \autoref{sec_changepoint}. This package provides an easy way of executing tasks in parallel. For this implementation with $N = 10$ and $T = 100$, the computational cost is a little more than double that of \autoref{algo_NRJ_andrieu_2013} with the same value for $T$.

% For the multiple change-point problem presented in \autoref{sec_changepoint}, the ESS is a little more than doubled by using \autoref{algo_NRJ_andrieu_2018} with $T = 100$ and $N = 10$ compared to \autoref{algo_NRJ_andrieu_2013} with $T = 100$. There clearly exists an implementation of \autoref{algo_NRJ_andrieu_2018} with an implementation cost less than twice that of \autoref{algo_NRJ_andrieu_2013}. We for instance try implementing it using the R package \texttt{parallel}, which provides an easy but not necessarily efficient way of executing tasks in parallel. For this implementation, the computational cost is a little more than doubled.

% In a parallel computing environment, an advantage of the approach presented here is that the additional computational cost (over \autoref{algo_NRJ_andrieu_2013}) is negligible considering that we can generate the $N$ proposals $\mathbf{y}_{k'}^{(T-1, 1)},\ldots,\mathbf{y}_{k'}^{(T-1, N)}$ and compute the corresponding estimates of the acceptance ratio in parallel.

\subsection{Convergence of Algorithms~\ref{algo_NRJ_andrieu_2013} and \ref{algo_NRJ_andrieu_2018} towards ideal NRJ}\label{sec_conv_algo2}

We presented at the beginning of \autoref{sec_towards_ideal} intuitive reasons explaining why Algorithms~\ref{algo_NRJ_andrieu_2013} and \ref{algo_NRJ_andrieu_2018} can be made as close as we want to their ideal counterparts. We present here theoretical arguments supporting this intuition by establishing the weak convergence of the Markov chains produced by \autoref{algo_NRJ_andrieu_2013} towards those simulated by its ideal version as $T\longrightarrow\infty$. This implies that \autoref{algo_NRJ_andrieu_2018} with large enough $T$ and fixed $N$ generates Markov chains sharing the same behaviour as its ideal counterpart given that the noise $\varepsilon$ \eqref{eqn_acc_noisy} is only made more stable around the constant 1 by additionally using the approach of \cite{andrieu2018utility}. The corresponding weak convergence result for RJ incorporating the method of \cite{karagiannis2013annealed} holds under the same assumptions as those presented here.

The Markov kernel simulated by \autoref{algo_NRJ_andrieu_2013} (when switching models) is given by:
\small
\begin{align*}
	 &P_T((k,\mathbf{x}_k, \nu),(k',\mathbf{y}_{k'}, \nu')) :=q_{k\mapsto k+\nu}(\mathbf{u}_{k \mapsto k+\nu}^{(0)})\prod_{t=1}^{T-1} K_{k\mapsto k+\nu}^{(t)}((\mathbf{y}_{k+\nu}^{(t-1)},\mathbf{u}_{k+\nu \mapsto k}^{(t-1)}),(\mathbf{y}_{k+\nu}^{(t)},\mathbf{u}_{k+\nu \mapsto k}^{(t)}))  \cr
& \qquad\qquad \times \delta_{(k+\nu, \mathbf{y}_{k+\nu}^{(T-1)}, \nu)}(k', \mathbf{y}_{k'}, \nu') \, \alpha_{\text{NRJ2}}((k,\mathbf{x}_{k}),(k',\mathbf{y}_{k'})) \cr
   &\qquad + \delta_{(k,\mathbf{x}_k, -\nu)}(k',\mathbf{y}_{k'}, \nu')\int q_{k\mapsto k+\nu}(\mathbf{u}_{k \mapsto k+\nu}^{(0)})\prod_{t=1}^{T-1} K_{k\mapsto k+\nu}^{(t)}((\mathbf{y}_{k+\nu}^{(t-1)},\mathbf{u}_{k+\nu \mapsto k}^{(t-1)}),(\mathbf{y}_{k+\nu}^{(t)},\mathbf{u}_{k+\nu \mapsto k}^{(t)})) \cr
   &\qquad\qquad\times (1 - \alpha_{\text{NRJ2}}((k,\mathbf{x}_{k}), (k + \nu, \mathbf{y}_{k+\nu}^{(T-1)})) \, d\mathbf{u}_{k \mapsto k+\nu}^{(0)} \, d(\mathbf{y}_{k+\nu}^{(1)},\mathbf{u}_{k+\nu \mapsto k}^{(1)}) \ldots d(\mathbf{y}_{k+\nu}^{(T-1)},\mathbf{u}_{k+\nu \mapsto k}^{(T-1)}).
  \end{align*} \normalsize
 Use $\{(K,\mathbf{X}_K, \nu)_T(m): m \in\na\}$ to denote the Markov chain associated with this kernel. The ideal version of \autoref{algo_NRJ_andrieu_2013} presented in \autoref{sec_ideal} sets $q_{k\mapsto k'}:=\pi(\, \cdot\mid k')$. This is when switching models. For the parameter update step, we assume that both samplers use the same MCMC kernels of invariant distributions $\pi(\, \cdot \mid k)$. The Markov kernel simulated by the ideal version (when switching models) is thus given by:
\small
\begin{align*}
P_{\text{ideal}}((k,\mathbf{x}_k, \nu),(k',\mathbf{y}_{k'}, \nu'))&:=
% \pi(\mathbf{u}_{k \mapsto k+\nu} \mid k+\nu) \, \delta_{(k+\nu, \mathbf{u}_{k \mapsto k+\nu}, \nu)}(k',\mathbf{y}_{k'}, \nu') \, \left(1\wedge \frac{\pi(k')}{\pi(k)}\right) \cr
%   &\hspace{-10mm} + \delta_{(k,\mathbf{x}_k, -\nu)}(k',\mathbf{y}_{k'}, \nu')\int \pi(\mathbf{u}_{k \mapsto k+\nu} \mid k+\nu)\left(1 - 1\wedge \frac{\pi(k+\nu)}{\pi(k)}\right) \, d\mathbf{u}_{k \mapsto k + \nu} \cr
   %&=
   \pi(\mathbf{u}_{k \mapsto k+\nu} \mid k+\nu) \, \delta_{(k+\nu, \mathbf{u}_{k \mapsto k+\nu}, \nu)}(k',\mathbf{y}_{k'}, \nu') \, \left(1\wedge \frac{\pi(k')}{\pi(k)}\right) \cr
   &\qquad + \delta_{(k,\mathbf{x}_k, -\nu)}(k',\mathbf{y}_{k'}, \nu')\left(1 - 1\wedge \frac{\pi(k+\nu)}{\pi(k)}\right).
\end{align*} \normalsize
Use $\{(K,\mathbf{X}_K, \nu)_{\text{ideal}}(m): m \in\na\}$ to denote the corresponding Markov chain. The transitions in the ideal case are therefore such that with probability $1\wedge \pi(k+\nu)/\pi(k)$ there is a move to model $k+\nu$ with parameters $\mathbf{y}_{k+\nu}\sim\pi(\, \cdot\mid k+\nu)$ (and the direction $\nu$ is conserved). Otherwise, the model and parameters stay the same (and the direction $\nu$ is reversed). The two distinctive elements of the ideal sampler are the form of the acceptance probability and distribution of the proposal $\mathbf{y}_{k+\nu}$. Intuitively, if \autoref{algo_NRJ_andrieu_2013} proposes parameters with a distribution close to $\pi(\, \cdot\mid k+\nu)$ and accept them with a probability close to $1\wedge \pi(k+\nu)/\pi(k)$ (in the limit), the weak convergence should happen as the transition probabilities share the same behaviour. This is essentially what Theorem 1 in \cite{karr1975weak} indicates: if \autoref{algo_NRJ_andrieu_2013} and its ideal version are initialised in the same way (i.e.\ $(K,\mathbf{X}_K, \nu)_T(0)$ and $(K,\mathbf{X}_K, \nu)_{\text{ideal}}(0)$ follow the same distribution), and $P_T\longrightarrow P_{\text{ideal}}$ in some sense as $T\longrightarrow\infty$, then $\{(K,\mathbf{X}_K, \nu)_T(m): m \in\na\}$ converges weakly towards $\{(K,\mathbf{X}_K, \nu)_{\text{ideal}}(m): m \in\na\}$, denoted by $\{(K,\mathbf{X}_K, \nu)_T(m): m \in\na\}\Longrightarrow\{(K,\mathbf{X}_K, \nu)_{\text{ideal}}(m): m \in\na\}$, as $T\longrightarrow\infty$.

We already know that the acceptance probabilities are the same in the limit for both samplers as it is mentioned in \cite{karagiannis2013annealed} that $r_{\text{NRJ2}}((k,\mathbf{x}_{k}), (k+\nu,\mathbf{y}_{k+\nu}^{(T-1)}))$ is a consistent estimator of $\pi(k+\nu)/\pi(k)$ as $T\longrightarrow\infty$ under realistic assumptions. For our result, we more precisely consider the following assumption.
\begin{Assumption}\label{ass_thm_conv_algo2_1}
 The random variable $r_{\text{NRJ2}}((k,\mathbf{x}_{k}),(k+\nu,\mathbf{Y}_{k+\nu}^{(T-1)}))$ converges in distribution towards $\pi(k+\nu)/\pi(k)$ as $T\longrightarrow \infty$, for any given $(k,\mathbf{x}_{k}, \nu)$.
\end{Assumption}

The proposals for the parameters in \autoref{algo_NRJ_andrieu_2013} $\mathbf{y}_{k+\nu}^{(T-1)}$ should in practice be distributed in the limit as $\pi(\, \cdot\mid k+\nu)$. Indeed, consider as in our practical example in \autoref{sec_changepoint} and those in \cite{karagiannis2013annealed} that $K_{k\mapsto k+\nu}^{(t)}$ are $\rho_{k\mapsto k+\nu}^{(t)}$-reversible MH kernels in which the proposal distributions are the same for all $t$. %Denote $q_{\text{NRJ2}}$ the joint proposal distribution (that thus does not depend on $t$) that is used to generate $(\mathbf{y}_{k+\nu}^{(t)},\mathbf{u}_{k+\nu \mapsto k}^{(t)})$ starting from $(\mathbf{y}_{k+\nu}^{(t-1)},\mathbf{u}_{k+\nu \mapsto k}^{(t-1)})$ in
This more precisely means that
\small
\begin{align*}
 &K_{k\mapsto k+\nu}^{(t)}((\mathbf{y}_{k+\nu}^{(t-1)},\mathbf{u}_{k+\nu \mapsto k}^{(t-1)}),(\mathbf{y}_{k+\nu}^{(t)},\mathbf{u}_{k+\nu \mapsto k}^{(t)})) := q_{\text{NRJ2}}^{k,\nu}((\mathbf{y}_{k+\nu}^{(t-1)},\mathbf{u}_{k+\nu \mapsto k}^{(t-1)}),(\mathbf{y}_{k+\nu}^{(t)},\mathbf{u}_{k+\nu \mapsto k}^{(t)})) \cr
 &\qquad \times \left(1\wedge \frac{\rho_{k\mapsto k+\nu}^{(t)}(\mathbf{y}_{k+\nu}^{(t)},\mathbf{u}_{k+\nu \mapsto k}^{(t)}) \, q_{\text{NRJ2}}^{k,\nu}((\mathbf{y}_{k+\nu}^{(t)},\mathbf{u}_{k+\nu \mapsto k}^{(t)}),(\mathbf{y}_{k+\nu}^{(t-1)},\mathbf{u}_{k+\nu \mapsto k}^{(t-1)}))}{\rho_{k\mapsto k+\nu}^{(t)}(\mathbf{y}_{k+\nu}^{(t-1)},\mathbf{u}_{k+\nu \mapsto k}^{(t-1)}) \, q_{\text{NRJ2}}^{k,\nu}((\mathbf{y}_{k+\nu}^{(t-1)},\mathbf{u}_{k+\nu \mapsto k}^{(t-1)}),(\mathbf{y}_{k+\nu}^{(t)},\mathbf{u}_{k+\nu \mapsto k}^{(t)}))}\right) \cr
 &\qquad +\delta_{(\mathbf{y}_{k+\nu}^{(t-1)},\mathbf{u}_{k+\nu \mapsto k}^{(t-1)})}(\mathbf{y}_{k+\nu}^{(t)},\mathbf{u}_{k+\nu \mapsto k}^{(t)}) \, \Prob_{q_{\text{NRJ2}}}(\text{rejection}\mid \mathbf{y}_{k+\nu}^{(t-1)},\mathbf{u}_{k+\nu \mapsto k}^{(t-1)}),
\end{align*}
\normalsize
where $\Prob_{q_{\text{NRJ2}}}(\text{rejection}\mid \mathbf{y}_{k+\nu}^{(t-1)}, \mathbf{u}_{k+\nu \mapsto k}^{(t-1)})$ is the rejection probability starting from $(\mathbf{y}_{k+\nu}^{(t-1)},$ $\mathbf{u}_{k+\nu \mapsto k}^{(t-1)})$ and using $q_{\text{NRJ2}}^{k,\nu}$ as the proposal distribution (which is the same for all $t$). If $0<t^*<T$ and $T$ are such that $t^*/T$ is close to 1, then $\rho_{k\mapsto k+\nu}^{(t^*)}$ is essentially proportional to $\pi(\, \cdot \mid k+\nu) \otimes q_{k+\nu\mapsto k}$ (see \eqref{eqn_def_rho}), and this is true for all $t\geq t^*$. Therefore, the process associated to $K_{k\mapsto k+\nu}^{(t)}$ with $t \geq t^*$ is essentially a time-homogeneous Markov chain with $\pi(\, \cdot\mid k+\nu)\otimes q_{k+\nu \mapsto k}$ as a stationary distribution. This is why if $T$ is additionally such that $T-t^*$ is large enough, then $\mathbf{Y}_{k+\nu}^{(T-1)}$ is (approximately) distributed as $\pi(\, \cdot\mid k+\nu)$.

Consider $\{(\mathbf{Y}_{k+\nu}, \mathbf{U}_{k+\nu \mapsto k})(m): m \in\na\}$ to be the time-homogeneous $\pi(\, \cdot\mid k+\nu)\otimes q_{k+\nu \mapsto k}$-reversible Markov chain associated with the proposal distribution $q_{\text{NRJ2}}^{k,\nu}$ (thus generated by a regular MH algorithm with the same proposal distribution $q_{\text{NRJ2}}^{k,\nu}$ for all iteration $m$ with a stationary distribution that is fixed and set to be $\pi(\, \cdot\mid k+\nu)\otimes q_{k+\nu \mapsto k}$). The design of $q_{\text{NRJ2}}^{k,\nu}$ has an impact on how large the distance between $T$ and $t^*$ need to be to have $\mathbf{Y}_{k+\nu}^{(T-1)}$ approximately distributed as $\pi(\, \cdot\mid k+\nu)$. In our weak convergence result, we assume to simplify that it is such that the associated Markov chain is uniformly ergodic. It is highlighted in the proof what modifications and which additional technical conditions are required if geometric ergodicity is instead assumed.

\begin{Assumption}\label{ass_thm_conv_algo2_2}
 For all $k$ and $\nu$, the time-homogeneous $\pi(\, \cdot\mid k+\nu)\otimes q_{k+\nu \mapsto k}$-reversible Markov chain associated with the proposal distribution $q_{\text{NRJ2}}^{k,\nu}$, $\{(\mathbf{Y}_{k+\nu}, \mathbf{U}_{k+\nu \mapsto k})(m): m \in\na\}$, is uniformly ergodic.
\end{Assumption}

Finally, we assume regularity conditions on the PDF $q_{\text{NRJ2}}^{k,\nu}$.

\begin{Assumption}\label{ass_thm_conv_algo2_3}
 For all $k$ and $\nu$, $q_{\text{NRJ2}}^{k,\nu}((\mathbf{y}_{k+\nu}^{(t-1)},\mathbf{u}_{k+\nu \mapsto k}^{(t-1)}),(\mathbf{y}_{k+\nu}^{(t)},\mathbf{u}_{k+\nu \mapsto k}^{(t)}))$ and
 \begin{align}\label{eqn_condition_ratio_q}
  \frac{q_{\text{NRJ2}}^{k,\nu}((\mathbf{y}_{k+\nu}^{(t)},\mathbf{u}_{k+\nu \mapsto k}^{(t)}),(\mathbf{y}_{k+\nu}^{(t-1)},\mathbf{u}_{k+\nu \mapsto k}^{(t-1)}))}{q_{\text{NRJ2}}^{k,\nu}((\mathbf{y}_{k+\nu}^{(t-1)},\mathbf{u}_{k+\nu \mapsto k}^{(t-1)}),(\mathbf{y}_{k+\nu}^{(t)},\mathbf{u}_{k+\nu \mapsto k}^{(t)}))}
 \end{align}
 are bounded above by a positive constant that depends only on $k$ and $\nu$.
\end{Assumption}

Note that \eqref{eqn_condition_ratio_q} is equal to 1 if $q_{\text{NRJ2}}^{k,\nu}$ is symmetric. We are now ready to present the weak convergence result.

\begin{Theorem}[Weak convergence of \autoref{algo_NRJ_andrieu_2013}] \label{thm_conv_algo2}

	Under Assumptions~\ref{ass_thm_conv_algo2_1} to \ref{ass_thm_conv_algo2_3} and assuming that $(K,\mathbf{X}_K, \nu)_T(0)\sim \pi\otimes \mathcal{U}\{-1, 1\}$ and $(K,\mathbf{X}_K, \nu)_{\text{ideal}}(0)\sim \pi\otimes \mathcal{U}\{-1, 1\}$, we have
 \[
  \{(K,\mathbf{X}_K,\nu)_T(m): m\in\na\}\Longrightarrow \{(K,\mathbf{X}_K,\nu)_{\text{ideal}}(m): m\in\na\} \quad \text{as} \quad T\longrightarrow \infty.
 \]

\end{Theorem}

\section{About NRJ performance}\label{sec_optimality}

We provide in \autoref{sec_aymp_var} a theoretical result showing that the proposed NRJ samplers yield ergodic averages of lower asymptotic variance than the corresponding RJ samplers proposing uniformly at random either model $k - 1$ or model $k + 1$ if a model switch is attempted. Empirically, the level of improvement of NRJ over RJ increases as the samplers better approximate the ideal ones. We show in \autoref{sec_analysis_shape} that the level of improvement also depends on the shape of the target. We finish in \autoref{sec_scaling_limits} with a discussion about scaling limit results formalising a sharp divide in the exploration behaviour of NRJ versus RJ.

\subsection{Asymptotic variance of ergodic averages}\label{sec_aymp_var}

A corollary of Theorem 3.17 in \cite{andrieu2019peskun} allows us to compare ergodic averages produced by $P_{\text{NRJ}}$ with those produced by $P_{\text{RJ}}^{\text{unif}}$, where $P_{\text{NRJ}}$ is the Markov kernel simulated by any NRJ (such as Algorithm \ref{algo_NRJ}, \ref{algo_NRJ_andrieu_2013} or \ref{algo_NRJ_andrieu_2018}) and $P_{\text{RJ}}^{\text{unif}}$ that simulated by its reversible counterpart with $g(k, k + 1) / (1 - \tau)  = g(k, k - 1) / (1 - \tau) = 1 / 2$ (representing conditional probabilities given that a model switch is proposed). This result establishes that the asymptotic variance of ergodic averages produced by $P_{\text{NRJ}}$ is at most equal to that of ergodic averages produced by $P_{\text{RJ}}^{\text{unif}}$ for bounded test functions.
%We however cannot expect for ergodic averages of any function $f$ of the parameters $\mathbf{x}_k$ produced by $P_{\text{NRJ}}$ to have a smaller variance than those produced by $P_{\text{RJ}}^{\text{unif}}$ because the same parameter proposal schemes are used in both samplers.
In particular, the estimates of posterior model probabilities have a lower asymptotic variance under $P_{\text{NRJ}}$ than under $P_{\text{RJ}}^{\text{unif}}$.
\begin{Corollary}\label{cor_asymptotic_var}
 For any real-valued bounded function $f$ of $(k, \mathbf{x}_k)$ considered without loss of generality to have zero mean under the target,
 \[
  \text{var}_\lambda(f, P_{\text{NRJ}}) \leq \text{var}_\lambda(f, P_{\text{RJ}}^{\text{unif}}),
 \]
 where $\text{var}_\lambda(f, P) := \E[\{f(K(0), \mathbf{X}_K(0))\}^2] + 2\sum_{m>0} \lambda^m \E[f(K(0), \mathbf{X}_K(0))f(K(m), \mathbf{X}_K(m))] $ of $\{(K(m), \mathbf{X}_K(m)): m\in\na\}$ being a Markov chain of transition kernel $P$ at equilibrium and $\lambda \in[0, 1)$. If $P_{\text{NRJ}}$ is uniformly ergodic, $\text{var}_\lambda(f, P_{\text{NRJ}})$ converges to the asymptotic variance $\text{var}(f, P_{\text{NRJ}}) := \E[\{f(K(0), \mathbf{X}_K(0))\}^2] + 2\sum_{m>0} \E[f(K(0), \mathbf{X}_K(0))f(K(m), \mathbf{X}_K(m))] $ as $\lambda \longrightarrow 1$, and therefore, $\text{var}(f, P_{\text{NRJ}})$ $ \leq \text{var}(f, P_{\text{RJ}}^{\text{unif}})$ (the limit always exists for a reversible Markov chain).
\end{Corollary}

We highlight in the proof of the corollary which additional technical condition is required for the limit to hold under geometric ergodicity.

\subsection{Dependence on the shape of the target}\label{sec_analysis_shape}

We have shown that it is possible to construct samplers as close as we want to their ideal counterparts, at least in the weak convergence sense. We focus in the rest of the section on the marginal ideal behaviour of $K$ associated with the ideal RJ and NRJ.
% to establish the dominance of the latter when the target distribution is not too concentrated.
We only consider iterations in which model switches are proposed to focus on this type of transitions. In particular we do not study the impact of the proportion of parameter updates $\tau$, but we discuss briefly how this parameter is selected in \autoref{sec_implementation}.

In \autoref{sec_existing_results}, existing results describing the behaviours of ideal RJ and NRJ when the marginal distribution is uniform or log concave are presented. NRJ outperform RJ in the former case, but not necessarily in the latter if we consider using functions $g$ incorporating information about this marginal distribution instead of the uniform as discussed earlier. To analyse this latter case further, we use a parameter $\phi\geq 1$ to characterise log concave distributions in \autoref{sec_worst_logconcave}, and present a family of ``worst'' (for NRJ) log concave distributions for which the larger is $\phi$ the more concentrated is the PMF. NRJ outperform RJ when $\phi$ is not too large and the target is a member of this family.

\subsubsection{Existing results}\label{sec_existing_results}

Denote by $\{K_{\text{ideal}}^{\text{RJ}, \, g}(m): m\in\na\}$ and $\{K_{\text{ideal}}^{\text{NRJ}}(m): m\in\na\}$ the Markov chains produced by ideal RJ and NRJ, where we highlighted that the behaviour of RJ depends on $g$. We show here how this proposal distribution can impact performance.

We consider a scenario where $\mathcal{K}:=\{1,\ldots,\text{K}_{\max}\}$. When the target is uniform on this set, the process $\{K_{\text{ideal}}^{\text{NRJ}}(m): m\in\na\}$ evolves deterministically; all proposals are accepted and it thus goes from $1$ to $\text{K}_{\max}$ without stopping, and changes direction at $\text{K}_{\max}$ to return to 1. It is thus periodic and the distribution of $K_{\text{ideal}}^{\text{NRJ}}(m)$ does not converge towards the target as $m \longrightarrow \infty$. This is however not an issue when approximating expectations with respect to the target. A randomised version of $\{K_{\text{ideal}}^{\text{NRJ}}(m): m\in\na\}$ exists, see \cite{diaconis2000analysis}\footnote{The difference is that instead of systematically changing direction at 1 and $\text{K}_{\max}$, the sampler  changes direction probabilistically after on average $\text{K}_{\max}$ steps, making it aperiodic.}. These authors prove that their process also explores the space in $\mathcal{O}(\text{K}_{\max})$ steps while it takes $\mathcal{O}(\text{K}_{\max}^2)$ for $\{K_{\text{ideal}}^{\text{RJ}, \, g^*}(m): m\in\na\}$, $g^*$ being the optimal proposal distribution. The usual symmetric distribution $g^*(k, k+1)=g^*(k, k-1)=1/2$ is the optimal (conditional) proposal distribution (given that a model switch is proposed) in this case among all symmetric stochastic tridiagonal matrices \citep{boyd2006fastest}; i.e.\ when one restricts oneself to proposals of the form $k\mapsto k'\in\{k-1, k+1\}$. This thus establishes the superiority of NRJ over RJ in this case among samplers with proposals of the form $k\mapsto k'\in\{k-1, k+1\}$. %We do not consider the sampler of \cite{diaconis2000analysis} as a competitor to ours because they are, in essence, the same.  This difference makes it non-deterministic and aperiodic.

Superiority for uniform targets is an interesting theoretical result, but this is not a scenario of interest in Bayesian model selection. We believe that it is more likely that the posterior distributions reflect a balance between too simple models (that are more stable but do not capture well the dynamics in the data) and too complex models (that overfit and have less generalisation power), in the spirit of Occam's razor. Unimodal distributions, which are such that $\pi(1)\leq \ldots\leq \pi(k^*)\geq \ldots \geq \pi(\text{K}_{\max})$, are in this sense more interesting to analyse. \cite{hildebrand2002analysis} generalised the result of \cite{diaconis2000analysis} on the Markov chain similar to $\{K_{\text{ideal}}^{\text{NRJ}}(m): m\in\na\}$ to log concave distributions, defined as distributions such that $\pi(k)/\pi(k-1)\geq \pi(k+1)/\pi(k)$ for all $k\in\{2,\ldots,\text{K}_{\max} - 1\}$. Log concave distributions belong to the family of unimodal distributions. Indeed, if we consider for instance $k>k^*$ (the mode), we observe that the ratios $\pi(k+1)/\pi(k)$ are smaller and smaller as we get further away from the mode.

An adaptation of the proof of \cite{hildebrand2002analysis} allows us to prove that $\mathcal{O}(\text{K}_{\max})$ steps are sufficient for $\{K_{\text{ideal}}^{\text{NRJ}}(m): m\in\na\}$ to traverse the state-space, if we assume that the distribution is log concave, but not uniform. For $\{K_{\text{ideal}}^{\text{RJ}, \, g}(m): m\in\na\}$, no such results are available. To establish the superiority of NRJ when the target is not too concentrated, we need to identify the optimal proposal distribution $g^*$ for RJ and to prove that the number of required steps is larger. We take here a step in this direction.

We choose the competitor to NRJ to be the RJ with the distribution $g^*$ given by
\begin{align}\label{eqn_g_sqrt}
  g^*(k, k')\propto \sqrt{\pi(k')/\pi(k)} \quad \text{for} \quad k'\in\{k-1, k+1\}.
\end{align}
This choice finds its justification in \cite{zanella2019informed}, in which it is shown that a class of what the author calls \textit{informed} distributions with $g^*$ as a special case are optimal within reversible samplers in some situations. In fact, it is possible to numerically show that the optimal distribution in terms of speed of convergence among distributions $g(k, \cdot \,)$ defined on $\{k-1, k+1\}$ is very close to $g^*$ with a negligible speed difference when the log concave distribution belongs to the family defined in the next section. % Other justifications are provided in supplementary material (see \autoref{sec_supp}).
Note that the optimal proposal distribution $g^*(k,k+1)=g^*(k,k-1)=1/2$ is retrieved when the target is uniform.

For any log concave target, $g^*$ is such that
\[
 \frac{1}{\pi(k)/\pi(k+1) + 1} \leq g^*(k, k+1)\leq \frac{1}{\pi(k-1)/\pi(k) + 1}.
\]
The acceptance probabilities are controlled in the same way by ratios of posterior model probabilities. As long as the target is not too concentrated, meaning ratios not too far from 1, $\{K_{\text{ideal}}^{\text{RJ}, \, g^*}(m): m\in\na\}$ thus still has a diffusive behaviour that makes it traverse the state-space in of order of $\text{K}_{\max}^2$ steps. However, for concentrated targets,  $g^*(k, k+1)$ gets close to 1 when the chain is at the left of the mode as $\pi(k)/\pi(k+1)$ and $\pi(k-1)/\pi(k)$ are close to 0. The stochastic process thus moves persistently towards the mode and wanders around it afterwards.

\begin{Remark}
 \cite{diaconis2000analysis}, and afterwards \cite{hildebrand2004rates}, also studied V-shaped distributions, which are a class of multimodal distributions. They showed that under regularity conditions it takes on the order of $\text{K}_{\max}^2$ and $\text{K}_{\max}^2 \log \text{K}_{\max}$ steps to converge towards the target for the non-reversible and reversible (with a uniform proposal distribution) samplers, respectively, suggesting that the non-reversible sampler makes the chain move quicker from a mode to another than its reversible counterpart for some multimodal distributions.
\end{Remark}

\subsubsection{Log concave distributions: a worst case scenario}\label{sec_worst_logconcave}

One way to characterise any log concave distribution is through the minimum of the ratios $\pi(k - 1)/\pi(k)$, for $k \leq k^*$, and $\pi(k + 1)/\pi(k)$, for $k \geq k^*$. Consider that this minimum is $1/\phi$. The distribution with a constant decreasing factor from the mode of $1/\phi$ leads to RJ with $g^*$ with the most significant advantage over RJ with a symmetric proposal. This is because the distribution is the most concentrated, which at the same time leaves not much room for persistent movement for NRJ.

We now introduce a class of distributions with this characteristic. This class is such that the mode $k^*$ is at the middle of the domain:
\begin{align}\label{eqn_distributions_phi}
 \frac{\pi(k+1)}{\pi(k)}=\frac{1}{\phi} \quad \text{for} \quad k\geq k^*, \quad \text{and} \quad \frac{\pi(k-1)}{\pi(k)}=\frac{1}{\phi} \quad \text{for} \quad k\leq k^*, \quad \text{with} \quad \phi > 1.
\end{align}

Numerically, if we set the target to be the distribution in \eqref{eqn_distributions_phi}, we observe in \autoref{fig_3}
\begin{wrapfigure}{r}{0.35\textwidth}
\begin{center}
\vspace{-5mm}
\includegraphics[width=0.37\textwidth]{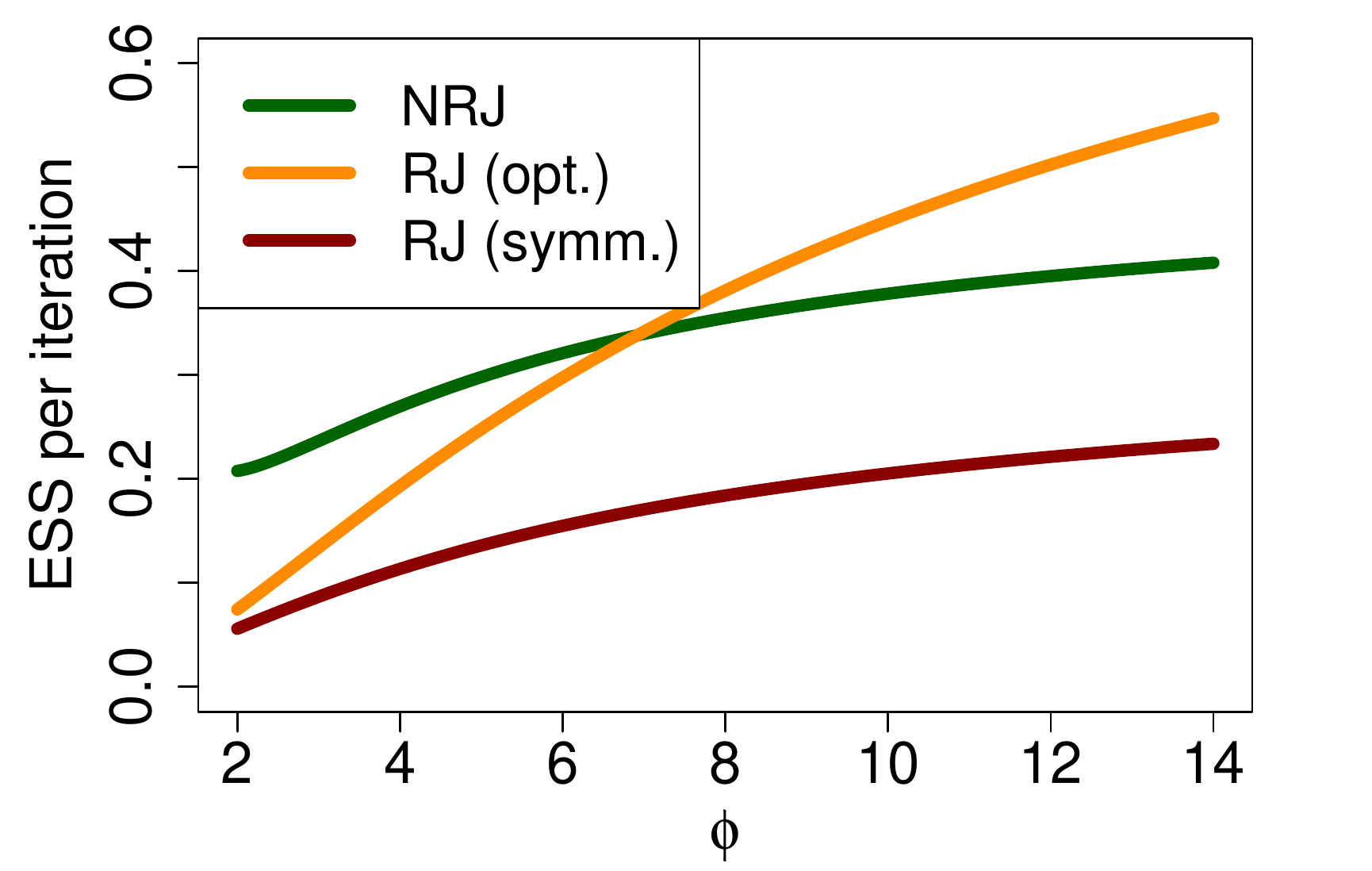}
\end{center}
\vspace{-9mm}
\caption{\small ESS as a function of $\phi$ for NRJ and RJ (with $g^*$ and symmetric $g$), when $\text{K}_{\max}:=11$}\label{fig_3}
\vspace{-5mm}
\end{wrapfigure}
     that NRJ outperforms RJ in terms of ESS when the target is not too concentrated by a factor, for instance, up to 2.8 when $\text{K}_{\max}:=11$. The concentration threshold is in this case around $\phi = 7$; beyond this the target is too concentrated and RJ with $g^*$ slowly starts to perform better. Beyond this threshold, RJ with $g^*$ is more efficient because there are basically 3 possible values for $K$: the mode $k^*$, $k^*+1$ and $k^*-1$. Indeed, when $\phi$ is exactly 7, the total mass outside of these values is $2 \sum_{k = k^* + 2}^{\text{K}_{\max}} \pi(k^*) / \phi^{k - k^*} \approx 3.57\%$. This is essentially true for any value of $\text{K}_{\max}$ as this percentage is equal to the limiting value (to two decimal places) as $\text{K}_{\max}\longrightarrow\infty$. When there are 3 possible values, starting from the mode $k^*$, both NRJ and RJ with $g^*$ go either to the right or the left with equal probability (on average for NRJ given that $\nu\sim\mathcal{U}\{-1,1\}$). Let us say that they go to $k^*+1$. The difference is that NRJ tries to go to $k^*+2$ (because of the direction), which is likely to be rejected, and therefore, it stays for one iteration at $k^*+1$; RJ directly goes back to $k^*$ given that $g^*(k^*+1, k^*)=\phi / (\phi + 1)$ is close to 1. RJ thus seems to have an advantage in terms of required number of steps to traverse the state-space.

To summarise, NRJ is expected to perform better than RJ with $g^*$ for any log concave distribution such that the minimum of the ratios $\pi(k - 1)/\pi(k)$ and $\pi(k + 1)/\pi(k)$ is larger than $1 / \phi^*$, where $\phi^*\approx 7$. We noticed that $g^*$ uses information about the target which is obviously not available prior pilot runs. Given that NRJ always outperform RJ with a symmetric proposal (\autoref{cor_asymptotic_var}), we thus recommend as practical guidelines to start by using NRJ, and if after pilot runs the target appears strongly concentrated, then it may be beneficial to switch to RJ with $g^*$. In the multiple change-point example in \autoref{sec_changepoint}, the target is for instance not too concentrated and RJ with $g^*$ performs similarly to RJ with the symmetric proposal.

% We finish this section by noting that when the target is strongly concentrated, attempts to go to models other than the mode have a high rejection rate. After the rejection, attempting a parameter update within the same iteration leaves the distribution invariant. This strategy may be implemented  to improve the mixing of the parameters.

\subsection{Scaling limits of model indicator process}\label{sec_scaling_limits}

Another way to evaluate the performance of algorithms is through the identification and analysis of scaling limits of their associated stochastic processes as the dimension $d$ of the state-space goes to infinity. \cite{roberts1997weak} and \cite{roberts1998optimal} applied this strategy to optimally tune the random walk Metropolis (RWM) and Metropolis-adjusted Langevin algorithm (MALA), but their analyses can also be used to establish that MALA is more efficient than RWM. This follows from the fact that to obtain non-trivial continuous limiting stochastic processes we need to speed up time by factors $d$ and $d^{1/3}$ for RWM and MALA, respectively. We explore such scaling limits for the processes $\{K_{\text{ideal}}^{\text{RJ}, \, g}(m): m\in\na\}$ and $\{K_{\text{ideal}}^{\text{NRJ}}(m): m\in\na\}$.

%A characteristic that makes the analysis difficult is the fact these processes take values in a discrete set. \cite{syed2019non} face the same challenge for their scaling limit analysis in parallel tempering (PT) of an index process which keeps track of the sequence of annealing parameters on a given machine. Fortunately for these authors, their \textit{position} variable takes values in $\{0, \ldots, N\}$ and has a uniform distribution, where $N$ is the number of machines. By dividing its value by $N$, they thus obtain a limiting continuous random variable on the unit interval, as $N \longrightarrow \infty$. They prove that the reversible and non-reversible PT weakly converge towards a diffusion and piecewise deterministic Markov process, respectively, identifying a sharp divide in the limiting behaviour of these schemes and establishing that the convergence is faster by an order of magnitude for non-reversible PT, as the required time accelerations are $N^2$ and $N$.

In our framework, we have no guarantee that the model indicator variable will converge towards a continuous random variable as $\text{K}_{\max}$ increases. In the supplementary material (\autoref{sec_theoretical_ana}), we present strong and technical assumptions on $\pi(k)$ under which results analogous to those of \cite{syed2019non} are obtained: the reversible process suitably rescaled converges to a diffusion while the non-reversible version converges to a piecewise-deterministic Markov process (see Theorems \ref{thm_conv_RJ} and \ref{thm_conv_NRJ} in \autoref{sec_theoretical_ana}). The required time rescalings lead to conclusions consistent with the results presented in the previous section showing that $\mathcal{O}(\text{K}_{\max}^2)$ and $\mathcal{O}(\text{K}_{\max})$ steps are required to explore the state-space for $\{K_{\text{ideal}}^{\text{RJ}, \, g}(m): m\in\na\}$ and $\{K_{\text{ideal}}^{\text{NRJ}}(m): m\in\na\}$, respectively.

\section{Numerical experiments}\label{sec_numerical}

Recall that in the usual non-ideal situation, the acceptance ratio in $\alpha_{\text{NRJ}}$ can be viewed as the ideal ratio $\pi(k' )/\pi(k)$ corrupted by some multiplicative noise; see \eqref{eqn_acc_noisy}. In practice, the noise fluctuates around 1. In \autoref{sec_noise}, we show how the difference in performance between NRJ and RJ varies when the noise amplitude changes (in a sense made precise in that section), or in other words as we move away or towards ideal NRJ and RJ. The methods presented in \autoref{sec_towards_ideal} are then applied to illustrate how their beneficial effect translates in practice for different noise behaviours. We also show how performances vary when the total number of models increases % and the shape of the marginal PMF $\pi(k)$ varies
on a simple target distribution for which we can precisely control the noise behaviour and number of models. % and the shape of the PMF.
In \autoref{sec_changepoint}, we evaluate the performance of NRJ and RJ in a real multiple change-point problem.

\subsection{Simulation study}\label{sec_noise}

 Let the target distribution be
\[
 \pi(k, \mathbf{x}_k) = p_{\phi, \text{K}_{\max}}(k)\prod_{i=1}^k \varphi(x_{i,k}),
\]
where $p_{\phi, \text{K}_{\max}}$ is the PMF defined in \autoref{sec_worst_logconcave} in \eqref{eqn_distributions_phi}, $\varphi$ is the density of a standard normal and $\mathbf{x}_k:=(x_{1,k},\ldots,x_{k,k})$. When switching from model $k$ to model $k+1$ in this case, one parameter needs to be added. It is not necessary to move the parameters that were in model $k$ given that they have the same distributions as the first $k$ parameters of model $k+1$. %This situation can be thought of as principal component regression (PCR) for which the samplers include an additional principal component (PC) or exclude one when there is a model switch. In fact, it is more accurate to say that it corresponds to the PCR situation where the scale parameters are known; consequently, the regression coefficients are independent and their distributions do not change when PCs are added. Here, it is additionally assume that the parameters have the same distribution.
In this context, it is straightforward to specify the functions $T_{k\mapsto k+1}$ that are required for the implementation of RJ and NRJ: they are such that the proposals for the parameters of model $k+1$ are $\mathbf{y}_{k+1} =(\mathbf{x}_k, u_{k \mapsto k+1})$. This also defines the functions $T_{k+1\mapsto k}$ for the (deterministic) reverse moves. Note that $\mathbf{u}_{k+1\mapsto k} = \varnothing$ for all $k$.

We have $\pi(k+1)/\pi(k)=p_{\phi, \text{K}_{\max}}(k+1)/p_{\phi, \text{K}_{\max}}(k)$, and the noise term $\varepsilon$ is given by
\begin{align}\label{eqn_noise}
  \frac{\pi(\mathbf{y}_{k+1}\mid k+1) \, q_{k+1\mapsto k}(\mathbf{u}_{k+1\mapsto k}) }{\pi(\mathbf{x}_{k}\mid k) \, q_{k\mapsto k+1}(\mathbf{u}_{k\mapsto k+1})\, |J_{T_{k\mapsto k+1}}(\mathbf{x}_{k}, \mathbf{u}_{k\mapsto k+1})|^{-1}}=\frac{\varphi(u_{k \mapsto k+1})}{q_{k\mapsto k+1}(u_{k \mapsto k+1})}.
\end{align}
We can therefore precisely control the noise behaviour by setting $q_{k\mapsto k+1}=\mathcal{N}(0,\sigma^2)$, where $\sigma>0$ is the varying parameter. Indeed, in this case
\[
 \frac{\varphi(u_{k \mapsto k+1})}{q_{k\mapsto k+1}(u_{k \mapsto k+1})} = \sigma \exp\left[-\frac{u_{k \mapsto k+1}^2}{2}\left(1-\frac{1}{\sigma^2}\right)\right],
\]
which behaviour varies with $\sigma$ given that $u_{k \mapsto k+1}\sim \mathcal{N}(0,\sigma^2)$. This is also true for the reverse move.
%Regarding the reverse move $k+1\mapsto k$, the noise term is
%\[
% \frac{q_{k\mapsto k+1}(x_{k+1,k+1})}{f(x_{k+1,k+1})} = \frac{1}{\sigma} \exp\left[\frac{x_{k+1,k+1}^2}{2}\left(1-\frac{1}{\sigma^2}\right)\right],
%\]
%which only depends on $\sigma$ ($x_{k+1,k+1}\sim \mathcal{N}(0, 1)$ regardless of the value of $\sigma$).
A small $\sigma$ represents a proposal distribution that is more concentrated around the mode than the target, whereas it is less concentrated when $\sigma$ is large.

For implementing \autoref{algo_NRJ_andrieu_2013} or the corresponding RJ, we only need to create paths for the proposals $u_{k \mapsto k+1}$ used to switch from model $k$ to model $k+1$. This is realised by looking at the noise term \eqref{eqn_noise}, and also by remembering that it is not necessary to move the parameters that were in model $k$. We thus essentially create a bridge between model $k$ to model $k+1$ made of weighted geometric averages of $f$ and $q_{k\mapsto k+1}$. The annealing intermediate distributions indeed have intuitive forms:
\small\begin{align*}
 \rho_{k\mapsto k+1}^{(t)}(u_{k \mapsto k+1}^{(t)}) &\propto % \left[\exp\left(-\frac{1}{2\sigma^2} \, (u_{k \mapsto k+1}^{(t)})^2\right)\right]^{1-\gamma_t} \left[\exp\left(-\frac{1}{2} \, (u_{k \mapsto k+1}^{(t)})^2\right)\right]^{\gamma_t}=
 \exp\left(-\frac{(u_{k \mapsto k+1}^{(t)})^2}{2} \, [(1-\gamma_t)\sigma^{-2}+\gamma_t]\right),
 % \rho_{k+1\mapsto k}^{(t)}(x_{k+1,k+1}^{(t)})&\propto \left[\exp\left(-\frac{1}{2\sigma^2} \, (x_{k+1,k+1}^{(t)})^2\right)\right]^{\gamma_t} \left[\exp\left(-\frac{1}{2} \, (x_{k+1,k+1}^{(t)})^2\right)\right]^{1-\gamma_t}=\exp\left(-\frac{(x_{k+1,k+1}^{(t)})^2}{2} \, [\gamma_t\sigma^{-2}+(1-\gamma_t)]\right),
\end{align*}\normalsize
where $\gamma_t := t/T$. Therefore, to go from model $k$ to model $k+1$, we target normal distributions with mean 0 and variances $[(1-t/T)\sigma^{-2}+t/T]^{-1}$; we thus start with variances close to $\sigma^2$ (corresponding to the initial proposal distribution) to finish with variances close to 1 (the target distribution). For the reverse move, we do the opposite. As we can sample from the distributions $\rho_{k\mapsto k+1}^{(t)}$, we use them as transitions kernels: $K_{k \mapsto k+1}^{(t)}(u_{k \mapsto k+1}^{(t)}, \cdot \,) := \rho_{k\mapsto k+1}^{(t)}$ (which satisfy the symmetry \eqref{eqn_symmetry} and reversibility \eqref{eqn_reversibility} conditions).

The results are presented in \autoref{fig_results_4_1}. They are based on 1,000 runs of 100,000 iterations for each value of $\sigma$ and $\text{K}_{\max}$; % and $\phi$; the ESS are computed considering only the iterations in which model switches are proposed.
recall that the impact of varying $\phi$ was discussed in \autoref{sec_optimality}. As expected, the further $\sigma$ is from 1 (the latter corresponding to ideal samplers), the lower is the ESS. We notice that the impact is almost symmetric in $\sigma$ if we consider distances from the distribution $\mathcal{N}(0, 1)$ (e.g.\ the normal with $\sigma=1/2$ is two times more concentrated, whereas the normal with $\sigma=2$ is two times less concentrated; they both are at a distance of two). %, but the former is on the more concentrated side and the latter on the less concentrated one).
We also notice that in extreme cases, for instance when $\sigma$ is close to 0, NRJ and RJ have similar performances. This is explained by the fact that the direction-assisted scheme characterising NRJ does not help any more; almost all moves are rejected, which implies that direction changes very often. This leads to the same diffusive behaviour as RJ. Applying Algorithms~\ref{algo_NRJ_andrieu_2013}  and \ref{algo_NRJ_andrieu_2018} improve performances. It is possible to obtain essentially flat lines around the maximum value of 0.21 ESS per iteration by increasing $T$ and $N$, leading to samplers that are at least 2.5 times more efficient than RJ for any value of $\sigma$. Note that we do not show the results for the RJ corresponding to Algorithms~\ref{algo_NRJ_andrieu_2013} and \ref{algo_NRJ_andrieu_2018} as it does not add information to \autoref{fig_results_4_1} (a) given that the lines would be on top of each other.

The ESS also decreases as the total number of models $\text{K}_{\max}$ increases (see \autoref{fig_results_4_1} (b)). This is expected as the difference between $k$ and $k'$ (representing the current model and the next one to explore) is constant, equal to 1. The exploration abilities of the stochastic processes thus diminish as a smaller fraction of the state-space is traversed at each iteration. In theory, a way to compensate is to generate a random variable at each iteration that dictates the difference between $k$ and $k'$, allowing larger jumps. However, as mentioned in \autoref{sec_NRJ}, proposal distributions for transitions to models at a distance of more than 1 are often very difficult to design. % The question is thus whether or not it is possible. And if yes, is it worth it given that it requires the tuning of an additional distribution?
Note that the total probability mass of the 15 most likely models is essentially 1 when $\phi:=2$ (and $\text{K}_{\max}\geq 15$), which explains why the ESS becomes constant beyond this value. Note also that we do not show the results for Algorithms~\ref{algo_NRJ_andrieu_2013} and \ref{algo_NRJ_andrieu_2018} as $\sigma:=1$, meaning that ideal samplers are applied.

\begin{figure}[ht]
\begin{center}
  $\begin{array}{cc}
  \vspace{-0mm} \includegraphics[width=0.4\textwidth]{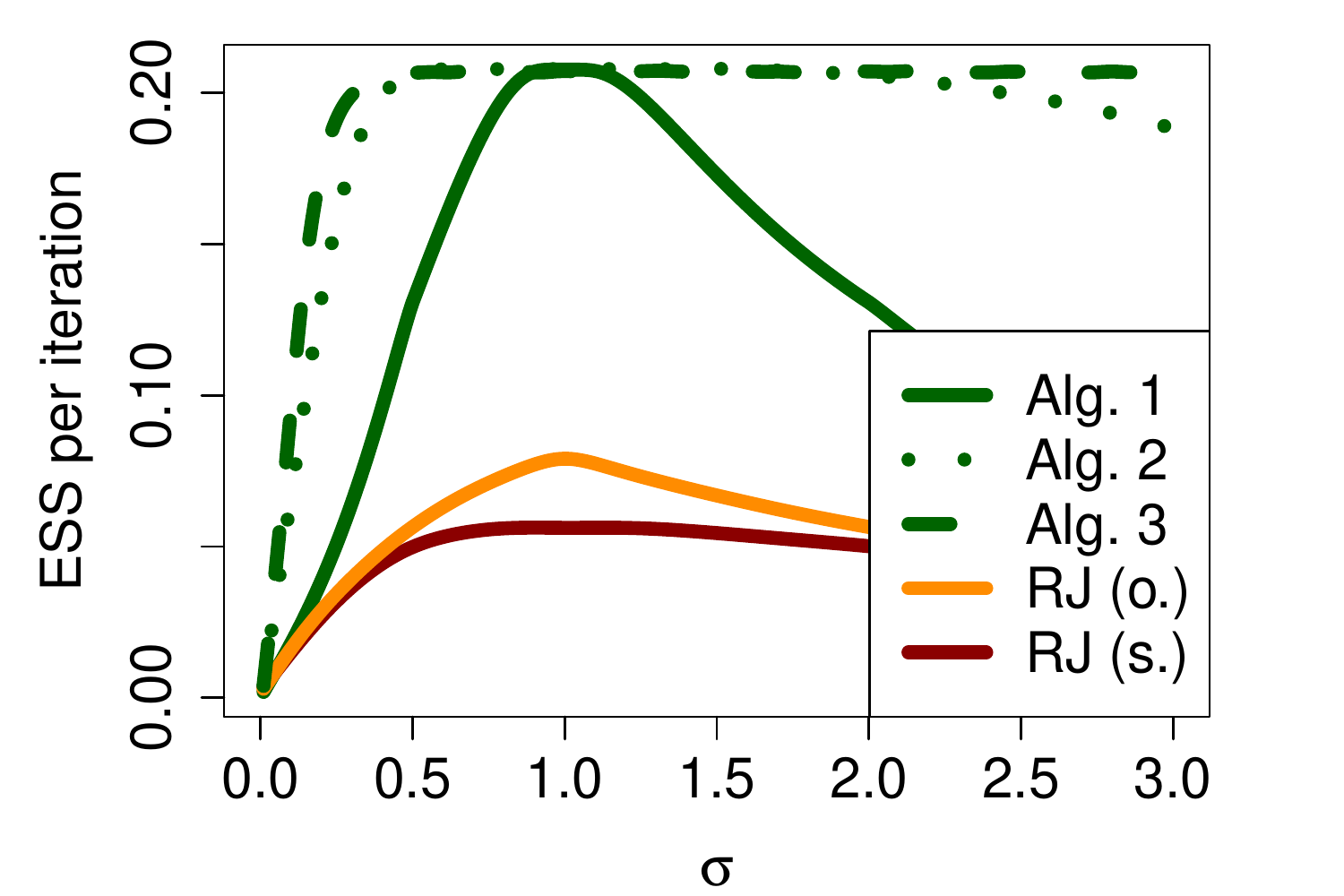} & \includegraphics[width=0.4\textwidth]{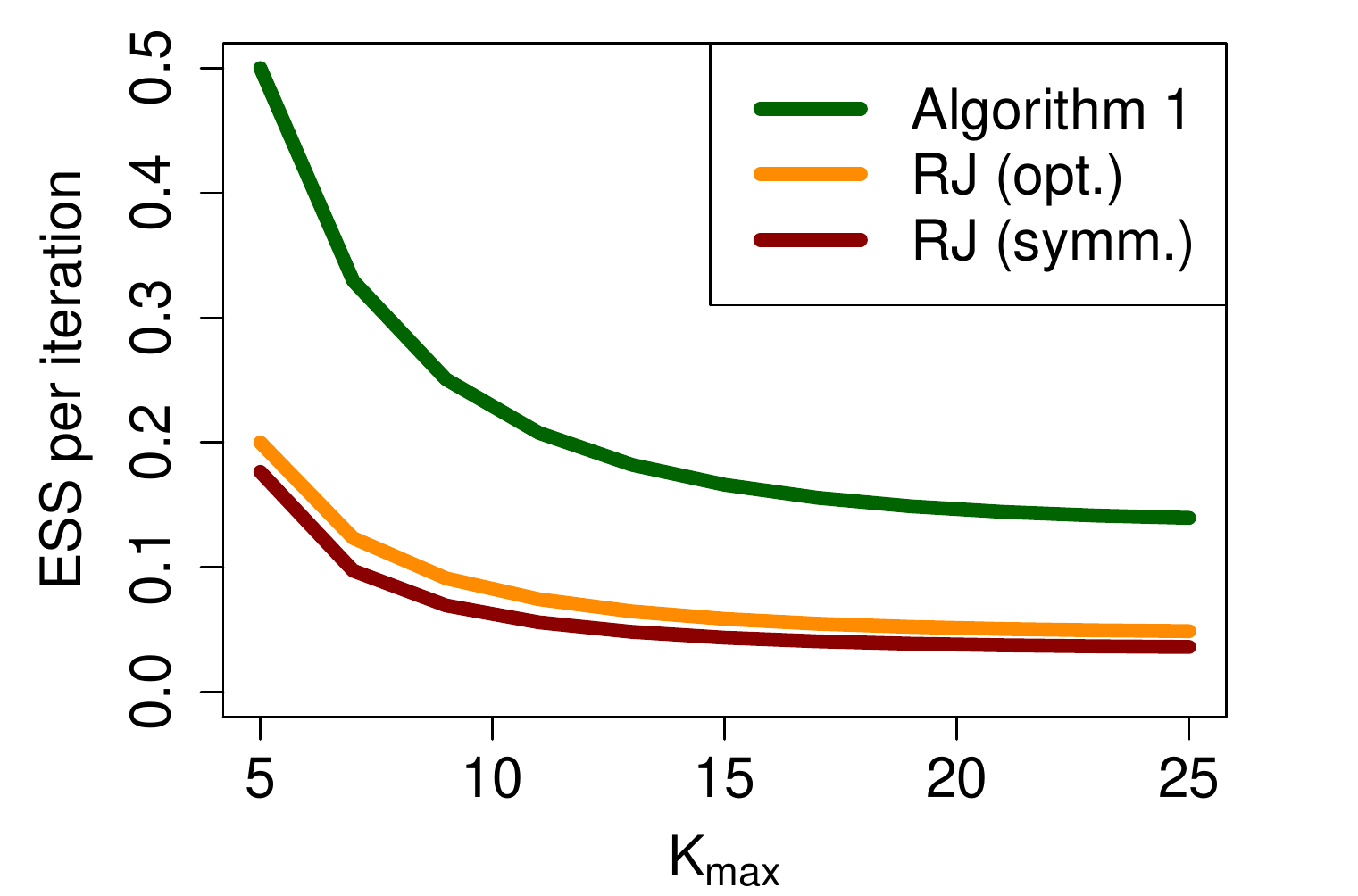} \cr % & \hspace{-5mm} \includegraphics[width=0.35\textwidth]{Fig3_varying_phi.pdf} \cr
  \hspace{4mm} \textbf{(a)} & \hspace{4mm} \textbf{(b)} % & \textbf{(c)}
  \end{array}$
 \end{center}
  \vspace{-8mm}
\caption{\small (a) ESS as a function of $\sigma$ for NRJ (\autoref{algo_NRJ}, \autoref{algo_NRJ_andrieu_2013} with $T := 15$, and \autoref{algo_NRJ_andrieu_2018} with $T := 15$ and $N:=15$) and RJ (with optimal and symmetric $g$), when $\phi:=2$ and $\text{K}_{\max}:=11$; (b) ESS as a function of $\text{K}_{\max}$ for NRJ (\autoref{algo_NRJ}) and RJ (with optimal and symmetric $g$), when $\phi:=2$ and $\sigma := 1$ %; (c) ESS as a function of $\phi$ for NRJ (\autoref{algo_NRJ}) and RJ (with optimal and symmetric $g$), when $\text{K}_{\max}:=11$ and $\sigma := 1$
}\label{fig_results_4_1}
\end{figure}

\subsection{Performance evaluation in multiple change-point problems}\label{sec_changepoint}

In this section, we evaluate the performance of RJ and NRJ algorithms when applied to sample from the posterior of the model in \cite{green1995reversible} for multiple change-point analysis, based on the coal mining disaster data set detailed in \cite{raftery1986bayesian}. The $n$ data points $\mathbf{t} := (t_1, \ldots, t_n)$ represent times of occurrence of disasters. %, to which we add the ``origin'' $t_0 := 0$.
It is assumed that they arose from a non-homogeneous Poisson process that has an intensity given by a step function $\lambda_k$ with $k + 1$ steps, where $k\in\mathcal{K}:=\{0,\ldots, \text{K}_{\max}\}$, $\text{K}_{\max}$ being a known positive integer.

We use the same prior distributions and the same proposals for the RJ and NRJ as \cite{green1995reversible}. For implementing \autoref{algo_NRJ_andrieu_2018} and the corresponding RJ, we proceed as in \cite{karagiannis2013annealed}. All details are provided in the supplementary material (\autoref{sec_changepoint_details}).

The performance of the different algorithms are summarised in \autoref{table_performance}. The results for \autoref{algo_NRJ_andrieu_2018} and the corresponding RJ are based on 1,000 runs with 100,000 iterations and burn-ins of 10,000. To reach the \emph{same computational budget} as these samplers, \autoref{algo_NRJ} and its reversible counterpart are run with an increased number of iterations. The performance of ideal samplers with the same run length as \autoref{algo_NRJ_andrieu_2018} is also presented in \autoref{table_performance} to show the kind of performance that can be achieved. To measure performance, we display ESS per iteration. We also use the relative difference in total variation (TV) with the ideal NRJ: $(\text{TV}(P) - \text{TV}(P_\text{ideal}^{\text{NRJ}})) / \text{TV}(P_\text{ideal}^{\text{NRJ}})$, where $\text{TV}(P)$ is the TV between the model distribution estimated using the Markov kernel $P$ and the posterior model probabilities.\footnote{We used accurate approximations to the posterior model probabilities. We verified that the TV goes to 0 for all algorithms as the number of iterations increases.}

We observe for this multiple change-point example that NRJ samplers always perform better the corresponding RJ samplers at no additional computational cost both in terms of relative TV error and ESS per iteration. Additionally, the displayed relative difference in TV is obtained by running the vanilla samplers for the same amount of compute time as \autoref{algo_NRJ_andrieu_2018} and its corresponding RJ. It is thus clear that the vanilla samplers provide estimates of the marginal posterior model probabilities which are inaccurate. This is consistent with previous experimental results in \cite{karagiannis2013annealed}.  To summarise, this example illustrates that, \emph{at fixed computational complexity}, \autoref{algo_NRJ_andrieu_2018} is an algorithm that can outperform vanilla samplers and the RJ schemes proposed in \cite{karagiannis2013annealed} and \cite{andrieu2018utility}.

%A comparison based on 1,000 runs of RJ and NRJ with 100,000 iterations and burn-ins of 10,000 is presented in \autoref{fig_results_4_2}; the ESS are computed considering only the iterations in which model switches are proposed. Only the results for RJ with symmetric $g$ are shown as those with $g^*$ defined in \eqref{eqn_g_sqrt} are similar. We notice that \autoref{algo_NRJ} and the corresponding RJ offer similar performance (\autoref{fig_results_4_2} (a)). This is due to too many rejected proposals, which induces in NRJ a lot of direction changes and a diffusive behaviour (as explained in \autoref{sec_noise}). The true nature of NRJ characterised by persistent movement is revealed when the proportion of rejections decreases as a by-product of applying the methods of \cite{karagiannis2013annealed} and  \cite{andrieu2018utility} (\autoref{fig_results_4_2} (b)) which allow to get closer to ideal samplers (\autoref{fig_results_4_2} (c)). The ESS per iteration is (on average) 2.1 times larger for \autoref{algo_NRJ_andrieu_2018} than for the corresponding RJ; this ratio is 4.0 for ideal samplers. % To measure the efficiency of ideal samplers, we in this case implemented marginal samplers for $K$, with a target distribution based on accurate approximations of the posterior model probabilities.

\begin{table}[ht]
\footnotesize
\centering
\begin{tabular}{l rr}
\toprule
\textbf{Algorithms} &  \textbf{Rel. diff. in TV} & \textbf{ESS per it.} \cr
\midrule
 Ideal NRJ & ---  &  0.35 \cr
 Ideal RJ    & 0.94 & 0.09  \cr
\midrule
 \autoref{algo_NRJ_andrieu_2018} & 0.94  &  0.15 \cr
 Corresponding RJ & 1.50 & 0.07 \cr
 \midrule
 Vanilla NRJ & 15.76 & 0.02 \cr
 Vanilla RJ & 16.66 & 0.01 \cr
 \bottomrule
\end{tabular}
  \caption{\small Performance of vanilla samplers (i.e.\ \autoref{algo_NRJ} and the corresponding RJ), \autoref{algo_NRJ_andrieu_2018} with $T=100$ and $N=10$ and the corresponding RJ, and ideal samplers} \label{table_performance}
\end{table}

%\begin{figure}[ht]
%\begin{center}
%  $\begin{array}{ccc}
%  \vspace{-2mm}
%  \hspace{-2mm}\includegraphics[width=0.33\textwidth]{4_2_vanilla_100K.pdf} & \hspace{-6mm}\includegraphics[width=0.33\textwidth]{4_2_closer_ideal_100K.pdf} & \hspace{-6mm}\includegraphics[width=0.33\textwidth]{4_2_ideal_100K.pdf} \cr
%  \hspace{-2mm}\textbf{(a) Vanilla samplers} & \hspace{-6mm} \textbf{(b) \autoref{algo_NRJ_andrieu_2018} \& corr. RJ} & \hspace{-6mm} \textbf{(c) Ideal samplers}
%  \end{array}$
% \end{center}
%  \vspace{-5mm}
%\caption{\small ESS for NRJ and RJ (with symmetric $g$) when the samplers are: (a) vanilla samplers (i.e.\ \autoref{algo_NRJ} and the corresponding RJ), (b) \autoref{algo_NRJ_andrieu_2018} with $T:=100$ and $N:=10$ and the corresponding RJ, (c) ideal samplers}\label{fig_results_4_2}
%\end{figure}
%\normalsize

\section{Discussion}\label{sec_discussion}

In this paper, we have introduced non-reversible trans-dimensional samplers that can be applied to Bayesian nested model selection. They are derived from RJ algorithms by making simple modifications which require no additional computational cost during implementation; the model indicator process now follows a direction $\nu$ which is conserved as long as the model switches are accepted, but reversed at the next rejection. Empirically, these samplers outperform their reversible counterparts when the marginal posterior distribution of $K$ is not too concentrated. We now discuss some implementation aspects that have not been addressed in previous sections and possible directions for future research.

\subsection{Other implementation aspects}\label{sec_implementation}

Several functions need to be specified for implementing trans-dimensional samplers: $g$ (which corresponds to the specification of $\tau$ for NRJ), $q_{k\mapsto k'}$ and $T_{k\mapsto k'}$. Significant amount of work has been carried out to address the specification of the last two when no prior information about the problems can be exploited (contrary to the examples in \autoref{sec_numerical}) or a more automatic perspective is adopted (see, e.g., \cite{green2003trans} and \cite{brooks2003efficient}). The approaches of these authors are arguably the most popular. They are directly applicable in the NRJ framework. We believe a particularly good way to proceed is to design the functions $q_{k\mapsto k'}$ and $T_{k\mapsto k'}$ according to the approach of \cite{green2003trans} to afterwards use them in \autoref{algo_NRJ_andrieu_2018} to benefit from the strategies of \cite{karagiannis2013annealed} and \cite{andrieu2018utility} that aim to ensure good mixing properties.

Little attention has been devoted to the impact of the specification of $\tau$. The choice and tuning of this parameter representing the proportion of parameter updates during an algorithm run is a non-trivial problem common to all trans-dimensional samplers whose solution depends on their ability at sampling both the parameters and model indicator. \cite{GAGNON201932} essentially devoted a whole paper on its impact on a specific reversible jump sampler' outputs. For a fixed computational budget, a value closer to 1 leads to more accurate parameter estimation (of the visited models), while a value closer to 0 yields better posterior model probability approximations. Studying more precisely the quantitative impact of $\tau$ on the performance of the samplers is beyond the scope of this paper. However, from a qualitative point of view, $\tau$ has no impact on the order between the asymptotic variances of NRJ and RJ; i.e. \autoref{cor_asymptotic_var} holds whatever being $\tau$.

 \cite{GAGNON201932} prove weak convergence results for RJ under strong assumptions and identify ranges of values for which a suitable balance between a lot of model switches (but few parameter updates) and a lot of parameter updates (but few model switches) is reached. Values around 0.4 are suitable in the situation where $q_{k \mapsto k'}$ and $T_{k \mapsto k'}$ are well designed; otherwise, smaller values should be used. In scenarios in which NRJ is better at sampling $K$ than RJ, it is expected that larger values for $\tau$ than in RJ would be suitable. Further investigations are however required.

\subsection{Possible directions for future research}\label{sec_future}

We identified in \autoref{sec_optimality} a specific ideal RJ (associated with $g^*$) as the main competitor to NRJ within all ideal RJ algorithms when the marginal posterior distribution of $K$ belongs to a family of unimodal PMF and samplers are restricted to model switching proposals of the form $k\mapsto k'\in\{k-1, k+1\}$. We next provided arguments explaining why ideal NRJ outperform this ideal RJ when the target is not too concentrated and numerically showed the range of concentration parameters $\phi$ in the PMF \eqref{eqn_distributions_phi} for which this is the case. It would be interesting to conduct an exhaustive theoretical analysis to expand the scope of the conclusions and make more precise the expected gain.
% As mentioned in \autoref{sec_noise}, it is possibly more efficient to generate a jump size $\epsilon:=|k'-k|$ at each iteration for a proposal for the model to visit next instead of setting $k'=k+\nu$ (which implies that $1=|k'-k|$). Using this approach leads to $k'=k+\epsilon\, \nu$. The random variable $\epsilon$ can be distributed as a Poisson with parameter $\lambda>0$, for instance. As with optimal scaling problems (see, e.g., \cite{roberts1997weak} and \cite{bedard2007weak}), the idea is to prove weak convergence results to find the optimal value for $\lambda$ for a given marginal posterior distribution on $K$ (and hopefully simple rules to identify this optimal value given that we usually do not have access to this distribution). Such results can be used to find optimal values for $\tau$ as well.

It would also be interesting to % take steps towards automatic NRJ, and obviously, to
develop NRJ that can be applied to non-nested models. However, developing efficient non-reversible samplers in such scenarios is much more difficult because, contrary to the nested case, there is no natural order among the models.
% If it possibleIf however this is possible, the approach presented in this paper can be directly applied. We note that if approximations to the posterior model probabilities are available, and if the models can be rearranged using these to obtain (as best as possible) a unimodal symmetric PMF while ensuring that it is possible to make all transitions from models $k$ to models $k' \in \{k - 1, k + 1\}$ in this new ordering, then our approach can be employed and is expected to perform well. We however expect such a situation to happen rarely, justifying further research.}

\section*{Acknowledgements}

The authors thank three anonymous referees for helpful suggestions that led to an improved paper. Philippe Gagnon acknowledges support from FRQNT (Le Fonds de recherche du Qu\'{e}bec - Nature et technologies). Arnaud Doucet was partially supported by the U.S. Army Research Laboratory and the U. S. Army Research Office, and by the U.K. Ministry of Defence (MoD) and the U.K. Engineering and Physical Research Council (EPSRC) under grant number EP/R013616/1. He is also supported by the EPSRC grants EP/R018561/1 and EP/R034710/1.

\bibliographystyle{Chicago}
\bibliography{reference}

\section{Supplementary material}\label{sec_supp}

We present in \autoref{sec_proofs} the proofs of \autoref{prop_invariance}, \autoref{thm_conv_algo2} and \autoref{cor_asymptotic_var} of our paper. In \autoref{sec_theoretical_ana}, weak convergence results for the ideal samplers as the size of the state-space increases are presented. The details about the multiple change-point example of our paper are provided in \autoref{sec_changepoint_details}

\subsection{Proofs}\label{sec_proofs}

\begin{proof}[Proof of \autoref{prop_invariance}]
 It suffices to prove that the probability to reach the state $k', \mathbf{y}_{k'}\in A, \nu'$ in one step is equal to the probability of this state under the target:
 \begin{align}\label{eqn_inv}
  \sum_{k, \nu}\int \pi(k, \mathbf{x}_k) \times (1 / 2) \left(\int_{A} P((k,\mathbf{x}_k, \nu), (k', \mathbf{y}_{k'}, \nu'))\, d\mathbf{y}_{k'} \right) \, d\mathbf{x}_k
  &= \int_A \pi(k', \mathbf{y}_{k'})\times (1 / 2) \, d\mathbf{y}_{k'},
 \end{align}
 where $P$ is the transition kernel. Note that we abuse notation here by denoting the measure $ d\mathbf{y}_{k'}$ on the left-hand side (LHS) given that we in fact use the vector $\mathbf{u}_{k\mapsto k'}$ when switching models, which often do not have the same dimension as $\mathbf{y}_{k'}$.

 We consider two distinct events: a model switch is proposed, that we denote $S$, or a parameter update is proposed (therefore denoted $S^c$). We know that the probabilities of these events are $1 - \tau$ and $\tau$, respectively, regardless of the current state of the Markov chain. We rewrite the LHS of \eqref{eqn_inv} as
 \begin{align}\label{eqn_inv_1}
  &\sum_{k, \nu}\int \pi(k, \mathbf{x}_k) \times (1 / 2) \left(\int_{A} P((k,\mathbf{x}_k, \nu), (k', \mathbf{y}_{k'}, \nu'))\, d\mathbf{y}_{k'} \right) \, d\mathbf{x}_k \cr
  &\quad = \Prob(S)\times (1 / 2) \sum_{k, \nu}\int_{A} \int \pi(k, \mathbf{x}_k) \, P((k,\mathbf{x}_k, \nu), (k', \mathbf{y}_{k'}, \nu')\mid S) \, d\mathbf{x}_k \, d\mathbf{y}_{k'} \cr
  &\qquad + \Prob(S^c) \times (1 / 2)\sum_{k, \nu}\int_{A} \int \pi(k, \mathbf{x}_k) \, P((k,\mathbf{x}_k, \nu), (k', \mathbf{y}_{k'}, \nu')\mid S^c) \, d\mathbf{x}_k \, d\mathbf{y}_{k'},
 \end{align}
 using Fubini's theorem. We analyse the two terms separately. We know that
  \[
   P((k,\mathbf{x}_k, \nu), (k', \mathbf{y}_{k'}, \nu')\mid S^c)=\delta_{(k,\nu)}(k',\nu')\, P_{S^c}(\mathbf{x}_{k'}, \mathbf{y}_{k'}),
  \]
  where $P_{S^c}$ is the transition kernel associated with the method used to update the parameters. Therefore, the second term on the right-hand side (RHS) of \eqref{eqn_inv_1} is equal to
  \begin{align*}
   &\Prob(S^c) \times (1 / 2)\sum_{k, \nu}\int_{A} \int \pi(k, \mathbf{x}_k) \, P((k,\mathbf{x}_k, \nu), (k', \mathbf{y}_{k'}, \nu')\mid S^c) \, d\mathbf{x}_k \, d\mathbf{y}_{k'} \cr
   &\quad =\Prob(S^c) \, \pi(k') \times (1 / 2)\int \pi(\mathbf{x}_{k'} \mid k') \left(\int_{A} P_{S^c}(\mathbf{x}_{k'}, \mathbf{y}_{k'}) \, d\mathbf{y}_{k'} \right) \, d\mathbf{x}_{k'}.
  \end{align*}
  We also know that $P_{S^c}$ leaves the conditional distribution $\pi(\, \cdot \mid k')$ invariant, implying that
  \begin{align}\label{eqn_conclusion_up}
   &\Prob(S^c) \, \pi(k') \times (1 / 2)\int \pi(\mathbf{x}_{k'} \mid k') \left(\int_{A} P_{S^c}(\mathbf{x}_{k'}, \mathbf{y}_{k'}) \, d\mathbf{y}_{k'} \right) \, d\mathbf{x}_{k'} \cr
   &\quad = \Prob(S^c) \, \pi(k') \times (1 / 2)\int_A  \pi(\mathbf{y}_{k'} \mid k') \, d\mathbf{y}_{k'}= \Prob(S^c) \int_A \pi(k', \mathbf{y}_{k'})\times (1 / 2) \, d\mathbf{y}_{k'}.
  \end{align}

  For the model switching case (the first term on the RHS of \eqref{eqn_inv_1}), we use the fact that there is a connection between $P((k,\mathbf{x}_k, \nu), (k', \mathbf{y}_{k'}, \nu')\mid S)$ and the kernel associated to a specific RJ. Consider that $g(k,k+1)=g(k,k-1)$ for all $k$ and that all other proposal distributions are the same as the NRJ. In this case, $\alpha_{\text{RJ}}=\alpha_{\text{NRJ}}$. Given the reversibility of RJ, the probability to go from model $k$ with parameters in $B$ to model $k+1$ with parameters in $A$ is
\begin{align}\label{eqn_reversibility_proof}
 &\Prob(S) \int_B \pi(k, \mathbf{x}_k) \left(\int_A P_{\text{RJ}}((k, \mathbf{x}_k),(k+1,\mathbf{y}_{k+1})\mid S) \, d\mathbf{y}_{k+1} \right) \, d\mathbf{x}_{k} \cr
 &\qquad = \Prob(S) \int_A \pi(k+1, \mathbf{y}_{k+1}) \left(\int_B P_{\text{RJ}}((k+1, \mathbf{y}_{k+1}),(k,\mathbf{x}_{k})\mid S) \, d\mathbf{x}_{k} \right) \, d\mathbf{y}_{k+1},
\end{align}   	
where $P_{\text{RJ}}$ is the transition kernel associated with the RJ. Note that
\[
P_{\text{RJ}}((k, \mathbf{x}_k),(k+1,\mathbf{y}_{k+1})\mid S)=(1/2)\, P((k,\mathbf{x}_k, 1), (k+1, \mathbf{y}_{k+1}, 1)\mid S),
\]
given that the difference between both kernels is that in RJ, once it is decided that a model switch is attempted, there is an additional probability of $1/2$ of trying model $k+1$. Analogously, $P_{\text{RJ}}((k+1,\mathbf{y}_{k+1}), (k, \mathbf{x}_k)\mid S)=(1/2)\, P((k+1,\mathbf{y}_{k+1}, -1), (k, \mathbf{x}_{k}, -1)\mid S)$. Using that and taking $B$ equals the whole parameter (and auxiliary) space in \eqref{eqn_reversibility_proof}, we have
\begin{align*}
 &\Prob(S)\int \pi(k, \mathbf{x}_k)\times(1/2) \left(\int_A P((k,\mathbf{x}_k, 1), (k+1, \mathbf{y}_{k+1}, 1)\mid S) \, d\mathbf{y}_{k+1} \right) \, d\mathbf{x}_{k} \cr
 &\quad = \Prob(S)\int_A \pi(k+1, \mathbf{y}_{k+1})\times(1/2) \left(\int P((k+1,\mathbf{y}_{k+1}, -1), (k, \mathbf{x}_{k}, -1)\mid S) \, d\mathbf{x}_{k} \right) \, d\mathbf{y}_{k+1}.
\end{align*}
We thus analyse the probability to reach $k+1$ with parameters in $A$ and direction $+1$. We know that the only other way of reaching this state (other than coming from $k$) is by being at $k+1$ with parameters in $A$ and direction $-1$ and rejecting, which probability is
\[
 \Prob(S)\int_A \pi(k+1, \mathbf{y}_{k+1})\times(1/2) \left(1 - \int P((k+1,\mathbf{y}_{k+1}, -1), (k, \mathbf{x}_{k}, -1)\mid S) \, d\mathbf{x}_{k} \right) \, d\mathbf{y}_{k+1}.
\]
Therefore, the total probability to reach $k+1$ with parameters in $A$ and direction $+1$ in one step (given that a model switch is proposed) is
\begin{align*}
 &\Prob(S)\int \pi(k, \mathbf{x}_k)\times(1/2) \left(\int_A P((k,\mathbf{x}_k, 1), (k+1, \mathbf{y}_{k+1}, 1)\mid S) \, d\mathbf{y}_{k+1} \right) \, d\mathbf{x}_{k} \cr
 & + \Prob(S)\int_A \pi(k+1, \mathbf{y}_{k+1})\times(1/2) \left(1 - \int P((k+1,\mathbf{y}_{k+1}, -1), (k, \mathbf{x}_{k}, -1)\mid S) \, d\mathbf{x}_{k} \right) \, d\mathbf{y}_{k+1} \cr
 &\qquad = \Prob(S)\int_A \pi(k+1, \mathbf{y}_{k+1})\times(1/2) \, d\mathbf{y}_{k+1}.
\end{align*}
Combining this with \eqref{eqn_conclusion_up} allows to conclude the proof.
\end{proof}

\begin{proof}[Proof of \autoref{thm_conv_algo2}]

We show that Algorithm 2 converges towards its ideal version as $T\longrightarrow\infty$. As mentioned, for the ideal version, we consider the case where $q_{k\mapsto k'}:=\pi(\, \cdot\mid k')$, the conditional distribution of the parameters of model $k'$. In this case, we set $\mathbf{y}_{k'}:=\mathbf{u}_{k\mapsto k'}$ to be the proposal for the parameters of model $k'$, and thus the function $T_{k\mapsto k'}$ to be the identity function.

To show the convergence, we use Theorem 1 in \cite{karr1975weak}. We thus have to verify three assumptions, and this will allow to conclude that $\{(K,\mathbf{X}_K,\nu)_T(m): m\in\na\}\Longrightarrow \{(K,\mathbf{X}_K,\nu)_{\text{ideal}}(m): m\in\na\}$  as $T\longrightarrow \infty$. We focus on the movements involving model switches as the same parameter update schemes are used in both samplers. Here are the three assumptions.

\textbf{1.} The distributions that are used to initialise Algorithm 2 converge towards that used to initialise the ideal NRJ.
	
	This is verified as we assume that the Markov chains produced by both Algorithm 2 and its ideal counterpart start at stationarity, i.e.\ $(K,\mathbf{X}_K, \nu)_T(0)\sim \pi\otimes \mathcal{U}\{-1, 1\}$ and $(K,\mathbf{X}_K, \nu)_{\text{ideal}}(0)\sim \pi \otimes \mathcal{U}\{-1, 1\}$.
	
\textbf{2.} For $h\in \bar{\mathcal{C}}^*$ (the space of bounded uniformly continuous functions), we have that
\begin{align*}
	 P_{\text{ideal}}(k,\mathbf{x}_k, \nu)h:=\sum_{k',\nu'}\int h(k', \mathbf{y}_{k'}, \nu') P_{\text{ideal}}((k,\mathbf{x}_k, \nu),(k',\mathbf{y}_{k'}, \nu')) \, d\mathbf{y}_{k'}
	\end{align*}
 is a bounded continuous function.

 This kernel is such that
	\begin{align*}
	 P_{\text{ideal}}(k,\mathbf{x}_k, \nu)h&= \left(1\wedge \frac{\pi(k+\nu)}{\pi(k)}\right)\int h(k+\nu, \mathbf{u}_{k\mapsto k+\nu}, \nu) \, \pi(\mathbf{u}_{k\mapsto k+\nu}\mid k+\nu) \, d\mathbf{u}_{k\mapsto k+\nu} \cr
	  &\qquad + h(k, \mathbf{x}_k, -\nu)\left(1-1\wedge \frac{\pi(k+\nu)}{\pi(k)}\right),
	\end{align*}
	which is bounded and continuous.
	
\textbf{3.} For every $h\in \bar{\mathcal{C}}^*$, the Markov kernel associated with Algorithm 2 $P_Th$ converges towards $P_{\text{ideal}}h$ uniformly on each compact subset of the state-space as $T\longrightarrow\infty$.

 We first show the pointwise convergence. Let us denote the conditional joint density of all the random variables involved in the proposal $\mathbf{y}_{k+\nu}^{(T-1)}$ given $(k,\mathbf{x}_{k}, \nu)$ by
 \small
	\begin{align*}
	 &q(\mathbf{u}_{k \mapsto k+\nu}^{(0)},\mathbf{y}_{k+\nu}^{(1:T-1)},\mathbf{u}_{k+\nu \mapsto k}^{(1:T-1)}):= q_{k\mapsto k+\nu}(\mathbf{u}_{k \mapsto k+\nu}^{(0)})  \prod_{t=1}^{T-1} K_{k\mapsto k+\nu}^{(t)}((\mathbf{y}_{k+\nu}^{(t-1)},\mathbf{u}_{k+\nu \mapsto k}^{(t-1)}),(\mathbf{y}_{k+\nu}^{(t)},\mathbf{u}_{k+\nu \mapsto k}^{(t)})) ,
	\end{align*}
\normalsize
	where $K_{k\mapsto k+\nu}^{(t)}$ is a MH kernel reversible with respect to $\rho_{k\mapsto k+\nu}^{(t)}$. We have that
	\begin{align*}
	 &P_Th(k,\mathbf{x}_k, \nu) :=\int h(k+\nu, \mathbf{y}_{k+\nu}^{(T-1)},\nu) \,  q(\mathbf{u}_{k \mapsto k+\nu}^{(0)},\mathbf{y}_{k+\nu}^{(1:T-1)},\mathbf{u}_{k+\nu \mapsto k}^{(1:T-1)}) \cr
& \hspace{50mm} \times \alpha_{\text{NRJ2}}((k,\mathbf{x}_{k}),(k+\nu,\mathbf{y}_{k+\nu}^{(T-1)})) \, d(\mathbf{u}_{k \mapsto k+\nu}^{(0)},\mathbf{y}_{k+\nu}^{(1:T-1)},\mathbf{u}_{k+\nu \mapsto k}^{(1:T-1)}) \cr
	 & \qquad +h(k, \mathbf{x}_k, -\nu)\int q(\mathbf{u}_{k \mapsto k+\nu}^{(0)},\mathbf{y}_{k+\nu}^{(1:T-1)},\mathbf{u}_{k+\nu \mapsto k}^{(1:T-1)}) \cr
& \hspace{40mm} \times (1-\alpha_{\text{NRJ2}}((k,\mathbf{x}_{k}),(k+\nu,\mathbf{y}_{k+\nu}^{(T-1)}))) \,d(\mathbf{u}_{k \mapsto k+\nu}^{(0)},\mathbf{y}_{k+\nu}^{(1:T-1)},\mathbf{u}_{k+\nu \mapsto k}^{(1:T-1)}).
	\end{align*}
	Using the triangle inequality, we thus have that
	\begin{align}\label{eqn_diff_kernels}
	 &|P_Th(k,\mathbf{x}_k, \nu)-P_{\text{ideal}}h(k,\mathbf{x}_k, \nu)| \cr
	 &\leq \left|\int h(k+\nu, \mathbf{y}_{k+\nu}^{(T-1)},\nu) \,  q(\mathbf{u}_{k \mapsto k+\nu}^{(0)},\mathbf{y}_{k+\nu}^{(1:T-1)},\mathbf{u}_{k+\nu \mapsto k}^{(1:T-1)}) \, \right. \cr
&\hspace{40mm} \left. \times \alpha_{\text{NRJ2}}((k,\mathbf{x}_{k}),(k+\nu,\mathbf{y}_{k+\nu}^{(T-1)})) \, d(\mathbf{u}_{k \mapsto k+\nu}^{(0)},\mathbf{y}_{k+\nu}^{(1:T-1)},\mathbf{u}_{k+\nu \mapsto k}^{(1:T-1)})\right. \cr
	 &\qquad\left. -\int h(k+\nu, \mathbf{u}_{k\mapsto k+\nu}, \nu) \, \pi(\mathbf{u}_{k\mapsto k+\nu}\mid k+\nu) \left(1\wedge \frac{\pi(k+\nu)}{\pi(k)}\right) \, d\mathbf{u}_{k\mapsto k+\nu}\right| \cr
	 &+\left|h(k, \mathbf{x}_k, -\nu)\int q(\mathbf{u}_{k \mapsto k+\nu}^{(0)},\mathbf{y}_{k+\nu}^{(1:T-1)},\mathbf{u}_{k+\nu \mapsto k}^{(1:T-1)}) \right. \cr &\hspace{40mm} \left. \times (1-\alpha_{\text{NRJ2}}((k,\mathbf{x}_{k}),(k+\nu,\mathbf{y}_{k+\nu}^{(T-1)}))) \,d(\mathbf{u}_{k \mapsto k+\nu}^{(0)},\mathbf{y}_{k+\nu}^{(1:T-1)},\mathbf{u}_{k+\nu \mapsto k}^{(1:T-1)}) \right. \cr
	 &\qquad \left.- h(k, \mathbf{x}_k, -\nu)\left(1-1\wedge \frac{\pi(k+\nu)}{\pi(k)}\right)\right|.	
	\end{align}	
	We analyse the first absolute value on the RHS. We write the integrals as (conditional) expectations (given $(k,\mathbf{x}_{k}, \nu)$):
	\begin{align*}
	 &\left|\E\left[h(k+\nu, \mathbf{Y}_{k+\nu}^{(T-1)},\nu) \,\alpha_{\text{NRJ2}}((k,\mathbf{x}_{k}),(k+\nu,\mathbf{Y}_{k+\nu}^{(T-1)}))\right] \right. \cr
&\hspace{40mm} \left. - \E\left[ h(k+\nu, \mathbf{U}_{k\mapsto k+\nu}, \nu)\left(1\wedge \frac{\pi(k+\nu)}{\pi(k)}\right) \right]\right| \cr
	 &\quad\leq \left|\E\left[h(k+\nu, \mathbf{Y}_{k+\nu}^{(T-1)},\nu) \,\alpha_{\text{NRJ2}}((k,\mathbf{x}_{k}),(k+\nu,\mathbf{Y}_{k+\nu}^{(T-1)}))\right] \right. \cr
 & \hspace{40mm} \left. -\E \left[h(k+\nu, \mathbf{Y}_{k+\nu}^{(T-1)},\nu) \left(1\wedge \frac{\pi(k+\nu)}{\pi(k)}\right)\right]\right| \cr
	 &\quad+\left|\left(1\wedge \frac{\pi(k+\nu)}{\pi(k)}\right)\E \left[h(k+\nu, \mathbf{Y}_{k+\nu}^{(T-1)},\nu) \right] - \left(1\wedge \frac{\pi(k+\nu)}{\pi(k)}\right)\E\left[ h(k+\nu, \mathbf{U}_{k\mapsto k+\nu}, \nu) \right]\right|,
	\end{align*}
	using again the triangle inequality. We now show that both absolute values on the RHS converge towards 0. For the first one, we have
	\begin{align*}
	 &\left|\E\left[h(k+\nu, \mathbf{Y}_{k+\nu}^{(T-1)},\nu) \,\alpha_{\text{NRJ2}}((k,\mathbf{x}_{k}),(k+\nu,\mathbf{Y}_{k+\nu}^{(T-1)}))\right] \right. \cr
&\hspace{40mm} \left. -\E \left[h(k+\nu, \mathbf{Y}_{k+\nu}^{(T-1)},\nu) \left(1\wedge \frac{\pi(k+\nu)}{\pi(k)}\right)\right]\right| \cr
	 &\qquad \leq M \, \E\left|\alpha_{\text{NRJ2}}((k,\mathbf{x}_{k}),(k+\nu,\mathbf{Y}_{k+\nu}^{(T-1)})) - \left(1\wedge \frac{\pi(k+\nu)}{\pi(k)}\right)\right| \longrightarrow 0,
	\end{align*}
	using that there exists a positive constant $M$ such that $\abs{h}\leq M$ and that $r_{\text{NRJ2}}((k,\mathbf{x}_{k}),(k+\nu,\mathbf{Y}_{k+\nu}^{(T-1)}))\longrightarrow \pi(k+\nu)/\pi(k)$ in distribution (by assumption). The convergence of the expectation follows from the fact that if a random variable $X_n$ converges towards a constant $c$ in distribution, then $X_n-c$ converges towards 0 in probability and $\E|g(X_n)-g(c)|\longrightarrow 0$ for any bounded uniformly continuous function $g$ ($\min(1, x)$ with $x\geq 0$ is a function having these characteristics). For the second absolute value, we have
 \begin{align*}
 &\left(1\wedge \frac{\pi(k+\nu)}{\pi(k)}\right)\left|\E\left[h(k+\nu, \mathbf{Y}_{k+\nu}^{(T-1)},\nu) \right] - \E\left[ h(k+\nu, \mathbf{U}_{k\mapsto k+\nu}, \nu) \right]\right| \cr
 &\qquad \leq \left|\E\left[h(k+\nu, \mathbf{Y}_{k+\nu}^{(T-1)},\nu) \right] - \E\left[ h(k+\nu, \mathbf{U}_{k\mapsto k+\nu}, \nu) \right]\right| \longrightarrow 0,
 \end{align*}	
 if the (conditional) distribution of $\mathbf{Y}_{k+\nu}^{(T-1)}$ (given $(k,\mathbf{x}_{k}, \nu)$) converges towards $\pi(\, \cdot\mid k+\nu)$ given that $h$ is a bounded continuous function.

 Let us now prove this convergence in distribution. The conditional distribution of $\mathbf{Y}_{k+\nu}^{(T-1)}$ given $(k,\mathbf{x}_{k}, \nu)$ is written as
 \begin{align*}
  &\Prob(\mathbf{Y}_{k+\nu}^{(T-1)}\in A \mid k,\mathbf{x}_{k}, \nu) :=\int_{\mathbf{y}_{k+\nu}^{(T-1)}\in A} q_{k\mapsto k+\nu}(\mathbf{u}_{k \mapsto k+\nu}^{(0)})  \cr
  &\hspace{20mm}\times \prod_{t=1}^{T-1} K_{k\mapsto k+\nu}^{(t)}((\mathbf{y}_{k+\nu}^{(t-1)},\mathbf{u}_{k+\nu \mapsto k}^{(t-1)}),(\mathbf{y}_{k+\nu}^{(t)},\mathbf{u}_{k+\nu \mapsto k}^{(t)})) \,d(\mathbf{u}_{k \mapsto k+\nu}^{(0)},\mathbf{y}_{k+\nu}^{(1:T-1)},\mathbf{u}_{k+\nu \mapsto k}^{(1:T-1)}) \cr
  &\quad=\int_{\mathbf{y}_{k+\nu}^{(T-1)}\in A} q_{k\mapsto k+\nu}(\mathbf{u}_{k \mapsto k+\nu}^{(0)})  \prod_{t=1}^{t^*-1} K_{k\mapsto k+\nu}^{(t)}((\mathbf{y}_{k+\nu}^{(t-1)},\mathbf{u}_{k+\nu \mapsto k}^{(t-1)}),(\mathbf{y}_{k+\nu}^{(t)},\mathbf{u}_{k+\nu \mapsto k}^{(t)})) \cr
 &\qquad \times \prod_{t=t^*}^{T-1} K_{k\mapsto k+\nu}^{(t)}((\mathbf{y}_{k+\nu}^{(t-1)},\mathbf{u}_{k+\nu \mapsto k}^{(t-1)}),(\mathbf{y}_{k+\nu}^{(t)},\mathbf{u}_{k+\nu \mapsto k}^{(t)})) \,d(\mathbf{u}_{k \mapsto k+\nu}^{(0)},\mathbf{y}_{k+\nu}^{(1:T-1)},\mathbf{u}_{k+\nu \mapsto k}^{(1:T-1)}).
 \end{align*}
 Under \autoref{ass_thm_conv_algo2_3}, one can show that $t^*$ and $T$ can be chosen such that $(T-t^*)/T$ is small and
 \begin{align*}
   \abs{K_{k\mapsto k+\nu}^{(t)}((\mathbf{y}_{k+\nu}^{(t-1)},\mathbf{u}_{k+\nu \mapsto k}^{(t-1)}),(\mathbf{y}_{k+\nu}^{(t)},\mathbf{u}_{k+\nu \mapsto k}^{(t)})) - K_{k\mapsto k+\nu}^{(T)}((\mathbf{y}_{k+\nu}^{(t-1)},\mathbf{u}_{k+\nu \mapsto k}^{(t-1)}),(\mathbf{y}_{k+\nu}^{(t)},\mathbf{u}_{k+\nu \mapsto k}^{(t)}))}< \frac{1}{T-t^*} \,\epsilon,
 \end{align*}
 for all $t\geq t^*$ and any $\epsilon>0$, where $K_{k\mapsto k+\nu}^{(T)}$ is the MH kernel for which $\rho_{k\mapsto k+\nu}^{(T)}:=\pi(\, \cdot\mid k+\nu)\otimes q_{k+\nu \mapsto k}$ is used instead in the acceptance probability. One can thus show that
 \begin{align*}
   \abs{\prod_{t=t^*}^{T-1} K_{k\mapsto k+\nu}^{(t)}((\mathbf{y}_{k+\nu}^{(t-1)},\mathbf{u}_{k+\nu \mapsto k}^{(t-1)}),(\mathbf{y}_{k+\nu}^{(t)},\mathbf{u}_{k+\nu \mapsto k}^{(t)})) - \prod_{t=t^*}^{T-1} K_{k\mapsto k+\nu}^{(T)}((\mathbf{y}_{k+\nu}^{(t-1)},\mathbf{u}_{k+\nu \mapsto k}^{(t-1)}),(\mathbf{y}_{k+\nu}^{(t)},\mathbf{u}_{k+\nu \mapsto k}^{(t)}))}< \epsilon,
 \end{align*}
 and therefore,
\begin{align*}
 &\left| q_{k\mapsto k+\nu}(\mathbf{u}_{k \mapsto k+\nu}^{(0)})  \prod_{t=1}^{t^*-1} K_{k\mapsto k+\nu}^{(t)}((\mathbf{y}_{k+\nu}^{(t-1)},\mathbf{u}_{k+\nu \mapsto k}^{(t-1)}),(\mathbf{y}_{k+\nu}^{(t)},\mathbf{u}_{k+\nu \mapsto k}^{(t)})) \right. \cr
  &\hspace{40mm}\left. \times \prod_{t=t^*}^{T-1} K_{k\mapsto k+\nu}^{(t)}((\mathbf{y}_{k+\nu}^{(t-1)},\mathbf{u}_{k+\nu \mapsto k}^{(t-1)}),(\mathbf{y}_{k+\nu}^{(t)},\mathbf{u}_{k+\nu \mapsto k}^{(t)})) \right. \cr
 &\left. -q_{k\mapsto k+\nu}(\mathbf{u}_{k \mapsto k+\nu}^{(0)})  \prod_{t=1}^{t^*-1} K_{k\mapsto k+\nu}^{(t)}((\mathbf{y}_{k+\nu}^{(t-1)},\mathbf{u}_{k+\nu \mapsto k}^{(t-1)}),(\mathbf{y}_{k+\nu}^{(t)},\mathbf{u}_{k+\nu \mapsto k}^{(t)})) \right. \cr
  &\hspace{40mm}\left. \times \prod_{t=t^*}^{T-1} K_{k\mapsto k+\nu}^{(T)}((\mathbf{y}_{k+\nu}^{(t-1)},\mathbf{u}_{k+\nu \mapsto k}^{(t-1)}),(\mathbf{y}_{k+\nu}^{(t)},\mathbf{u}_{k+\nu \mapsto k}^{(t)}))\right| \cr
 & \qquad <\epsilon.
\end{align*}
We have that the integral of the two functions in the absolute value converges towards 0 as well as a result of Scheff\'{e}'s lemma (see \cite{scheffe1947useful}):
\footnotesize
 \begin{align}\label{eqn_conv_stationary}
  & \left| \Prob(\mathbf{Y}_{k+\nu}^{(T-1)}\in A \mid k,\mathbf{x}_{k}, \nu)  -  \int_{\mathbf{y}_{k+\nu}^{(T-1)}\in A} q_{k\mapsto k+\nu}(\mathbf{u}_{k \mapsto k+\nu}^{(0)}) \prod_{t=1}^{t^*-1} K_{k\mapsto k+\nu}^{(t)}((\mathbf{y}_{k+\nu}^{(t-1)},\mathbf{u}_{k+\nu \mapsto k}^{(t-1)}),(\mathbf{y}_{k+\nu}^{(t)},\mathbf{u}_{k+\nu \mapsto k}^{(t)})) \right.\cr
 & \left. \hspace{30mm}  \times \prod_{t=t^*}^{T-1} K_{k\mapsto k+\nu}^{(T)}((\mathbf{y}_{k+\nu}^{(t-1)},\mathbf{u}_{k+\nu \mapsto k}^{(t-1)}),(\mathbf{y}_{k+\nu}^{(t)},\mathbf{u}_{k+\nu \mapsto k}^{(t)})) \,d(\mathbf{u}_{k \mapsto k+\nu}^{(0)},\mathbf{y}_{k+\nu}^{(1:T-1)},\mathbf{u}_{k+\nu \mapsto k}^{(1:T-1)})\right|<\epsilon.
 \end{align}
 \normalsize
 We also have that
 \begin{align}\label{eqn_ergodicity}
  &\int_{\mathbf{y}_{k+\nu}^{(T-1)}\in A} \prod_{t=t^*}^{T-1} K_{k\mapsto k+\nu}^{(T)}((\mathbf{y}_{k+\nu}^{(t-1)},\mathbf{u}_{k+\nu \mapsto k}^{(t-1)}),(\mathbf{y}_{k+\nu}^{(t)},\mathbf{u}_{k+\nu \mapsto k}^{(t)})) \,d(\mathbf{y}_{k+\nu}^{(t^*:T-1)},\mathbf{u}_{k+\nu \mapsto k}^{(t^*:T-1)}) \cr
  &\quad\leq \left|\int_{\mathbf{y}_{k+\nu}^{(T-1)}\in A} \prod_{t=t^*}^{T-1} K_{k\mapsto k+\nu}^{(T)}((\mathbf{y}_{k+\nu}^{(t-1)},\mathbf{u}_{k+\nu \mapsto k}^{(t-1)}),(\mathbf{y}_{k+\nu}^{(t)},\mathbf{u}_{k+\nu \mapsto k}^{(t)})) \,d(\mathbf{y}_{k+\nu}^{(t^*:T-1)},\mathbf{u}_{k+\nu \mapsto k}^{(t^*:T-1)}) \right. \cr
  &\qquad \left. - \Prob_{\rho_{k\mapsto k+\nu}^{(T)}}(\mathbf{Y}_{k+\nu}^{(T-1)}\in A)\right| +\Prob_{\rho_{k\mapsto k+\nu}^{(T)}}(\mathbf{Y}_{k+\nu}^{(T-1)}\in A),
 \end{align}
 where $\Prob_{\rho_{k\mapsto k+\nu}^{(T)}}$ is the probability measure using the density $\rho_{k\mapsto k+\nu}^{(T)}$. We choose $t^*$ and $T$ such that the absolute value above is smaller than $\epsilon$ which does not depend on $(\mathbf{y}_{k+\nu}^{(t^*-1)},\mathbf{u}_{k+\nu \mapsto k}^{(t^*-1)})$. This is possible given that the time-homogeneous $\pi(\, \cdot\mid k+\nu)\otimes q_{k+\nu \mapsto k}$-reversible Markov chain associated with the proposal distribution $q_{\text{NRJ2}}^{k,\nu}$, $\{(\mathbf{Y}_{k+\nu}, \mathbf{U}_{k+\nu \mapsto k})(m): m \in\na\}$, is uniformly ergodic (by assumption). This yields the convergence of the (conditional) distribution of $\mathbf{Y}_{k+\nu}^{(T-1)}$ (given $(k,\mathbf{x}_{k}, \nu)$) towards $\pi(\, \cdot\mid k+\nu)$.

 It is proved that the second absolute value in \eqref{eqn_diff_kernels} converges towards 0 using the same arguments, which allows to establish the pointwise convergence $P_Th(k,\mathbf{x}_k, \nu)\longrightarrow P_{\text{ideal}}h(k,\mathbf{x}_k, \nu)$. The uniform convergence on each compact subset of the state-space follows from the uniform ergodicity of the MH kernels.
 \end{proof}

We now highlight what modifications and which additional technical conditions are required if geometric ergodicity is instead assumed. The absolute value on the RHS in \eqref{eqn_ergodicity} is in this case bounded above by $M(\mathbf{y}_{k+\nu}^{(t^*-1)},\mathbf{u}_{k+\nu \mapsto k}^{(t^*-1)}) \, \rho^{T-1-t^*}$, where $M(\mathbf{y}_{k+\nu}^{(t^*-1)},\mathbf{u}_{k+\nu \mapsto k}^{(t^*-1)})$ is finite for all $(\mathbf{y}_{k+\nu}^{(t^*-1)},\mathbf{u}_{k+\nu \mapsto k}^{(t^*-1)})$ and $\rho<1$. If the following integral is finite
\begin{align*}
 &\int q_{k\mapsto k+\nu}(\mathbf{u}_{k \mapsto k+\nu}^{(0)})  \prod_{t=1}^{t^*-1} K_{k\mapsto k+\nu}^{(t)}((\mathbf{y}_{k+\nu}^{(t-1)},\mathbf{u}_{k+\nu \mapsto k}^{(t-1)}),(\mathbf{y}_{k+\nu}^{(t)},\mathbf{u}_{k+\nu \mapsto k}^{(t)})) \, \cr
 &\hspace{50mm} \times M(\mathbf{y}_{k+\nu}^{(t^*-1)},\mathbf{u}_{k+\nu \mapsto k}^{(t^*-1)}) \,d(\mathbf{u}_{k \mapsto k+\nu}^{(0)},\mathbf{y}_{k+\nu}^{(1:t^*-1)},\mathbf{u}_{k+\nu \mapsto k}^{(1:t^*-1)}),
\end{align*}
then we know that we have the same conclusion as above, i.e.\ we can choose $t^*$ and $T$ such that the  absolute value on the RHS in \eqref{eqn_ergodicity} is smaller than $\epsilon$. That integral shall be finite when the process associated with the kernels $K_{k\mapsto k+\nu}^{(t)}$ do not reach states $(\mathbf{y}_{k+\nu}^{(t^*-1)},\mathbf{u}_{k+\nu \mapsto k}^{(t^*-1)})$ such that $M(\mathbf{y}_{k+\nu}^{(t^*-1)},\mathbf{u}_{k+\nu \mapsto k}^{(t^*-1)})$ is extremely large (or at least if it does, it is with small enough probability).

This condition thus suffices to show the pointwise convergence $P_Th(k,\mathbf{x}_k, \nu)\longrightarrow P_{\text{ideal}}h(k,\mathbf{x}_k, \nu)$. To establish the uniform convergence under geometric ergodicity, we use the same strategy as that applied to show \eqref{eqn_conv_stationary}. We can choose $t^*$ and $T$ such that the first $t^*$ steps (after having generated $\mathbf{u}_{k \mapsto k+\nu}^{(0)}$) with density $ \prod_{t=1}^{t^*} K_{k\mapsto k+\nu}^{(t)}((\mathbf{y}_{k+\nu}^{(t-1)},\mathbf{u}_{k+\nu \mapsto k}^{(t-1)}),(\mathbf{y}_{k+\nu}^{(t)},\mathbf{u}_{k+\nu \mapsto k}^{(t)}))$ are essentially MH steps with an invariant distribution given by $\rho_{k\mapsto k+\nu}^{(0)}:=\pi(\, \cdot\mid k)\times q_{k\mapsto k+\nu}\times |J_{T_{k\mapsto k+\nu}}|^{-1}$. This implies that
\begin{align*}
 &\left| \int \prod_{t=1}^{t^*} K_{k\mapsto k+\nu}^{(t)}((\mathbf{y}_{k+\nu}^{(t-1)},\mathbf{u}_{k+\nu \mapsto k}^{(t-1)}),(\mathbf{y}_{k+\nu}^{(t)},\mathbf{u}_{k+\nu \mapsto k}^{(t)})) \, d(\mathbf{y}_{k+\nu}^{(1:t^*)},\mathbf{u}_{k+\nu \mapsto k}^{(1:t^*)}) \right. \cr
 &\qquad \left.- \int \prod_{t=1}^{t^*} K_{k\mapsto k+\nu}^{(0)}((\mathbf{y}_{k+\nu}^{(t-1)},\mathbf{u}_{k+\nu \mapsto k}^{(t-1)}),(\mathbf{y}_{k+\nu}^{(t)},\mathbf{u}_{k+\nu \mapsto k}^{(t)})) \, d(\mathbf{y}_{k+\nu}^{(1:t^*)},\mathbf{u}_{k+\nu \mapsto k}^{(1:t^*)}) \right|<\epsilon,
\end{align*}
 which in turns implies that
 \begin{align*}
 &\left| \int \prod_{t=1}^{t^*} K_{k\mapsto k+\nu}^{(0)}((\mathbf{y}_{k+\nu}^{(t-1)},\mathbf{u}_{k+\nu \mapsto k}^{(t-1)}),(\mathbf{y}_{k+\nu}^{(t)},\mathbf{u}_{k+\nu \mapsto k}^{(t)})) \, d(\mathbf{y}_{k+\nu}^{(1:t^*)},\mathbf{u}_{k+\nu \mapsto k}^{(1:t^*)}) \right. \cr
 &\qquad \left.- \int \rho_{k\mapsto k+\nu}^{(0)}(\mathbf{y}_{k+\nu}^{(t^*)},\mathbf{u}_{k+\nu \mapsto k}^{(t^*)}) \, d(\mathbf{y}_{k+\nu}^{(t^*)},\mathbf{u}_{k+\nu \mapsto k}^{(t^*)}) \right|<M_1(\mathbf{y}_{k+\nu}^{(0)},\mathbf{u}_{k+\nu \mapsto k}^{(0)}) \, \rho_1^{t^*},
\end{align*}
where $M(\mathbf{y}_{k+\nu}^{(0)},\mathbf{u}_{k+\nu \mapsto k}^{(0)})$ is finite for all $(\mathbf{y}_{k+\nu}^{(0)},\mathbf{u}_{k+\nu \mapsto k}^{(0)})$ and $\rho_1<1$. The uniform convergence $P_Th\longrightarrow P_{\text{ideal}}h$ on each compact subset of the state-space as $T\longrightarrow\infty$ follows if
\[
 \int q_{k\mapsto k+\nu}(\mathbf{u}_{k \mapsto k+\nu}^{(0)}) M_1(\mathbf{y}_{k+\nu}^{(0)},\mathbf{u}_{k+\nu \mapsto k}^{(0)}) \, d\mathbf{u}_{k \mapsto k+\nu}^{(0)}
\]
is finite and continuous in $\mathbf{x}_{k}$ (recall that $\mathbf{x}_{k}$ and $\mathbf{u}_{k \mapsto k+\nu}^{(0)}$ are mapped to $(\mathbf{y}_{k+\nu}^{(0)},\mathbf{u}_{k+\nu \mapsto k}^{(0)})$ using $T_{k\mapsto k+\nu}$).

\begin{proof}[Proof of Corollary 1]
 The proof is an application of Theorem 3.17 in \cite{andrieu2019peskun} which will allow to establish that
 \[
  \text{var}_\lambda(f, P_{\text{NRJ}}) \leq \text{var}_\lambda(f, P_{\text{RJ}}^{\text{unif}}),
 \]
 where $\text{var}_\lambda(f, P_{\text{NRJ}}) := \E[\{f(K(0), \mathbf{X}_K(0))\}^2] + 2\sum_{m>0} \lambda^m \E[f(K(0), \mathbf{X}_K(0))f(K(m), \mathbf{X}_K(m))]$ of $\{(K(m), \mathbf{X}_K(m)): m\in\na\}$ being a Markov chain of transition kernel $P$ at equilibrium and $\lambda \in[0, 1)$. The limit of the RHS $\lim_{\lambda \longrightarrow 1} \text{var}_\lambda(f, P_{\text{RJ}}^{\text{unif}})$ exists and is equal to $\text{var}(f, P_{\text{RJ}}^{\text{unif}})$ because the Markov chain is reversible (see \cite{andrieu2019peskun}). We will be able to conclude that the limit of the LHS exists as well using \autoref{lemma_asymp_var} that is presented after this proof.

In order to apply Theorem 3.17, we must verify that
\begin{align*}
 &P_{\text{RJ}}^{\text{unif}}((k, \mathbf{x}_k), (k', \mathbf{y}_{k'})) = \frac{1}{2} T_{+1}((k, \mathbf{x}_k), (k', \mathbf{y}_{k'})) + \frac{1}{2} T_{-1}((k, \mathbf{x}_k), (k', \mathbf{y}_{k'})) \cr
 &\qquad + \delta_{(k, \mathbf{x}_k)}(k', \mathbf{y}_{k'})\left(1 - \frac{1}{2} T_{+1}((k, \mathbf{x}_k), (k', \re^{d_{k'}})) - \frac{1}{2} T_{-1}((k, \mathbf{x}_k), (k', \re^{d_{k'}}))\right),
\end{align*}
where $T_{-1}$ and $T_{+1}$ are two sub-stochastic kernels associated with accepted proposals when the current and next values for the direction variable are $\nu = -1$ and $\nu = +1$, respectively.

We set
\[
 T_{\nu}((k, \mathbf{x}_k), (k', \mathbf{y}_{k'})) := P_{S, \nu}((k,\mathbf{x}_k), (k+\nu, \mathbf{y}_{k+\nu})) \Prob(S) + \delta_{k}(k')\, P_{S^c, \nu}(\mathbf{x}_{k'}, \mathbf{y}_{k'}) \Prob(S^c),
\]
where $S$ is used to denote that a model switch is proposed, $P_{S, \nu}$ is the conditional transition kernel given $S$ and $\nu$ and $P_{S^c, \nu}$ is the conditional transition kernel given $S^c$ and $\nu$. Note that $P_{S^c, \nu}$ is in fact independent of $\nu$ (the parameters are updated in the same way whether $\nu = -1$ or $\nu = +1$), therefore we simplify the notation by denoting this kernel by $P_{S^c} := P_{S^c, \nu}$.

We thus have that
\begin{align*}
 &\frac{1}{2} T_{+1}((k, \mathbf{x}_k), (k', \mathbf{y}_{k'})) + \frac{1}{2} T_{-1}((k, \mathbf{x}_k), (k', \mathbf{y}_{k'})) \cr
 &\quad = \frac{1}{2} (P_{S, +1}((k,\mathbf{x}_k), (k+1, \mathbf{y}_{k+1})) \Prob(S) + \delta_{k}(k')\, P_{S^c}(\mathbf{x}_{k'}, \mathbf{y}_{k'}) \Prob(S^c)) \cr
 &\qquad + \frac{1}{2} (P_{S, -1}((k,\mathbf{x}_k), (k-1, \mathbf{y}_{k-1})) \Prob(S) + \delta_{k}(k')\, P_{S^c}(\mathbf{x}_{k'}, \mathbf{y}_{k'}) \Prob(S^c)) \cr
 &\quad = \Prob(S)\left(\frac{1}{2}P_{S, +1}((k,\mathbf{x}_k), (k+1, \mathbf{y}_{k+1})) + \frac{1}{2}P_{S, -1}((k,\mathbf{x}_k), (k-1, \mathbf{y}_{k-1})) \right) \cr
 &\qquad + \Prob(S^c) P_{S^c}(\mathbf{x}_{k'}, \mathbf{y}_{k'}),
\end{align*}
which corresponds as explained in the proof of Proposition 1 to the sub-stochastic kernel associated with accepted proposals for standard RJ. This concludes the proof.
\end{proof}

\begin{Lemma}\label{lemma_asymp_var}
 Assume that $P_{\text{NRJ}}$ is uniformly ergodic. Then, for any real-valued bounded function $f$ of $(k, \mathbf{x}_k)$,
\begin{align*}
\lim_{\lambda\to 1}\sum_{m>0} \lambda^m \E[f(K(0), \mathbf{X}_K(0))f(K(m), \mathbf{X}_K(m))]  = \sum_{m>0} \E[f(K(0), \mathbf{X}_K(0))f(K(m), \mathbf{X}_K(m))],
\end{align*}
where $\{(K(m), \mathbf{X}_K(m)): m\in\na\}$ is a Markov chain of transition kernel $P_{\text{NRJ}}$ at equilibrium.
\end{Lemma}

\begin{proof}
    To simplify the notation, define $ \langle f, P_{\text{NRJ}}^m f \rangle := \E[f(K(0), \mathbf{X}_K(0))f(K(m), \mathbf{X}_K(m))]$. Define the sequence of functions $S_n:\lambda\mapsto \sum_{0<k\leq n}\lambda^k \langle f, P_{\text{NRJ}}^k f \rangle$ defined for $\lambda\in[0,1)$ and its limit $S(\lambda) = \sum_{0<k}\lambda^k \langle f, P_{\text{NRJ}}^k f \rangle$. We now show that the partial sum $S_n$ converges uniformly to $S$ on $[0,1)$, and given that for each $n\in \na$, the function $\lambda \to \lambda^n \langle f, P_{\text{NRJ}}^n f \rangle$ admits a limit when $\lambda\to 1$,  we have that $S$ admits a limit when $\lambda\to 1$, given by
 $$
 \lim_{\lambda \to 1}S(\lambda)=S(1)=\sum_{k>0} \langle f, P_{\text{NRJ}}^k f \rangle
 $$
 First, note that
 \begin{align*}
   \sup_{\lambda\in[0,1)}\left|S_n(\lambda)-S(\lambda)\right|=   \sup_{\lambda\in[0,1)}\left|\sum_{k>n}\lambda^k \langle f, P_{\text{NRJ}}^k f \rangle \right| &\leq
   \sup_{\lambda\in[0,1)} \sum_{k>n} \lambda^k \left|\langle f, P_{\text{NRJ}}^k f \rangle\right| \cr
   &= \sum_{k>n}\left|\langle f, P_{\text{NRJ}}^k f \rangle\right|.
 \end{align*}
Thus, to prove that $\sup_{\lambda\in[0,1)}\left|S_n(\lambda)-S(\lambda)\right|\longrightarrow 0$, it is sufficient to prove that the series $\sum_{k>0}\left|\langle f, P_{\text{NRJ}}^k f \rangle\right|$ converges.

Given that $f$ is bounded we can consider without loss of generality that its expectation is 0 and that it takes values between $-1$ and $+1$ (we can re-normalise it). Because $P_{\text{NRJ}}$ is assumed to be uniformly ergodic, there exists constants $\rho\in(0,1)$ and $M\in(0,\infty)$ such that for any $m \in \na$,
\begin{equation}\label{eq_unif_erg}
\sup_{(k, \mathbf{x}_k, \nu)} \|\delta_{k, \mathbf{x}_k, \nu} P_{\text{NRJ}}^m - \pi \otimes \mathcal{U}\{-1, +1\}\|_{\text{tv}} \leq M\rho^m,
\end{equation}
where for any signed measure $\mu$, $\|\mu\|_{\text{tv}}$ denotes its total variation. Note that $\|\mu\|_{\text{tv}} = (1/2) \sup_{f : \mathcal{X} \to [-1, +1]} |\mu f|$ (see for instance \cite{roberts2004general}, Proposition 3).

We have that
\begin{align*}
  |\langle f, P_{\text{NRJ}}^k f \rangle|= |\E f(K, \mathbf{X}_k, \nu) P_{\text{NRJ}}^k f(K, \mathbf{X}_k, \nu)| &\leq \E |f(K, \mathbf{X}_k, \nu)||P_{\text{NRJ}}^k f(K, \mathbf{X}_k, \nu)| \cr
  &\leq \E|P_{\text{NRJ}}^k f(K, \mathbf{X}_k, \nu)| \cr
  &= \E|P_{\text{NRJ}}^k f(K, \mathbf{X}_k, \nu) - \pi f| \cr
  &\leq \E \sup_f \left| P_{\text{NRJ}}^k f(K, \mathbf{X}_k, \nu) - \pi f \right| \cr
  &\leq M\rho^k,
\end{align*}
which is clearly summable. As a consequence, $S_n$ converges uniformly to $S$ on $[0,1)$ which concludes the proof.
\end{proof}

We now highlight what modifications and which additional technical conditions are required if geometric ergodicity is instead assumed. The constant $M$ in \eqref{eq_unif_erg} would depend on $(K, \mathbf{X}_k, \nu)$. Therefore, if $\E M(K, \mathbf{X}_k, \nu)$ is finite the result is also valid.

%
%\subsection{Other justifications for the conjecture of \autoref{sec_worst_logconcave}}\label{sec_other_justifications}
%
%If we look at the probability to reach $k'=2$ from $k=1$, it is given by $\tilde{g}(1,2)(1\wedge (\pi(2) \, \tilde{g}(2,1)) / (\pi(1) \, \tilde{g}(1,2)))$, considering an arbitrary proposal distribution $\tilde{g}$. One can show that $g^*(1,2)=1$ maximises this probability for any $\tilde{g}(2,1)$. By symmetry, $g^*(\text{K}_{\max}, \text{K}_{\max} - 1) = 1$ maximises the probability to reach $k'=\text{K}_{\max}-1$ from $k=\text{K}_{\max}$. Consequently, these transitions are optimal in the sense of \cite{peskun1973optimum} (the function $\tilde{g}$ that maximises the off-diagonal elements of the transition matrix is optimal). It makes intuitively a lot of sense to have a 0 probability to propose to go to states 0 or $\text{K}_{\max}+1$ which have 0 mass.
%
%When the distribution is \eqref{eqn_distributions_phi}, $\text{K}_{\max}$ is odd and the current state is the mode, i.e.\ $k=k^*=(\text{K}_{\max} + 1)/2$, $g^*(k, k + 1) = g^*(k, k - 1) = 1/2$, which again seems optimal when the PMF is symmetric. Additionally, for $3\leq k \leq k^* - 2$, $g^*(k, k + 1)=\phi / (\phi + 1)= 1 - g^*(k, k - 1)$ and the acceptance probabilities for these moves are exactly 1. The MH correction is thus unnecessary. By symmetry, it is the same when $k^* + 2\leq k \leq \text{K}_{\max} - 2$. This minimises the diagonal elements, which in a sense ``maximises'' the off-diagonal ones.

\subsection{Weak convergence results for the ideal samplers}\label{sec_theoretical_ana}

 We analyse the asymptotic scenario in which the number of models grows to infinity. It will be noticed that the reversible and non-reversible Markov chains produced respectively by ideal RJ and NRJ have two distinct asymptotic behaviours which are consistent with what is observed for fixed numbers of models (see, e.g., \autoref{fig_traces}), explaining their different state-space exploration speed.

We prove convergence towards continuous-time stochastic processes that take values on the real line. We thus need to consider functions of $K$ to achieve that. Firstly, we consider that the model indicator $K$ takes values in $\mathcal{K}^n:=\{1,\ldots,\lfloor \sqrt{n} \log n \rfloor\}$, where $\lfloor \cdot \rfloor$ is the floor function. We added the superscript $n$ to highlight the dependence on this variable. We select $\mathcal{K}^n$ in this way to obtain a random variable $S_K^n:=(K^n - \psi(n)) / \sqrt{n}$ that is (in the limit) continuous in addition to taking values on the real line, for a given function $\psi$ (which can be thought of as the mean that can be for instance $\lfloor \sqrt{n} \log n \rfloor/2$). Imagine that the mode is around $\lfloor \sqrt{n} \log n \rfloor/2$ (so the mass is moving towards infinity), this transformation puts the mass around 0 and makes the different values of the centred variable ($-1, 0, 1$ and so on) close to each other (e.g.\ $|1-0|/\sqrt{n}\longrightarrow 0$). We squeeze the state-space as in the proof of existence of Brownian motion from random walks. We assume that $\pi^n(k)>0$ for all $k\in \mathcal{K}^n$. For $t \geq 0$, we define the following rescaled stochastic process:
\[
 Z_{\text{RJ}}^n(t):=\frac{K_{\text{RJ}}^n(\lfloor nt \rfloor) - \psi(n)} {\sqrt{n}},
\]
where $\{K_{\text{RJ}}^n(m): m\in\na\}$ is a Markov chain produced by the ideal RJ corresponding to the ideal NRJ described in Section 2.2 in our paper. We consider that this RJ updates parameters and switches models with probabilities $\tau$ and $1-\tau$, respectively, and that $g(k,k+1)=g(k,k-1)=0.5(1-\tau)$, so it proposes to increase or decrease the model indicator with the same probability and $\alpha_{\text{RJ}}=\alpha_{\text{NRJ}}$. The continuous-time stochastic process $\{Z_{\text{RJ}}^n(t): t\geq 0\}$ is a sped up and modified version of $\{K_{\text{RJ}}^n(m): m\in\na\}$. The decreasing size of the jumps of $\{Z_{\text{RJ}}^n(t): t\geq 0\}$ as $n$ increases (the size is $1/\sqrt{n}$), combined with its time acceleration, result in a continuous and non-trivial limiting process, as specified in \autoref{thm_conv_RJ}. This time acceleration can be thought of as squeezing the time axis to make the iterations close to each other, again as in the proof of existence of Brownian motion from random walks.

\setcounter{Theorem}{1}

\begin{Theorem}[Weak convergence of RJ]\label{thm_conv_RJ}
 Assume that:
 \begin{description}
  \itemsep 0mm

  \item[(a)] the function $\psi$ can be chosen such that $S_K^n$ is asymptotically distributed as $f_{S}\in \mathcal{C}^1(\re)$, a strictly positive probability density function (PDF), where $\mathcal{C}^1(\re)$ denotes the space of real-valued functions on $\re$ with continuous first derivative;

  \item[(b)] the function $(\log f_{S}(\cdot))'$ is Lipschitz continuous;

  \item[(c)] $\psi$ can be chosen such that
 \begin{align}\label{eqn_cond_RJ}
  \frac{1}{\pi^n(k)} \, \frac{\pi^n(k + 1) - \pi^n(k)}{1/\sqrt{n}} - (\log f_{S}(S_k^n))'
 \end{align}
 is bounded for all $n$ and converges towards 0 as $n\longrightarrow\infty$, for all $k$;

 \item[(d)] $\lim_{n\longrightarrow\infty} \sqrt{n} \,\pi^n(1)= \lim_{n\longrightarrow\infty} \sqrt{n} \, \pi^n(\lfloor \sqrt{n} \log n \rfloor) = 0$.

\end{description}

     If $K_{\text{RJ}}^n(0)\sim \pi^n$, then $\{Z_{\text{RJ}}^n(t): t\geq 0\}$ converges weakly towards a Langevin diffusion as $n\longrightarrow \infty$, i.e.\
\begin{align*}\label{eqn_stoc_RJ}
 \{Z_{\text{RJ}}^n(t): t\geq 0\} \Longrightarrow \{Z_{\text{RJ}}(t): t\geq 0\} \quad \text{as} \quad  n\longrightarrow \infty,
\end{align*}
where the process $\{Z_{\text{RJ}}(t): t\geq 0\}$ is such that $Z_{\text{RJ}}(0)\sim f_{S}$ and
\[
 dZ_{\text{RJ}}(t)= \frac{1-\tau}{2} \, (\log f_{S}(Z_{\text{RJ}}(t)))' \, dt + \sqrt{1-\tau} 	\, dB(t),
\]
with $\{B(t): t\geq 0\}$ being a Wiener process.
\end{Theorem}

\begin{proof}
 It is a straightforward adaptation of Theorem 1 in \cite{GAGNON201932}. For sake of completeness, it is detailed in \autoref{sec_proofs_supp}.
\end{proof}

The notation ``$\Longrightarrow$'' represents here weak convergence of processes in the Skorokhod topology (see Section 3 of \cite{ethier1986markov} for more details about this type of convergence).

The two main assumptions are (a) and (c). The former requires to find a transformation of $K^n$ such that the limit in distribution of the transformed random variable is a continuous random variable with density $f_S$. The latter requires that the ``discrete version'' of the derivative of $\log \pi^n$ share the same asymptotic behaviour as the derivative of $\log f_S$. Indeed, in \cite{GAGNON201932}, it is explained that the left term in \eqref{eqn_cond_RJ} can be seen as the discrete version of the derivative of $\log \pi^n$ because $\pi^n$ is also the PMF of $S_k^n$ (evaluated at a different point) and $S_{k+1}^n - S_k^n=1/\sqrt{n}$. Assumption (b) is standard in the weak convergence literature; it ensures the existence of a unique strong solution to the stochastic differential equation given above. Assumption (d) is a regularity condition. In \cite{GAGNON201932}, to illustrate how a PMF that satisfies the conditions looks like, the authors show one that is such that $S_K^n$ converges in distribution towards a standard normal.

We now analyse the behaviour of the stochastic process produced by the ideal NRJ algorithm. We consider as before that $\mathcal{K}^n=\{1,\ldots,\lfloor \sqrt{n} \log n \rfloor\}$ and $\pi^n(k)>0$ for all $k\in \mathcal{K}^n$. For $t \geq 0$, we define the following rescaled stochastic process:
\begin{align}\label{eqn_stoc_NRJ}
 \mathbf{Z}_{\text{NRJ}}^n(t):=\left(\frac{K_{\text{NRJ}}^n(\lfloor \sqrt{n}t \rfloor) - \psi(n)}{\sqrt{n}}, \nu(\lfloor  \sqrt{n}t \rfloor)\right),
\end{align}
where $\{(K_{\text{NRJ}}^n, \nu)(m): m\in\na\}$ is a Markov chain produced by ideal NRJ described in Section 2.2 in our paper. Note that the distribution of $\nu$ does not change with $n$.

\begin{Theorem}[Weak convergence of NRJ]\label{thm_conv_NRJ}
 Assume that the same conditions (a)-(d) as in \autoref{thm_conv_RJ} are satisfied. Assume additionally that there exist two positive constants $c$ and $x_0$ such that $|(\log f_{S}(x))'|\geq c$ for all $|x|\geq x_0$. If $(K_{\text{NRJ}}^n, \nu)(0)\sim \pi^n\otimes \mathcal{U}\{-1, 1\}$, then $\{\mathbf{Z}_{\text{NRJ}}^n(t): t\geq 0\}$ converges weakly towards a piecewise deterministic Markov process (PDMP) as $n\longrightarrow\infty$,  i.e.\
\[
 \{\mathbf{Z}_{\text{NRJ}}^n(t): t\geq 0\} \Longrightarrow \{\mathbf{Z}_{\text{NRJ}}(t): t\geq 0\} \quad \text{as} \quad  n\longrightarrow \infty,
\]
where the process $\{\mathbf{Z}_{\text{NRJ}}(t): t\geq 0\}$ is such that $\mathbf{Z}_{\text{NRJ}}(0)\sim f_{S} \otimes \mathcal{U}\{-1, 1\}$ with generator
\[
 Gh(x,y):=(1-\tau) y h_x(x, y)+\max\{0, -y\, (\log f_{S}(x))'\} (1-\tau) (h(x, -y) - h(x, y)),
\]
where $h(\, \cdot \,, y)\in \mathcal{C}^1(\re)$ and such that itself and $h_x(\, \cdot \,, y)$ vanish at infinity, for $y\in\{-1,1\}$, $h_x$ denoting the first derivative of $h$ with respect to its first argument.
\end{Theorem}

\begin{proof}
See \autoref{sec_proofs_supp}.
\end{proof}

The additional regularity condition on $f_S$ in \autoref{thm_conv_NRJ} essentially ensures that outside of a bounded set, this PDF decreases  sufficiently quickly. Indeed, given that $(\log f_{S}(x))'=f_{S}'(x)/f_{S}(x)$ and $f_{S}$ is strictly positive, it is required that the tail decay is bounded from below (relatively to $f_{S}$). This guarantees that the limiting PDMP has some important properties (e.g.\ non-explosiveness and $f_{S}\otimes \mathcal{U}\{-1, 1\}$ is an invariant distribution, see \cite{bierkens2017piecewise}).

The PDMP in \autoref{thm_conv_NRJ} corresponds to a zig-zag Markov process (\cite{bierkens2019zig}), and in fact, a bouncy particle sampler (BPS, \cite{bouchard2018bouncy}) given that they both coincide when the position variable is unidimensional. This position variable evolves with constant drift $1-\tau$ either to the right or left of the real line depending on the direction variable, and changes direction with rate $\max\{0, -y\, (\log f_{S}(x))'\} (1-\tau)$ when the position is $x$ and direction $y$. PDMP are known for being non-diffusive and having persistency-driven paths. We constructed NRJ to induce such a behaviour, but we do not know \textit{a priori} when this will happen and how this will translate. An analysis was conducted in Section 4 in our paper to provide some answers.  \autoref{thm_conv_NRJ} and \autoref{thm_conv_RJ} indicate that in the (asymptotic) theoretical framework considered, the model indicator's paths produced by RJ and NRJ behave exactly as expected; the former show diffusive patterns and the latter not. This suggests that (at least under those conditions) NRJ outperform RJ. We even have a guarantee for the speed of convergence towards the target distribution for NRJ: \cite{bierkens2017piecewise} prove that the PDMP in \autoref{thm_conv_NRJ} is exponentially ergodic. We additionally know that the convergence is an order of magnitude slower for $\{K_{\text{RJ}}^n(m): m\in\na\}$. Indeed, the different behaviour of $\{K_{\text{NRJ}}^n(m): m\in\na\}$ compared with $\{K_{\text{RJ}}^n(m): m\in\na\}$ requires to accelerate the time by a factor of only $\sqrt{n}$ in the definition of $\{\mathbf{Z}_{\text{NRJ}}^n(t): t\geq 0\}$ comparatively to $n$ in that of $\{Z_{\text{RJ}}^n(t): t\geq 0\}$ to obtain non-trivial limiting stochastic processes. This highlights again that $\{K_{\text{NRJ}}^n(m): m\in\na\}$ explores its state-space more quickly.

\subsubsection{Proofs of Theorems~\ref{thm_conv_RJ} and \ref{thm_conv_NRJ}}\label{sec_proofs_supp}

\begin{proof}[Proof of \autoref{thm_conv_RJ}]
 In order to prove the result, we demonstrate the convergence of the finite-dimen\-sional distributions of $\{Z_{\text{RJ}}^n(t): t\geq 0\}$ to those of $\{Z_{\text{RJ}}(t): t\geq 0\}$. To achieve this, we verify Condition (c) of Theorem 8.2 from chapter 4 of \cite{ethier1986markov}. The weak convergence then follows from Corollary 8.6 of Chapter 4 of \cite{ethier1986markov}. The remaining conditions of Theorem 8.2 and the conditions specified in Corollary 8.6 are either straightforward or easily derived from the proof given here.

 The proof of the convergence of the finite-dimensional distributions relies on the convergence of (what we call) the ``pseudo-generator'', a quantity that we define as:
 \[
  \varrho_{\text{RJ}}^n(t):= n\,\E[h(Z_{\text{RJ}}^n(t + 1/n)) - h(Z_{\text{RJ}}^n(t))\mid \mathcal{F}^{Z_{\text{RJ}}^n}(t)],
 \]
 where $h\in\mathcal{C}_c^\infty(\re)$, the space of infinitely differentiable functions on $\re$ with compact support. Theorem 2.1 from Chapter 8 of \cite{ethier1986markov} allows us to restrict our attention to this set of functions when studying the limiting behaviour of the pseudo-generator.   In our situation, the pseudo-generator has a more precise expression:
 \begin{align}\label{eqn_generator_RJ}
  \varrho_{\text{RJ}}^n(t) &= \frac{n (1-\tau)}{2}\left((h(S_{K+1}^n) - h(S_{K}^n))\left(1 \wedge \frac{\pi^n(K^n + 1)}{\pi^n(K^n)}\right)\right) \cr
  &\qquad + \frac{n (1-\tau)}{2}\left((h(S_{K-1}^n) - h(S_{K}^n))\left(1 \wedge \frac{\pi^n(K^n - 1)}{\pi^n(K^n)}\right)\right).
 \end{align}
 Note that the Markov process $\{K_{\text{RJ}}^n(m): m\in\na\}$ is time-homogeneous, and because of this we replaced the random variable $Z_{\text{RJ}}^n(t)$ by $S_{K}^n$ and $Z_{\text{RJ}}^n(t + 1/n)$ by $S_{K+1}^n$ or $S_{K-1}^n$ given that we will work under expectations. Indeed, Condition (c) of Theorem 8.2 from chapter 4 of \cite{ethier1986markov} essentially reduces to the following convergence:
 \[
  \E\left[\abs{\varrho_{\text{RJ}}^n(t) - Gh(Z_{\text{RJ}}^n(t))}\right]\longrightarrow 0 \quad \text{as} \quad n\longrightarrow \infty,
 \]
 where $G$ is the generator of the limiting diffusion with
 \[
  Gh(Z_{\text{RJ}}^n(t)) := \frac{1-\tau}{2}\, (\log f_{S}(Z_{\text{RJ}}^n(t)))' h'(Z_{\text{RJ}}^n(t)) + \frac{1-\tau}{2} \, h''(Z_{\text{RJ}}^n(t)).
 \]
 Note that there exists a positive constant $M$ such that $h$ and all its derivatives are bounded in absolute value by this constant. We choose $M$ such that it is a Lipschitz constant for the function $(\log f_{S}(\cdot))'$.

 The key here is to use Taylor expansions in \eqref{eqn_generator_RJ} to obtain derivatives of $h$ as in $G$. By noting that $S_{K+1}^n=S_{K}^n + 1/\sqrt{n}$ and $S_{K-1}^n=S_{K}^n - 1/\sqrt{n}$, and using Taylor expansions of $h$ around $S_K^n$, we obtain
 \begin{align*}
  h(S_{K}^n + 1/\sqrt{n}) - h(S_{K}^n)&=\frac{1}{\sqrt{n}}\, h'(S_{K}^n) + \frac{1}{2n}\, h''(S_{K}^n) + \frac{1}{6n^{3/2}}h'''(W), \cr
  h(S_{K}^n - 1/\sqrt{n}) - h(S_{K}^n)&=-\frac{1}{\sqrt{n}}\, h'(S_{K}^n) + \frac{1}{2n}\, h''(S_{K}^n) - \frac{1}{6n^{3/2}}h'''(T),
 \end{align*}
 where $W$ and $T$ belong to $(S_{K}^n, S_{K}^n+1/\sqrt{n})$ and $(S_{K}^n - 1/\sqrt{n}, S_{K}^n)$, respectively. We also note that the first term on the RHS of \eqref{eqn_generator_RJ} equals 0 when $K^n = \lfloor \sqrt{n} \log n \rfloor$ because $\pi^n(\lfloor \sqrt{n} \log n \rfloor + 1)=0$. For the analogous reason, the second term on the RHS of \eqref{eqn_generator_RJ} equals 0 when $K^n = 1$. Therefore,
 \begin{align}\label{eqn_diff_gens}
  & \varrho_{\text{RJ}}^n(t) - Gh(Z_{\text{RJ}}^n(t)) = \ind(2 \leq K^n \leq \lfloor \sqrt{n} \log n \rfloor - 1) \, \frac{1-\tau}{2} \, h'(S_{K}^n) \cr
  &\qquad \times \left(\sqrt{n}\left(1 \wedge \frac{\pi^n(K^n + 1)}{\pi^n(K^n)} - 1 \wedge \frac{\pi^n(K^n - 1)}{\pi^n(K^n)}\right) - (\log f_{S}(S_{K}^n))'\right) \cr
  &\quad + \ind(K^n = 1)\, \frac{1-\tau}{2} \, h'(S_{K}^n)\left(\sqrt{n} \left(1 \wedge \frac{\pi^n(K^n + 1)}{\pi^n(K^n)}\right) - (\log f_{S}(S_{K}^n))'\right) \cr
  &\quad - \ind(K^n = \lfloor \sqrt{n} \log n \rfloor) \, \frac{1-\tau}{2} \, h'(S_{K}^n)\left(\sqrt{n} \left(1 \wedge \frac{\pi^n(K^n - 1)}{\pi^n(K^n)}\right) - (\log f_{S}(S_{K}^n))'\right) \cr
  &\quad + \ind(2 \leq K^n \leq \lfloor \sqrt{n} \log n \rfloor - 1)\, \frac{1-\tau}{4}\,h''(S_{K}^n)\cr
  &\qquad \times\left(1 \wedge \frac{\pi^n(K^n + 1)}{\pi^n(K^n)} + 1 \wedge \frac{\pi^n(K^n - 1)}{\pi^n(K^n)} - 2\right) \cr
  &\quad + \ind(K^n = 1) \, \frac{1-\tau}{4} h''(S_{K}^n)\left( 1 \wedge \frac{\pi^n(K^n + 1)}{\pi^n(K^n)} - 2\right) \cr
  &\quad + \ind(K^n = \lfloor \sqrt{n} \log n \rfloor)\, \frac{1-\tau}{4} h''(S_{K}^n)\left( 1 \wedge \frac{\pi^n(K^n - 1)}{\pi^n(K^n)} - 2\right) \cr
  &\quad + \frac{1-\tau}{12\sqrt{n}}\,h'''(W) \left( 1 \wedge \frac{\pi^n(K^n + 1)}{\pi^n(K^n)}\right)\ind(1 \leq K^n \leq \lfloor \sqrt{n} \log n \rfloor - 1) \cr
  &\quad - \frac{1-\tau}{12\sqrt{n}}\,h'''(T) \left( 1 \wedge \frac{\pi^n(K^n - 1)}{\pi^n(K^n)}\right)\ind(2 \leq K^n \leq \lfloor \sqrt{n} \log n \rfloor).
 \end{align}

  We now prove that expectation of the absolute value of each term on the RHS in \eqref{eqn_diff_gens} converges towards 0 as $n\longrightarrow \infty$. We start with the last terms  and make our way up. It is clear that the expectation of the absolute value of each of the last two terms converges towards 0 as $n\longrightarrow \infty$ given that $|h'''|\leq M$ and $0\leq 1\wedge x\leq 1$ for positive $x$. We now analyse the fourth one (starting from the bottom). As $n\longrightarrow \infty$,
 \begin{align*}
  \E\left[\abs{\ind(K^n = 1) \, \frac{1-\tau}{2} h''(S_{K}^n)\left( 1 \wedge \frac{\pi^n(K^n + 1)}{\pi^n(K^n)} - 2\right)}\right]\leq \frac{(1-\tau) M}{2}  \Prob(K^n = 1)\longrightarrow 0,
 \end{align*}
 using $|h''|\leq M$ and
 \[
  0\leq\left| 1 \wedge \frac{\pi^n(K^n + 1)}{\pi^n(K^n)} - 2\right|\leq 2.
 \]
 Recall that $\Prob(K^n = 1)\longrightarrow 0$ by assumption. The proof for the third term (starting from the bottom) is similar.

 Applying Lemmas~\ref{lemma_1} to \ref{lemma_3} (that follow), each of the remaining terms is seen to converge towards 0 in $L^1$ as $n\longrightarrow\infty$, which concludes the proof.
\end{proof}

\begin{Lemma}\label{lemma_1}
 As $n\longrightarrow\infty$, we have
 \[
  \E\left[\abs{\ind(2 \leq K^n \leq \lfloor \sqrt{n} \log n \rfloor - 1)\, \frac{1-\tau}{4}\,h''(S_{K}^n)
  \left(1 \wedge \frac{\pi^n(K^n + 1)}{\pi^n(K^n)} + 1 \wedge \frac{\pi^n(K^n - 1)}{\pi^n(K^n)} - 2\right)}\right]\longrightarrow 0.
 \]
\end{Lemma}

\begin{proof}
 We have
 \begin{align*}
   &\E\left[\abs{\ind(2 \leq K^n \leq \lfloor \sqrt{n} \log n \rfloor - 1)\, \frac{1-\tau}{4}\,h''(S_{K}^n)
  \left(1 \wedge \frac{\pi^n(K^n + 1)}{\pi^n(K^n)} + 1 \wedge \frac{\pi^n(K^n - 1)}{\pi^n(K^n)} - 2\right)}\right] \cr
  &\leq \frac{(1-\tau)M}{4} \, \E\left[\abs{\ind(2 \leq K^n \leq \lfloor \sqrt{n} \log n \rfloor - 1)\, \left(1 \wedge \frac{\pi^n(K^n + 1)}{\pi^n(K^n)} + 1 \wedge \frac{\pi^n(K^n - 1)}{\pi^n(K^n)} - 2\right)}\right],
 \end{align*}
 because $|h''|\leq M$. We show that
 \[
  \frac{\pi^n(k + 1)}{\pi^n(k)}\longrightarrow 1 \quad \text{for all } k\in \{1,\ldots,\lfloor \sqrt{n} \log n \rfloor - 1\},
 \]
 which allows to conclude using the triangle inequality, the continuity of the function $1 \wedge x$, and the Lebesgue's dominated convergence theorem. We have
 \begin{align*}
  \abs{\frac{\pi^n(k + 1)}{\pi^n(k)} - 1}&= \abs{\frac{1}{\pi^n(k)}\, \frac{\pi^n(k + 1) - \pi^n(k)}{1/\sqrt{n}} - (\log f_{S}(S_k^n))' + (\log f_{S}(S_k^n))'} \frac{1}{\sqrt{n}}  \cr
  &\leq \abs{\frac{1}{\pi^n(k)}\, \frac{\pi^n(k + 1) - \pi^n(k)}{1/\sqrt{n}} - (\log f_{S}(S_k^n))'} \frac{1}{\sqrt{n}} \cr
  &\qquad + \abs{(\log f_{S}(S_k^n))'} \frac{1}{\sqrt{n}},
 \end{align*}
 using again the triangle inequality. By assumption, we have that
 \[
  \abs{\frac{1}{\pi^n(k)}\, \frac{\pi^n(k + 1) - \pi^n(k)}{1/\sqrt{n}} - (\log f_{S}(S_k^n))'} \frac{1}{\sqrt{n}}\longrightarrow 0.
 \]
 We also have that
 \begin{align*}
  \abs{(\log f(S_k^n))' }\frac{1}{\sqrt{n}} &=  \abs{(\log f_{S}(S_k^n))' - (\log f_{S}(0))' + (\log f_{S}(0))'}\frac{1}{\sqrt{n}} \cr
  &\leq \abs{(\log f_{S}(S_k^n))' - (\log f_{S}(0))'}\frac{1}{\sqrt{n}} + \frac{\abs{(\log f_{S}(0))'}}{\sqrt{n}} \cr
  &\leq M \abs{\frac{k - \psi(n)}{\sqrt{n}}}\frac{1}{\sqrt{n}} + \frac{\abs{(\log f_{S}(0))'}}{\sqrt{n}},
 \end{align*}
 using first the triangle inequality, and next the fact that $(\log f_{S}(\cdot))'$ is Lipschitz continuous. We have that $|(\log f_{S}(0))'|/\sqrt{n}\longrightarrow 0$ because $f_{S}\in \mathcal{C}^1(\re)$. Also,
 \begin{align*}
  \abs{\frac{k - \psi(n)}{\sqrt{n}}}\frac{1}{\sqrt{n}} \leq 2\, \frac{\lfloor \sqrt{n}\log n\rfloor}{n} \longrightarrow 0,
 \end{align*}
 using the triangle inequality and the fact that $k, \psi(n)\leq \lfloor \sqrt{n}\log n\rfloor$.
\end{proof}

\begin{Lemma}\label{lemma_2}
 As $n\longrightarrow$, we have
 \[
  \E\left[\abs{\ind(K^n = 1)\, \frac{1-\tau}{2} \, h'(S_{K}^n)\left(\sqrt{n} \left(1 \wedge \frac{\pi^n(K^n + 1)}{\pi^n(K^n)}\right) - (\log f_{S}(S_{K}^n))'\right)}\right]\longrightarrow 0,
 \]
 and
 \[
  \E\left[\abs{\ind(K^n = \lfloor \sqrt{n} \log n \rfloor) \, \frac{1-\tau}{2} \, h'(S_{K}^n)\left(\sqrt{n} \left(1 \wedge \frac{\pi^n(K^n - 1)}{\pi^n(K^n)}\right) - (\log f_{S}(S_{K}^n))'\right)}\right]\longrightarrow 0.
 \]
\end{Lemma}

\begin{proof}
 We have that
 \begin{align*}
  &\E\left[\abs{\ind(K^n = 1)\, \frac{1-\tau}{2} \, h'(S_{K}^n)\left(\sqrt{n} \left(1 \wedge \frac{\pi^n(K^n + 1)}{\pi^n(K^n)}\right) - (\log f_{S}(S_{K}^n))'\right)}\right] \cr
  &\qquad \leq  \frac{(1-\tau) M}{2} \, \E\left[\ind(K^n = 1)\sqrt{n}\left(1 \wedge \frac{\pi^n(K^n + 1)}{\pi^n(K^n)}\right)\right] + \frac{(1-\tau) M}{2} \, \E\left[\ind(K^n = 1)\abs{(\log f_{S}(S_{K}^n))'}\right],
 \end{align*}
 using that $|h'|\leq M$ and the triangle inequality. The first term on the RHS converges towards 0 by assumption because $0\leq 1\wedge x\leq 1$ for positive $x$. Using the same mathematical arguments as in the proof of \autoref{lemma_1}, we have that
 \[
  \abs{(\log f_{S}(S_{K}^n))'}\leq 2 M \, \frac{\lfloor \sqrt{n} \log n\rfloor}{\sqrt{n}}+\abs{(\log f_{S}(0))'}.
 \]
 Therefore, using the triangle inequality
 \begin{align*}
  \E\left[\ind(K^n = 1)\abs{(\log f_{S}(S_{K}^n))'}\right]\leq \Prob(K^n = 1) \left(2 M \, \frac{\lfloor \sqrt{n} \log n\rfloor}{\sqrt{n}} + \abs{(\log f_{S}(0))'}\right)\longrightarrow 0,
 \end{align*}
 by assumption (and because $f_{S}\in \mathcal{C}^1(\re)$). The proof that
 \[
  \E\left[\abs{\ind(K^n = \lfloor \sqrt{n} \log n \rfloor) \, \frac{1-\tau}{2} \, h'(S_{K}^n)\left(\sqrt{n} \left(1 \wedge \frac{\pi^n(K^n - 1)}{\pi^n(K^n)}\right) - (\log f_{S}(S_{K}^n))'\right)}\right]\longrightarrow 0
 \]
 is similar.
\end{proof}

\begin{Lemma}\label{lemma_3}
 As $n\longrightarrow\infty$, we have
 \begin{align*}
  &\E\left[\left|\ind(2 \leq K^n \leq \lfloor \sqrt{n} \log n \rfloor - 1) \, \frac{1-\tau}{2} \, h'(S_{K}^n)\right.\right. \cr
  &\qquad\times \left.\left. \left(\sqrt{n}\left(1 \wedge \frac{\pi^n(K^n + 1)}{\pi^n(K^n)} - 1 \wedge \frac{\pi^n(K^n - 1)}{\pi^n(K^n)}\right) - (\log f_{S}(S_{K}^n))'\right)\right|\right]\longrightarrow 0.
 \end{align*}
\end{Lemma}

\begin{proof}

 First, we have that
 \begin{align*}
  &\E\left[\left|\ind(2 \leq K^n \leq \lfloor \sqrt{n} \log n \rfloor - 1) \, \frac{1-\tau}{2} \, h'(S_{K}^n)\right.\right. \cr
  &\qquad\times \left.\left. \left(\sqrt{n}\left(1 \wedge \frac{\pi^n(K^n + 1)}{\pi^n(K^n)} - 1 \wedge \frac{\pi^n(K^n - 1)}{\pi^n(K^n)}\right) - (\log f_{S}(S_{K}^n))'\right)\right|\right] \cr
  &\quad \leq  \frac{(1-\tau)M}{2} \, \E\left[\left|\ind(2 \leq K^n \leq \lfloor \sqrt{n} \log n \rfloor - 1)\right.\right. \cr
  &\qquad\times \left.\left. \left(\sqrt{n}\left(1 \wedge \frac{\pi^n(K^n + 1)}{\pi^n(K^n)} - 1 \wedge \frac{\pi^n(K^n - 1)}{\pi^n(K^n)}\right) - (\log f_{S}(S_{K}^n))'\right)\right|\right],
 \end{align*}
 because $|h'|\leq M$. We now consider four cases for $K^n$:
 \begin{enumerate}
 \itemsep 0mm

  \item $\pi^n(K^n + 1)/\pi^n(K^n) < 1$ and $\pi^n(K^n - 1)/\pi^n(K^n) \geq 1$,

  \item $\pi^n(K^n + 1)/\pi^n(K^n) \geq 1$ and $\pi^n(K^n - 1)/\pi^n(K^n) < 1$,

  \item $\pi^n(K^n + 1)/\pi^n(K^n) \geq 1$ and $\pi^n(K^n - 1)/\pi^n(K^n) \geq 1$,

  \item $\pi^n(K^n + 1)/\pi^n(K^n) < 1$ and $\pi^n(K^n - 1)/\pi^n(K^n) < 1$.
 \end{enumerate}

 In Case 1, we have that
 \begin{align*}
  &\sqrt{n}\left(1 \wedge \frac{\pi^n(K^n + 1)}{\pi^n(K^n)} - 1 \wedge \frac{\pi^n(K^n - 1)}{\pi^n(K^n)}\right) - (\log f_{S}(S_{K}^n))' \cr
  &\qquad =  \sqrt{n}\left(\frac{\pi^n(K^n + 1)}{\pi^n(K^n)} - 1 \right) - (\log f_{S}(S_{K}^n))' \cr
  &\qquad = \frac{1}{\pi^n(K^n)}\, \frac{\pi^n(K^n + 1) - \pi^n(K^n)}{1/\sqrt{n}} - (\log f_{S}(S_{K}^n))'\longrightarrow 0,
 \end{align*}
 by assumption. We can prove that it converges towards 0 in Case 2 in the same way. Case 3 corresponds to a local minimum. In this case,
 \[
  \sqrt{n}\left(1 \wedge \frac{\pi^n(K^n + 1)}{\pi^n(K^n)} - 1 \wedge \frac{\pi^n(K^n - 1))}{\pi^n(K^n)}\right) = 0,
 \]
 for all $n$, and $(\log f_{S}(S_{K}^n))'=f'_{Z_{\text{RJ}}}(S_{K}^n)/f_{S}(S_{K}^n)\longrightarrow 0$. Case 4 corresponds to a local (or global) maximum. Again, $(\log f_{S}(S_{K}^n))'\longrightarrow 0$. Additionally,
 \begin{align*}
 &\sqrt{n}\left( 1 \wedge \frac{\pi^n(K^n + 1)}{\pi^n(K^n)} - 1 \wedge \frac{\pi^n(K^n - 1)}{\pi^n(K^n)}\right) - (\log f_{S}(S_{K}^n))' \cr
 &\qquad = \sqrt{n}\left( \frac{\pi^n(K^n + 1)-\pi^n(K^n)}{\pi^n(K^n)} -  \frac{\pi^n(K^n - 1)-\pi^n(K^n)}{\pi^n(K^n)}\right) - (\log f_{S}(S_{K}^n))' \cr
 &\qquad = \frac{1}{\pi^n(K^n)}\, \frac{\pi^n(K^n + 1) - \pi^n(K^n)}{1/\sqrt{n}} - (\log f_{S}(S_{K}^n))' \cr
 &\qquad \qquad - \frac{1}{\pi^n(K^n)}\, \frac{\pi^n(K^n - 1) - \pi^n(K^n)}{1/\sqrt{n}},
 \end{align*}
 but both terms converge towards 0. Consequently, Lebesgue's dominated convergence theorem allows to conclude the proof.
\end{proof}

\begin{proof}[Proof of \autoref{thm_conv_NRJ}]
 Analogously to the proof of \autoref{thm_conv_RJ}, we demonstrate the convergence of the finite-dimen\-sional distributions of $\{\mathbf{Z}_{\text{NRJ}}^n(t): t\geq 0\}$ to those of $\{\mathbf{Z}_{\text{NRJ}}(t): t\geq 0\}$. The same strategy as in that proof is employed: we verify Condition (c) of Theorem 8.2 from chapter 4 of \cite{ethier1986markov}. The weak convergence then follows from Corollary 8.6 of Chapter 4 of \cite{ethier1986markov}. The remaining conditions of Theorem 8.2 and the conditions specified in Corollary 8.6 are either straightforward or easily derived from the proof given here.

 Beforehand, we note that the additional assumption on $f_S$ (about the lower bound on $|(\log f_{S}(\cdot))'|$ outside of a bounded set) implies that Assumption 3 in Section 5 of \cite{bierkens2017piecewise} is satisfied. In that paper, it is proved that it implies that the PDMP defined in \autoref{thm_conv_NRJ} is a non-explosive strong Markov process. The authors also demonstrate that the Markov transition semigroup to which the generator corresponds is Feller.

 For this proof, the time acceleration factor is different, and accordingly, the pseudo-generator is defined as:
 \begin{align*}
  \varrho_{\text{NRJ}}^n(t)&:= \sqrt{n}\,\E[h(\mathbf{Z}_{\text{NRJ}}^n(t + 1/\sqrt{n})) - h(\mathbf{Z}_{\text{NRJ}}^n(t))\mid \mathcal{F}^{\mathbf{Z}_{\text{NRJ}}^n}(t)] \cr
  &\hspace{1mm}= \sqrt{n} (1-\tau) (h(S_{K+\nu}^n, \nu) - h(S_{K}^n, \nu))\left(1 \wedge \frac{\pi^n(K^n + \nu)}{\pi^n(K^n)}\right) \cr
  &\qquad + \sqrt{n} (1-\tau) (h(S_{K}^n, -\nu) - h(S_{K}^n, \nu))\left(1- 1 \wedge \frac{\pi^n(K^n + \nu)}{\pi^n(K^n)}\right).
 \end{align*}
 As in the proof of \autoref{thm_conv_RJ}, we replaced $\mathbf{Z}_{\text{NRJ}}^n(t)$ by $(S_{K}^n, \nu)$ and $\mathbf{Z}_{\text{NRJ}}^n(t + 1/\sqrt{n})$ by $(S_{K+\nu}^n, \nu)$ or $(S_{K}^n, -\nu)$ given that the Markov process $\{(K_{\text{NRJ}}^n, \nu)(m): m\in\na\}$ is time-homogeneous and we will work under expectations. Recall that Condition (c) of Theorem 8.2 from chapter 4 of \cite{ethier1986markov} is essentially
 \[
  \E\left[\abs{\varrho_{\text{NRJ}}^n(t) - Gh(\mathbf{Z}_{\text{NRJ}}^n(t))}\right]\longrightarrow 0 \quad \text{as} \quad n\longrightarrow\infty,
 \]
 where $G$ is in this case the generator expressed in \autoref{thm_conv_NRJ}. We have that
 \begin{align}\label{eqn_diff_gens_NRJ}
  &\E\left[\abs{\varrho_{\text{NRJ}}^n(t) - Gh(\mathbf{Z}_{\text{NRJ}}^n(t))}\right] \nonumber \\
  &\quad\leq \E\left[\abs{\sqrt{n} (1-\tau) (h(S_{K+\nu}^n, \nu) - h(S_{K}^n, \nu))\left(1 \wedge \frac{\pi^n(K^n + \nu)}{\pi^n(K^n)}\right) - (1-\tau) \nu h_x(S_{K}^n, \nu)}\right] \nonumber \\
  &\qquad + \E\left[\left|\sqrt{n} (1-\tau) (h(S_{K}^n, -\nu) - h(S_{K}^n, \nu))\left(1- 1 \wedge \frac{\pi^n(K^n + \nu)}{\pi^n(K^n)}\right) \right.\right. \nonumber \\
  &\qquad\qquad - \left.\left.\max\{0, -\nu\, (\log f_{S}(S_{K}^n))'\} (1-\tau) (h(S_{K}^n, -\nu) - h(S_{K}^n, \nu))\right|\right],
 \end{align}
 using the triangle inequality. We analyse the two terms separately. We start with the first one. By the mean value theorem and using that $S_{K+\nu}^n - S_{K}^n = \nu/\sqrt{n}$, we have that
 \[
  h(S_{K+\nu}^n, \nu) - h(S_{K}^n, \nu) = \frac{\nu}{\sqrt{n}}\, h_x(T, \nu),
 \]
 where $T$ is in $(S_{K}^n, S_{K+\nu}^n)$ or $(S_{K+\nu}^n, S_{K}^n)$. We therefore also know that $T\longrightarrow S_{K}^n$ with probability 1. In the proof of \autoref{lemma_1}, it is shown that
 \[
  1 \wedge \frac{\pi^n(K^n + \nu)}{\pi^n(K^n)} \longrightarrow 1 \quad \text{with probability 1},
 \]
 and consequently,
 \[
 \E\left[\abs{\sqrt{n} (1-\tau) (h(S_{K+\nu}^n, \nu) - h(S_{K}^n, \nu))\left(1 \wedge \frac{\pi^n(K^n + \nu)}{\pi^n(K^n)}\right) - (1-\tau) \nu h_x(S_{K}^n, \nu)}\right]\longrightarrow 0,
 \]
 using Lebesgue's dominated convergence theorem (given that the quantity in the expectation is bounded, because $h_x(\cdot, \nu)$ is bounded for $\nu\in\{-1,1\}$). For the second term in \eqref{eqn_diff_gens_NRJ}, we have
 \begin{align}\label{eqn_diff_gens_NRJ_2}
  &\E\left[\left|\sqrt{n} (1-\tau) (h(S_{K}^n, -\nu) - h(S_{K}^n, \nu))\left(1- 1 \wedge \frac{\pi^n(K^n + \nu)}{\pi^n(K^n)}\right) \right.\right.  \nonumber\\
  &\qquad - \left.\left.\max\{0, -\nu\, (\log f_{S}(S_{K}^n))'\} (1-\tau) (h(S_{K}^n, -\nu) - h(S_{K}^n, \nu))\right|\right] \nonumber \\
  &\quad\leq (1-\tau)2M \, \E\left[\left|\sqrt{n} \left(1- 1 \wedge \frac{\pi^n(K^n + \nu)}{\pi^n(K^n)}\right) - \max\{0, -\nu\, (\log f_{S}(S_{K}^n))'\} \right|\right],
 \end{align}
 because there exists a positive constant $M$ such that $|h(\cdot,\nu)|\leq M$ for $\nu\in\{-1,1\}$ (recall that $h(\cdot,\nu)$ is continuous and vanishes at infinity for any value of $\nu$).

 We now consider four cases for $K^n$ and $\nu$:
 \begin{enumerate}
 \itemsep 0mm

 \item $\nu=+1$ and $\pi^n(K^n + \nu)/\pi^n(K^n)\geq1$ (we are going to the right on the real line and in this direction the PMF increases),

 \item $\nu=+1$ and $\pi^n(K^n + \nu)/\pi^n(K^n)<1$ (we are going to the right on the real line and in this direction the PMF decreases),

  \item $\nu=-1$ and $\pi^n(K^n + \nu)/\pi^n(K^n)\geq 1$ (we are going to the left on the real line and in this direction the PMF increases),

  \item $\nu=-1$ and $\pi^n(K^n + \nu)/\pi^n(K^n)< 1$ (we are going to the left on the real line and in this direction the PMF decreases).

 \end{enumerate}

 In Case 1,
 \[
  \left(1- 1 \wedge \frac{\pi^n(K^n + \nu)}{\pi^n(K^n)}\right) = 0,
 \]
 for all $n$ and $-\nu\, (\log f_{S}(S_{K}^n))'=-(\log f_{S}(S_{K}^n))'$ is negative in the limit because $f_{S}'(S_{K}^n)$ is positive in the limit. Therefore, $\max\{0, -\nu\, (\log f_{S}(S_{K}^n))'\}\longrightarrow0$. Using Lebesgue's dominated convergence theorem, we thus know that the expectation at the RHS in \eqref{eqn_diff_gens_NRJ_2} converges towards 0 when restricted to Case 1. We can prove that it converges towards 0 in Case 3 in the same way. In Case 2,
\[
 \sqrt{n} \left(1- 1 \wedge \frac{\pi^n(K^n + \nu)}{\pi^n(K^n)}\right) =- \frac{1}{\pi^n(K^n)} \, \frac{\pi^n(K^n + \nu) - \pi^n(K^n)}{1/\sqrt{n}}.
\]
By assumption, we know that this behaves asymptotically like $(\log f_{S}(S_{K}^n))'$. We also know that $-\nu\, (\log f_{S}(S_{K}^n))'=-(\log f_{S}(S_{K}^n))'$ is positive in the limit because $f_{S}'(S_{K}^n)$ is negative in the limit. Therefore, $- \max\{0, -\nu\, (\log f_{S}(S_{K}^n))'\}$ behaves like $(\log f_{S}(S_{K}^n))'$ in the limit. Using Lebes\-gue's dominated convergence theorem, we thus know that the expectation at the RHS in \eqref{eqn_diff_gens_NRJ_2} converges towards 0 when restricted to Case 2 (recall the assumed boundedness of the limiting quantity in the expectation). We can prove that it converges towards 0 in Case 4 in the same way.
\end{proof}

\subsection{Details about the multiple change-point example}\label{sec_changepoint_details}

It is assumed that the Poisson process has been observed on the time interval $[0, L]$, where $L>0$ is known. The starting point for each step is denoted by $s_{j,k}$, $j=0,\ldots,k$, to which we add the endpoint of the last step $s_{k+1,k}$, where these are subject to the constraint $0=:s_{0,k}<s_{1,k}<\ldots<s_{k+1,k}:=L$. The height of the $j$-th step is  denoted by $h_{j,k}$, $j=1,\ldots,k+1$. The log-likelihood of model $k$ is
\[
 \log\mathcal{L}(\mathbf{x}_k\mid k, \mathbf{t}):=\sum_{i=1}^n\log(\lambda_k(t_i\mid \mathbf{x}_k))-\int_0^L \lambda_k(t\mid \mathbf{x}_k) \, dt,
\]
where $\lambda_k(t\mid \mathbf{x}_k):=\sum_{j=0}^k h_{j+1,k}\ind_{[s_{j,k}, s_{j+1,k})}(t)$ for $t\in[0,L]$ and $\mathbf{x}_k:=(s_{1,k},\ldots,s_{k,k}, h_{1,k},\ldots, $ $h_{k+1,k})^T$, $\ind$ being the indicator function.

We use the same prior structure as \cite{green1995reversible}. The prior on $K$ is a Poisson distribution with parameter $\lambda>0$, but conditioned on $K\leq \text{K}_{\max}$. Given $K=k$, the starting points $s_{1,k},\ldots,s_{k,k}$ are \textit{a priori} distributed as the even-numbered order statistics from $2k+1$ points uniformly distributed on $[0,L]$, and the heights are independently and identically distributed as $\Gamma(\alpha,\beta)$, where $\alpha>0$ and $\beta>0$ are the shape and rate parameters, respectively. % In particular,
%\[
% \pi(s_{1,k},\ldots,s_{k,k}\mid k)={2k+1 \choose k} \, k! \, (k+1)! \,\frac{1}{L^k} \prod_{j=1}^{k+1}\frac{s_{j,k}-s_{j-1,k}}{L}, \quad s_{0,k}<s_{1,k}<\ldots<s_{k+1,k}.
%\]
In \cite{green1995reversible}, the hyperparameters are set to $\lambda:=3, \text{K}_{\max}:=30, \alpha := 1$, and $\beta := 200$.

As done in Section 5.1 of our paper, one may take advantage of the information at its disposal about the problem and model to design the sampler. \cite{green1995reversible} follows this approach. We design the RJ and the corresponding NRJ as this author. For parameter updates, we randomly choose to modify either one of the heights $h_{j,k}$ or one of the starting points $s_{j,k}$. We modify a starting point $s_{j,k}$ by proposing a new value uniformly between $s_{j-1,k}$ and $s_{j+1,k}$. We modify a height $h_{j,k}$ by proposing a new value $h_{j,k}'$ that is such that $\log(h_{j,k}'/h_{j,k})\sim\mathcal{U}[-1/2, 1/2]$. For model switches, we randomly choose to either add or withdraw a step. When we add a step, we first generate its starting point $s^*\sim \mathcal{U}[0, L]$. Deterministically, given $s^*$, we know which step will be splitted in two, in the sense that the proposal for the starting points is: $(s_{0,k},\ldots,s_{j^*,k}, s^*, s_{j^*+1,k},\ldots,s_{k+1,k})$, where $s_{0,k}<\ldots<s_{j^*,k}< s^*< s_{j^*+1,k}<\ldots<s_{k+1,k}$ (the step $(s_{j^*,k}, s_{j^*+1,k}]$ is splitted in two). We perturb as follows the height of this step $h_{j^*+1,k}$ to obtain proposals for the two heights $h_{j^*+1,k+1}'$ and $h_{j^*+2,k+1}'$ in the proposed model: generate $u_p\sim\mathcal{U}[0,1]$ which is such that $h_{j^*+2,k+1}'/h_{j^*+1,k+1}'=(1-u_p)/u_p$, and set the height proposals such that
\[
 (h_{j^*+1,k+1}')^{\frac{s^* - s_{j^*,k}}{s_{j^*+1,k} - s_{j^*,k}}} (h_{j^*+2,k+1}')^{\frac{s_{j^*+1,k} - s^*}{s_{j^*+1,k} - s_{j^*,k}}}=h_{j^*+1,k}.
\]
The height proposals are $(h_{1,k},\ldots, h_{j^*,k}, h_{j^*+1,k+1}', h_{j^*+2,k+1}', h_{j^*+2,k},\ldots, h_{k+1,k})$. When we withdraw a step, we proceed with the reverse move, which is deterministic after having generated $j^*\sim \mathcal{U}\{0,\ldots,k-1\}$ (starting from model $k$). See \cite{green1995reversible} for the acceptance probabilities and more details.

For implementing Algorithm 3 and the corresponding RJ, we proceed as in \cite{karagiannis2013annealed} for generating the paths. More precisely, when switching from model $k$ to model $k+1$, we use the same strategy as above to set the starting point of the path to
$(s_{0,k},\ldots,s_{j^*,k}, s^*, s_{j^*+1,k},\ldots,s_{k+1,k})$ and $(h_{1,k},\ldots, h_{j^*,k}, h_{j^*+1,k+1}', h_{j^*+2,k+1}', h_{j^*+2,k},\ldots, $ $h_{k+1,k})$, which are parameters in model $k+1$. We next update the parameters in model $k+1$ using blockwise MCMC sweeps. In a random order, we modify one of the heights $h_{j,k+1}$ and one of the starting points $s_{j,k+1}$ as when updating the parameters in Algorithm 1 (as explained above), and we update $j^*$. Note that when updating $h_{j,k+1}$ and $s_{j,k+1}$, we update the corresponding parameters in model $k$ as they are linked through deterministic functions. When updating $j^*$ given the rest, the parameters in model $k$ may be updated as we may change which step is splitted in two. The intermediate distributions are%\footnote{There is a typo in the definition of these functions in \cite{karagiannis2013annealed}.}
\small
\begin{align*}
  \rho_{k\mapsto k+1}^{(t)}(\mathbf{x}_k^{(t)},\mathbf{u}_{k\mapsto k'}^{(t)})&\propto \left[\pi(k,\mathbf{x}_k^{(t)})\, \frac{1}{L}  \, \frac{h_{j^*+1,k}^{(t)}}{(h_{j^*+2,k+1}^{(t)}+h_{j^*+1,k+1}^{(t)})^2}\right]^{1-t/T} \left[\pi(k+1,\mathbf{y}_{k+1}^{(t)})\, \frac{1}{k+1}\right]^{t/T}.
\end{align*}
\normalsize
See \cite{karagiannis2013annealed} for more details.

We finish this section by explaining how we established the number of iterations for the vanilla samplers. In Algorithm 2 when we switch models (see Step 2.(a)), we first generate the starting point of the path, which is done as in Algorithm 1, and next we generate the path using $T - 1$ MCMC steps. In each MCMC step, we try to modify one of the heights, one of the starting points and we sample $j^*$. Each MCMC step is thus essentially equivalent to 3 parameter updates, which is in turn essentially equivalent to 3 model switching attempts. To one model switching attempt in Algorithm 1 we thus essentially need to add $3 (T - 1)$ model switching attempts to obtain an equivalent cost.

On average, in one Algorithm 2 run, there are $I (1 - \tau)$ model switching attempts. Thus they correspond in Algorithm 1 to
\[
 I (1 - \tau) + I (1 - \tau) 3 (T - 1) \leq I (1 - \tau) 3 T
\]
model switching attempts. To identify the equivalence between Algorithm 3 and Algorithm 1, we need to multiply the number above by 1.5, as explained in Section 3.2 of our paper. Therefore, if in one run of Algorithm 3 there are on average  $I (1 - \tau)$ model switching attempts, then they correspond to $I (1 - \tau) 4.5 T$ model switching attempts in Algorithm 1. Algorithm 1 must thus be run for $I \tau + I (1 - \tau) 4.5 T$ iterations and $\tau$ in this algorithm must be set to
\[
 \tau := \frac{I \tau}{I \tau + I (1 - \tau) 4.5 T}.
\]

\end{document}